%% file: main.tex
\documentclass[a4paper,UKenglish]{lipics-v2019}
\usepackage{amsmath,amssymb,amsfonts,latexsym}
\usepackage{graphics}
\usepackage{xfrac}
\usepackage{hyperref}
\usepackage{xcolor}
\usepackage[quotation,notion,electronic]{knowledge}
\usepackage{import}
\input{knowledge.tex}

\usepackage{cleveref}

\usepackage{pdfrender}

\usepackage{tikz}
\usetikzlibrary{automata, positioning, arrows,arrows.meta,calc,decorations.markings,math,shapes.geometric}
\tikzset{
	>=stealth',
	-={stealth',ultra thick,scale=3} 
	node distance=1cm, 
	every state/.style={thick, fill=gray!10}, 
	initial text=$ $, 
}

\input{macros}

\title{Optimal Transformations of Muller Conditions 
}

\titlerunning{Optimal Transformations of Muller Conditions} 

\author{Antonio Casares}{LaBRI, Université de Bordeaux, France}{antonio.casares-santos@labri.fr}{https://orcid.org/0000-0002-6539-2020}{}
\author{Thomas Colcombet}{CNRS, IRIF, Université de Paris, France}{thomas.colcombet@irif.fr}{https://orcid.org/0000-0001-6529-6963}{}
\author{Nathanaël Fijalkow}{CNRS, LaBRI, Université de Bordeaux, France \and
 The Alan Turing Institute of Data Science, London, United Kingdom}{nathanael.fijalkow@labri.fr}{https://orcid.org/0000-0002-6576-4680}{}

\authorrunning{A.Casares, T. Colcombet and N. Fijalkow}



\category{} 

\relatedversion{} 
\relatedversion{This version is superseded by \url{https://arxiv.org/abs/2305.04323}.}

\acknowledgements{We want to thank Klara J. Meyer, Salomon Sickert and Ohad Drucker for pointing out some errors in previous versions of this paper.}

\nolinenumbers 

\hideLIPIcs  

\begin{document}
\maketitle

\begin{abstract}

	In this paper we are interested in automata over infinite words and infinite duration games, that we view as general transition systems. We study transformations of systems using a Muller condition into ones using a parity condition, extending Zielonka's construction. We introduce the alternating cycle decomposition transformation, and we prove a strong optimality result: for any given deterministic Muller automaton, the obtained parity automaton is minimal both in size and in number of priorities among those automata admitting a morphism into the original Muller automaton.
	
	 We give two applications. The first is an improvement in the process of determinisation of Büchi automata into parity automata by Piterman and Schewe. The second is to present characterisations on the possibility of relabelling automata with different acceptance conditions.
	

\end{abstract}

\paragraph*{}
This document contains hyperlinks.
\AP Each occurrence of a "notion@@example" is linked to its ""definition@@example"".
On an electronic device, the reader can click on words or symbols (or just hover over them on some PDF readers) to see their definition.

\vskip4mm

\section{Introduction}
\label{sec:intro}
\input{intro}

\section{Notations and definitions}
\label{sec:notations}
\input{notations}

\section{An optimal transformation of Muller into parity conditions}
\label{sec:zielonka-tree}
\input{zielonka-tree}

\section{An optimal transformation of Muller into parity transition systems}
\label{sec:acd}
\input{acd}

\section{Applications}
\label{sec:applications}
\input{applications}
%

\section{Conclusions}
\label{sec:conclusions}
\input{conclusions}

\bibliography{bibShort}




\end{document}

%% file: knowledge.tex
\definecolor{Dark Ruby Red}{HTML}{580507}
\definecolor{Dark Blue Sapphire}{HTML}{053641}
\definecolor{Dark Gamboge}{HTML}{be7c00}

\definecolor{Blue Sapphire}{HTML}{005f73} 
\definecolor{Blue Matt}{HTML}{094c7b} 
\definecolor{Blue Dark}{HTML}{1B4C6E} 
\definecolor{Golden}{HTML}{E27400}
\definecolor{Brown}{HTML}{ae740e}
\definecolor{Gamboge}{HTML}{ee9b00}
\definecolor{Ruby Red}{HTML}{9b2226}

\definecolor{Blue Marine}{HTML}{022687}
\definecolor{Very Dark Blue}{HTML}{0c1b44}
\definecolor{Dark Green}{HTML}{005715}
\definecolor{Very Dark Green}{HTML}{1e400c}
\definecolor{Dark Magenta}{HTML}{005715}

\IfKnowledgePaperModeTF{
}{
	\knowledgestyle{intro notion}{color={Dark Ruby Red}, emphasize}
	\knowledgestyle{notion}{color={Dark Blue Sapphire}}
	\hypersetup{
		colorlinks=true,
		breaklinks=true,
		linkcolor={}, 
		citecolor={}, 
		filecolor={Blue Marine}, 
		urlcolor={Blue Marine},
	}
}
\IfKnowledgeCompositionModeTF{
	\knowledgeconfigure{anchor point color={Ruby Red}, anchor point shape=corner}
	\knowledgestyle{intro unknown}{color={Dark Gamboge}, emphasize}
	\knowledgestyle{intro unknown cont}{color={Dark Gamboge}, emphasize}
	\knowledgestyle{kl unknown}{color={Dark Gamboge}}
	\knowledgestyle{kl unknown cont}{color={Dark Gamboge}}
}{
}

\knowledge{notion}
| notion@example
| definition@example

\knowledge{notion}
	| $\omega$-word
	| infinite word
	| infinite words

\knowledge{notion}
| \mlast

\knowledge{notion}
| \mfirst

\knowledge{notion}
	| automaton
	| Automata
	| automata
	| finite automaton
	| finite automata
	| transducer
	| non-deterministic automaton
	| non-deterministic automata
	| non-deterministic
	
\knowledge{notion}
	| deterministic
	| deterministic automata
	| deterministic automaton

\knowledge{notion}
	| transition system
	| transition systems
	| Transition systems

\knowledge{notion}
 | labelled transition systems
 | labelled transition system
 | labelled
 
\knowledge{notion}
	| $\oo$-regular language

\knowledge{notion}
	| Muller condition
	| Muller acceptance condition
	| Muller acceptance conditions
	| Muller conditions
	| Muller acceptance condition
	| Muller

\knowledge{notion}
	| Parity condition
	| parity condition
	| parity acceptance condition
	| parity acceptance conditions
	| $[\mu,\eta]$-parity condition
	| $[\mu,\eta]$-condition
	| $[1,3]$-parity condition
	| Rabin chain conditions
	| parity
	| $[1,\eta ]$-parity condition
	| $[0,\eta -1]$-parity condition
	| parity conditions
	| $[1,2r+1]$-parity condition
	| $[0,2s]$-parity condition
	| $[1,k+1]$-parity conditions
	| $[1,k+1]$-parity
	| $[\mu,\eta]$-parity
	| $[1,\eta]$-parity
	| $[0,\eta-1]$-parity

\knowledge{notion}
	| Muller transition system
	| Muller transition systems
	| $\R $-transition systems
	| parity transition system
	| Muller automaton
	| Muller automata
	| parity automaton
	| parity automata
	| Büchi automata
	| deterministic parity automaton
	| Rabin automaton
	| Büchi automaton
	| Muller game
	| parity game
\knowledge{notion}
	| Zielonka tree
	| Zielonka trees
	| split trees
	| T_\F
	| T_{\F _2}
	| \TF 
	
\knowledge{notion}
	| Zielonka tree-automaton
	| Zielonka tree automaton
	| Zielonka tree automata
	| \Z _\F 
	| \ZF 
	
\knowledge{notion}
	| loop
	| loops
	| \mathpzc {Loop}

\knowledge{notion}
| alphabet

\knowledge{notion}
| word

\knowledge{notion}
| empty word

\knowledge{notion}
| \mathit {Inf}(u)
| \mathit {Inf}
| \minf (u)
| \minf

\knowledge{notion}
| \mathit {Occ}(u)
| \mathit {Occ}
| \mathit{Occ}

\knowledge{notion}
| output
| \moutput

\knowledge{notion}
| accepted

\knowledge{notion}
| run
| runs
| \mathpzc {Run}

\knowledge{notion}
| run produced by
| run over $u$
| run over
| runs over

\knowledge{notion}
| acceptance condition
| acceptance conditions
| accepting condition

\knowledge{notion}
| composition
| compositions
| \A \lhd  \T _1
| \B  \lhd \T_1

\knowledge{notion}
| Büchi condition
| Büchi conditions
| Büchi

\knowledge{notion}
| co-Büchi condition
| co-Büchi

\knowledge{notion}
| Rabin condition
| $R_k$-condition
| Rabin
| Rabin pair
| pairs of Rabin

\knowledge{notion}
| Streett condition
| $S_k$-condition
| Streett

\knowledge{notion}
| priorities
| priority

\knowledge{notion}
| $M_k$-condition

\knowledge{notion}
| ordered partition

\knowledge{notion}
| component of the ordered partition

\knowledge{notion}
| equivalent

\knowledge{notion}
| equivalent over
| \simeq
| equivalent to
| equivalent condition over

\knowledge{notion}
| prefix
| prefix relation
| \sqsubseteq
| \sqsupseteq

\knowledge{notion}
| \trianglelefteq

\knowledge{notion}
| transition system graph
| underlying graph

\knowledge{notion}
| coloured transitions

\knowledge{notion}
| \mathit {States}

\knowledge{notion}
| \mIn 

\knowledge{notion}
| \mout 

\knowledge{notion}
| tree
| trees

\knowledge{notion}
| nodes
| node

\knowledge{notion}
| root

\knowledge{notion}
| leaf

\knowledge{notion}
| branch

\knowledge{notion}
| \mathit{Branch}

\knowledge{notion}
| \mathit{Subtree}
| \mathit{Subtree}_T

\knowledge{notion}
| child
| \mathit {Children}
| children 

\knowledge{notion}
| parent
\knowledge{notion}
| ancestor
| ancestors
\knowledge{notion}
| descendant
| descendants
| strict descendants
\knowledge{notion}
| depth
| \mathit{Depth}
| deepest
\knowledge{notion}
| height of a tree
| height
| \mathit{Height}
\knowledge{notion}
| height of the node
\knowledge{notion}
| labelled tree
| labelled trees

\knowledge{notion}
| locally bijective morphisms
| Locally surjective
| Locally injective
| Locally bijective
| locally surjective
| locally injective
| locally bijective morphism
| locally bijective
| local bijectivity
| local surjectivity
| locally surjective morphism
| Locally bijective morphisms

\knowledge{notion}
| Alternating cycle decomposition
| alternating cycle decomposition
| \ACD 
| \mathcal{ACD}
| alternating cycle decompositions
| alternating cycle de\-com\-po\-si\-tion

\knowledge{notion}
 | morphism of transition systems
 | morphism
 | morphisms
 
 \knowledge{notion}
 | morphism of labelled transition systems

\knowledge{notion}
| guides
| guide

\knowledge{notion}
| proper trees
| proper tree

\knowledge{notion}
| index of an edge
| index

\knowledge{odd}[|even]{scope=Zielonka,notion}
\knowledge{odd}[even|ambiguous]{scope=ACD,notion}

\knowledge{priority}[p_j|p_i|level|priorities]{scope=node,notion}

\knowledge{notion}
	| Automata as transition systems
	| Example~\ref{Example_AutomataAsTransSyst}
	
 \knowledge{notion}
 | |u|_a
 | |u|_1

 \knowledge{notion}
 | $\Sigma $-complete
 | $X$-complete

 \knowledge{notion}
 | language accepted
 | \mathcal {L}(\A )
 | \L (\A )
 | \L (\A ')
 | \L (\B )
 | recognises
 | language accepted by an automaton
 | recognising
 | accepted by

 \knowledge{notion}
 | \A (u)
 | \A (w)
 
 \knowledge{notion}
 | generalised weak condition
 | generalised weak
 | \mathit {Weak}_k
 | Weak
 | \mathit {Weak}
 | \mathit {Weak}_1
 | weak
 | weak condition
 | Weak condition

\knowledge{notion}
| \mathit {Below}

 \knowledge{notion} 
 | directed graph
 | path
 \knowledge{notion} 
 | strongly connected component
 | strongly connected components
 | strongly 

 \knowledge{notion} 
 | subgraph
 \knowledge{notion} 
 | strongly connected
 | strongly connected subgraph
 | strongly connected graph
 
\knowledge{notion}
 | |\T |
 | |\P _\T |
 | |\ZF |
 | |\P |

\knowledge{notion}
 | accessible
 | reachable
 
\knowledge{notion}
| accessible part

\knowledge{notion}
 | game
 | games
 | games on graphs

\knowledge{notion}
| p_Z(\tt )
| p_Z
\knowledge{notion}
 | X-strongly connected component
 |$X$-SCC
 |$A$-SCC
 |$B$-SCC
 |$A_i$-SCC
\knowledge{notion}
 	| $l$-SCC
 	| $l_1 $-SCC
 	| $l_2$-SCC
 	| $l_1,l_2$-SCC
 	| $\nu _i(\ss _j)$-
 	| $\nu_i(\tt)$-SCC
 	| $l'$-SCC
 	| $\nu _i(\ss _j)$-SCC

 \knowledge{notion}
 | subtree
 
\knowledge{notion}
 | \mathit{Nextbranch}
 
 \knowledge{notion}
 | \mathit{Nextchild}
 
\knowledge{notion}
 | \mathit{Nextbranch}_{t_q}
 | \mathit{Nextbranch}_{t_{q'}}

\knowledge{notion}
 | \mathit{Nextchild}_{t_q}

\knowledge{notion}
 | ACD-parity transition system
 | \P _\T 
 | ACD-transition system
 | ACD-parity automaton
 | ACD-transformation
 | \P _{\mathcal {ACD}(\T )}
 | \P _{\mathcal {ACD}(\A )}
 | alternating cycle decomposition transformation
 | \P _{\mathcal {ACD}(\R _B)}
 
 \knowledge{notion}
 |Update of node labels
 |Step 1
 |step 1
\knowledge{notion}
 |\litem{Spawning new children}
 |Step 2
 |step 2
\knowledge{notion}
|\litem{Horizontal pruning}
|Step 3
|step 3
\knowledge{notion}
|\litem{Vertical pruning}
|Step 4
|step 4
\knowledge{notion}
|\litem{Removing nodes with empty label and reordering}
|Step 5
|step 5

\knowledge{notion}
| siblings
| sibling
\knowledge{notion}
| older
\knowledge{notion}
| circle
| round node
| square node
| round nodes
| square
| round

\knowledge{notion}
| support	
| \mathit{Supp}

\knowledge{\mathit {Supp}}[]{scope=Zielonka,notion}

\knowledge{notion}
| Wagner Hierarchy
| Wagner hierarchy

\knowledge{notion}
| history tree
| history trees

\knowledge{notion}
| \Theta _n
| \Theta 
| \Theta _{|\B |}

\knowledge{notion}
| T_{\macc }

\knowledge{notion}
| stable
\knowledge{notion}
| rejecting
\knowledge{notion}
| accepting nodes
| accepting
\knowledge{notion}
| ordered history tree
| ordered history trees

\knowledge{notion}
| full Büchi automata

\knowledge{notion}
| \R _\B 

\knowledge{notion}
| \P _B
| \P _\B 
\knowledge{notion}
| Rabin shape
\knowledge{notion}
| Streett shape
\knowledge{notion}
| Parity shape
| parity shape

\knowledge{notion}
| Rabin ACD
\knowledge{notion}
| Streett ACD
\knowledge{notion}
| Parity ACD
| parity ACD
\knowledge{notion}
| $[1,\eta ]$-parity ACD
| $[1,k+1]$-parity ACD
| $[0,k]$-parity ACD
\knowledge{notion}
| Büchi ACD
\knowledge{notion}
| co-Büchi ACD
\knowledge{notion}
| $\mathit {Weak}_k$ ACD

\knowledge{notion}
| deterministic pre-automaton

\knowledge{notion}
| deterministic parity hierarchy
| parity hierarchy

\knowledge{notion}
| $\oo $-regular conditions

\knowledge{notion}
| \Sigma^\infty
| E^\infty

\knowledge{notion}
| play

\knowledge{notion}
| strategy

\knowledge{notion}
| wins
| win
| winner

\knowledge{notion}
| winning region
| winning regions
| \W_P(\G)
| \W_P

\knowledge{notion}
|strategy

\knowledge{notion}
| Good-for-games
| good-for-games
| GFG
| history-deterministic

\knowledge{notion}
| unambiguous

\knowledge{notion}
| strongly unambiguous
| Strongly unambiguous

\knowledge{notion}
| subtree associated to the state $q$
| subtree associated to
| Subtree associated to
| t_{q'}
| t_q
| t_{q_0}
| t_{q_j}

%% file: macros.tex
\let\ab\allowbreak

\definecolor{Green2}{HTML}{3EA514}
\definecolor{Red2}{HTML}{FF0400}
\definecolor{Orange2}{HTML}{FE9507}
\definecolor{Violet2}{HTML}{CE1ff9}

\DeclareMathAlphabet{\mathpzc}{OT1}{pzc}{m}{it}

\newcommand{\NN}{\mathbb{N} }

\newcommand{\F}{\mathcal{F}}
\newcommand{\G}{\mathcal{G}}

\renewcommand{\L}{\mathcal{L}}

\newcommand{\C}{\mathcal{C}}
\newcommand{\A}{\mathcal{A}}
\newcommand{\B}{\mathcal{B}}
\newcommand{\Z}{\mathcal{Z}}
\newcommand{\R}{\mathcal{R}}
\newcommand{\T}{\mathcal{T}}
\newcommand{\W}{\mathcal{W}}

\renewcommand{\P}{\mathcal{P}}

\newcommand{\dd}{\delta}

\renewcommand{\ss}{\sigma}

\newcommand{\rr}{\varrho}

\renewcommand{\aa}{\alpha}

\renewcommand{\tt}{\tau}
\newcommand{\bb}{\beta}

\renewcommand{\SS}{\Sigma}
\newcommand{\GG}{\Gamma}
\newcommand{\oo}{\omega}
\newcommand{\pp}{\varphi}

\newcommand{\ZF}{\mathcal{\Z_{\F}}}

\newcommand{\TF}{\mathcal{T}_{\F}}

\newcommand{\rest}{\upharpoonright}

\newcommand{\mstate}{\kl{\mathit{States}}}
\newcommand{\msupp}{\kl{\mathit{Supp}}}

\newcommand{\msource}{\mathit{Source}}
\newcommand{\mtarget}{\mathit{Target}}
\newcommand{\moutput}{\mathit{Output}}
\newcommand{\mout}{\mathit{Out}}
\newcommand{\mIn}{\mathit{In}}
\newcommand{\minf}{\mathit{Inf}}
\newcommand{\mocc}{\kl{\mathit{Occ}}}
\newcommand{\macc}{\mathit{Acc}}
\newcommand{\mfirst}{\mathit{First}}
\newcommand{\mlast}{\mathit{Last}}
\newcommand{\mdepth}{\kl{\mathit{Depth}}}
\newcommand{\mheight}{\kl{\mathit{Height}}}

\newcommand{\mbranch}{\kl{\mathit{Branch}}}
\newcommand{\msubtree}{\kl{\mathit{Subtree}}}
\newcommand{\mnextb}{\kl{\mathit{Nextbranch}}}
\newcommand{\mnextc}{\kl{\mathit{Nextchild}}}

\newcommand{\mrun}{\kl{\mathpzc{Run}}}
\newcommand{\mloop}{\kl{\mathpzc{Loop}}}

\newcommand{\prefix}{\kl{\sqsubseteq}}

\newcommand{\ACD}{\kl{\mathcal{ACD}}}

%% file: intro.tex
Automata over infinite words were first introduced in the 60s by Büchi~\cite{buchi1960decision}, in his proof of the decidability of the monadic second order theory of the natural numbers with successor $(\NN,0,\mathit{succ})$. 
Contrary to automata over finite words, there is not a unique natural definition for acceptance in the infinite setting. 
The condition used by Büchi (called "Büchi condition"), accepts those runs that visit infinitely often a final state. The "Muller condition", introduced in~\cite{muller1963infSequences}, specifies a family $\F$ of subsets of states and accepts a run if the states visited infinitely often form a set in that family. Other acceptance conditions have been defined, and in particular the "parity condition", introduced by Mostowski in~\cite{mostowski1984RegularEF}, is of notable importance. It has the same expressive power as the Muller condition and it is specially well-behaved and adequate for algorithmic manipulation. For example, it allows for easy complementation of automata, it admits memoryless strategies for games, and parity game solvers have been proved very efficient. The "parity condition" assigns to each edge of an automaton a natural number, called a priority, and a run is accepting if the smallest priority visited infinitely often is even. Therefore, another important parameter dealing with parity conditions is the number of different priorities that it uses. An automaton using an acceptance condition of type $\C$ is called a $\C$ automaton.

As it is to be expected, different acceptance conditions have different expressive power. It is not difficult to see that non-deterministic Büchi automata have strictly more expressive power than deterministic ones. However, McNaughton Theorem~\cite{mcNaughton1966Testing} states that Büchi automata can be transformed into deterministic Muller automata accepting the same language. Non-deterministic Büchi and Muller automata are also equivalent to regular expressions and MSO logic over infinite words. Deterministic parity automata stand out as the simplest type of deterministic automata being equivalent to all these models~\cite{mostowski1984RegularEF}.
	
In this work we study methods which transform automata and games that use Muller acceptance conditions into others using parity acceptance conditions. We present the con\-struc\-tions using the formalism of "transition systems" in order to obtain the most general results. These constructions can be immediately applied to transform deterministic or non-deterministic, as well as games. We work with \emph{transition-labelled} systems (instead of state-labelled) for technical convenience. 

The standard way to transform a Muller transition system into a parity transition system is to build a deterministic parity automaton that "recognises" the Muller condition, and then take the "composition" of the Muller transition system and the automaton, which is a transition system using a parity condition. We can find the first example of one such automaton implicitly in the work of Gurevich and Harrington~\cite{gurevich1982trees}, called a \emph{later appearance record (LAR)}. The ideas of Gurevich and Harrington have recently been refined in order to find smaller automata~\cite{loding1999Optimal,kretinski2017IAR}. On the other hand, in his work on the memory requirements of Muller games~\cite{zielonka1998infinite}, Zielonka presents the notion of the \emph{split tree} associated to a "Muller condition" (later called \emph{"Zielonka tree"}~\cite{djw1997memory,horn2008random}). This construction yields another "parity automaton" (that we call "Zielonka tree automaton") recognising a Muller condition, and we make this construction explicit in Section~\ref{sec:zielonka-tree}. We prove in Section~\ref{sec:zielonka-tree} that the Zielonka tree automaton is a minimal parity automaton recognising a Muller condition and it uses a parity condition with the minimal number of priorities. 
	
	However, the "composition" of a Muller transition system $\T$ and the "Zielonka tree automaton" does not result in general in the ``best'' parity transition system simulating $\T$. In order to optimise this transformation, we present in Section~\ref{sec:acd} a data structure, the \emph{"alternating cycle decomposition" (ACD)} of $\T$, that incorporates the information of how the transition system $\T$ makes use locally of the Muller condition. The alternating cycle decomposition is obtained applying the Zielonka tree construction while taking into account the structure of $\T$ by considering the loops that are alternatively accepting and rejecting.
	 The idea of considering the alternating chains of loops of an automaton is already present in the work of Wagner~\cite{wagner1979omega}. In Section~\ref{sec:acd} we prove our main result: the transformation given by the "alternating cycle decomposition" is optimal and it uses a "parity condition" with the optimal number of priorities.
	
	In this report we are concerned with \emph{state complexity}, this is, the efficiency of a construction is measured based on the number of states of the resulting transition system. However, we emphasize that the word \textit{optimal} is used in a strong sense. When we say that a transformation is optimal we do not only mean that there exists a family of transition systems for which we need at least the number of states given by this transformation,
	 but than \emph{in all cases} we obtain a minimal transition system with the desired properties. For instance, it is sometimes stated that the \textit{LAR} automaton of~\cite{gurevich1982trees} is optimal (see~\cite{loding1999Optimal}), however, it only is in the worst case, as in most other cases the "Zielonka tree automaton" has strictly smaller size.

	At the end of this report we study the applicability of the proposed transformation to one of the main concerns in many applications of automata over infinite words (as for instance the synthesis for LTL formulas): the determinisation of "Büchi automata". The first efficient determinisation procedure was proposed by Safra~\cite{safra1988onthecomplexity}, and since then many other constructions have been proposed. In~\cite{piterman2006fromNDBuchi}, Piterman proposed a modification on Safra's construction that improves the complexity and directly produces a "parity automaton". In~\cite{schewe2009tighter}, Schewe revisits Piterman's construction differentiating two steps: a first one producing a deterministic "Rabin" automaton, and a second one producing a "parity" automaton. In~\cite{colcombetz2009tight}, Colcombet and Zdanowski found a tight worst-case lower bound for the first step, and Schewe and Varghese found a tight (up to a constant) worst-case lower bound for the second step in~\cite{varghese2014PhD,Varghese12}.
	The "alternating cycle decomposition" presented in this report provides a new procedure to transform the Rabin automaton to a parity one that improves the con\-struc\-tion proposed in~\cite{schewe2009tighter} (Theorem~\ref{Th_ACD-Better than Schewe}).
	
	Finally, the "alternating cycle decomposition" clarifies the relation between the structure of a transition system and the different acceptance conditions that can be used to relabel it. In~\cite{zielonka1998infinite}, it is proven that a "Muller" condition $\F$ is equivalent to a "Rabin" (resp. "Streett") condition if and only if the family $\F$ is closed under intersection (resp. union), and it is equivalent to a "parity" condition if and only if it is closed under both unions and intersections. We extend this characterisations to transition systems. In Propositions~\ref{Prop_RelabellinRabin} and~\ref{Prop_RelabellinStreett} we prove that a transition system can be labelled with a Rabin (resp. Streett) condition if and only if the union of two rejecting (resp. accepting) "loops" is rejecting (resp. accepting). In Proposition~\ref{Prop_RelabellinParity} we prove that it can be relabelled with a parity condition if and only if the two previous conditions hold.  
	As corollaries, we obtain some already known results, first proven in~\cite{kupferman2010ParityizingRA}: a deterministic automaton can be labelled with a parity condition if and only if it can be labelled with both Rabin and Streett conditions; and it can be labelled with a "Weak" condition if and only if it can be labelled with both "Büchi" and "co-Büchi" conditions.

\paragraph*{Contributions} 
In this report, we establish four results:
\begin{description}
\item[Optimal parity automata for Muller conditions.] We present how to use the "Zielonka tree" of a "Muller condition" in order to build a parity automaton "recognising" this condition. We prove that this automaton is optimal for every "Muller condition", both in terms of number of priorities used (Proposition~\ref{Prop_OptimalPrioritiesZielonka}) and of size (Theorem~\ref{Th_OptimalityZielonkaTree}).
This result is new here, but this construction can be considered as already known by the community.

\item[Optimal transformation of Muller into parity transition systems.]
We provide a con\-struc\-tion translating Muller transition systems into parity transition systems that can be seen as a generalisation of the above one. We introduce the "alternating cycle decompositions" which are generalisations of "Zielonka trees" to transition systems, and derive from them our
"alternating cycle decomposition transformation" (ACD-transformation).

We state and prove an optimality result using the concept of "locally bijective morphisms": when the ACD-transformation is applied to a Muller transition system $\T$, it outputs a parity transition system $\T'$ such that there exists a "locally bijective morphism" from $\T'$ to $\T$. The optimality result states that $\T'$ has the minimum number of states such that this property holds (Theorem~\ref{Th_OptimalityACDTransformation}), and that its parity condition uses an optimal number of priorities (Proposition~\ref{Prop_OptimalityACD_Priorities}). 

 \item[Improvement on Piterman-Schewe's determinisation of Büchi automata.]
 Piterman~\cite{piterman2006fromNDBuchi} and Schewe~\cite{schewe2009tighter} have described an efficient translation from non-deterministic Büchi automata to deterministic parity automata (both are variations of the famous construction of Safra~\cite{safra1988onthecomplexity}).
 Schewe describes this construction as first building a Rabin automaton $"\R_\B"$, followed by an ad-hoc transformation of this automaton for producing a parity automaton $"\P_\B"$.
This second step induces in fact a "locally bijective morphism". This implies that we obtain a smaller parity automaton by applying the "ACD-transformation" to $\R_\B$. We also provide an example where this automaton is indeed strictly smaller and uses less priorities (Example~\ref{Example_Buchi_betterACD}).
 
\item[New proofs for results of automata relabelling.]
 Finally, the "alternating cycle decomposition" enables us to prove characterisations of systems that can be labelled with Rabin, Streett and parity conditions (Propositions~\ref{Prop_RelabellinRabin},~\ref{Prop_RelabellinStreett} and~\ref{Prop_RelabellinParity}). As a consequence, we obtain simple proofs of two existing results: if a deterministic automaton can be labelled by a "Rabin condition" and by a "Streett condition" while accepting the same language, then it can be labelled by a "parity condition" while accepting the same language~\cite{kupferman2010ParityizingRA},
 and if a deterministic automaton can be labelled by a Büchi condition and by a co-Büchi condition while accepting the same language, then it can be labelled by a "weak condition" while accepting the same language.
\end{description}

\paragraph*{Organisation of this report}
In Section~\ref{sec:notations} we present the definitions and notations that we will use throughout the report.

In Section~\ref{sec:zielonka-tree} we define the "Zielonka tree" and the "Zielonka tree automaton" and we prove the optimality of the latter. In Section~\ref{Section_ZielonkaTree-SomeTypes} we present some examples of "Zielonka trees" for special acceptance conditions, that will be useful in Section~\ref{Section_4.2.StructuralParity}. Most constructions and proofs of this section can be regarded as special cases of those from Section~\ref{sec:acd}. However, we find instructive to include them separately since this is the opportunity to describe all the core ideas that will appear in Section~\ref{sec:acd} in a simpler setting.

We begin Section~\ref{sec:acd} by defining "locally bijective morphisms". In Section~\ref{Subsection_ACD} we present the main contribution of this work: the "alternating cycle decomposition", and we prove its optimality in Section~\ref{Section_OptimalityACD}.

 Section~\ref{sec:applications} is divided in two very different parts. In Section~\ref{Section_DeterminisationBuchi} we show how the "ACD-transformation" could provide a smaller deterministic parity automaton than the con\-struc\-tions of~\cite{piterman2006fromNDBuchi} and~\cite{schewe2009tighter}. In Section~\ref{Section_4.2.StructuralParity} we analyse the information given by the "alternating cycle decomposition" and we provide two original proofs concerning the possibility of labelling automata with different acceptance conditions.

We have included detailed examples all throughout the report. We hope that these will help the reader to better understand the sometimes intricate formalism.

\paragraph*{Extended version (2023)}
The current paper has been superseded by the extended version ``From Muller to Parity and Rabin Automata: Optimal Transformations Preserving (History-)Determinism''~\cite{CCFL23FromMtoP}. Some of the main modifications and additions that can be found in that new version are:
\begin{itemize}
	\item Generalisation of the results to "history-deterministic" (also called "good-for-games") and "Rabin" automata.
	\item A conceptually simpler proof of Theorem~\ref{Th_OptimalityACDTransformation}.
	\item The introduction of a normal form for "parity" automata.
	\item An algorithm for the minimisation of "parity" automata "recognising" "Muller conditions".
\end{itemize}

\newpage

%% file: notations.tex
In this section we introduce 
 standard notions that will be used throughout the report.
We begin with some basic notations in Section~\ref{subsection:basic-notations}.

	\subsection{Basic notations}
	\label{subsection:basic-notations}

	 We let $\P(A)$ denote the power set of a set $A$ and $|A|$ its cardinality. 
	 The symbol $\oo$ denotes the ordered set of non-negative integers. For $i,j\in \oo$, $i\leq j$, $[i,j]$ stands for $\{i,i+1, \dots, j-1,j \}$.
	 
	  For a set $\Sigma$, a ""word"" over $\Sigma$ is a sequence of elements 
	   from $\Sigma$. The length of a word $u$ is $|u|$. An ""$\omega$-word"" (or simply an "infinite word") is a word of length $\oo$.
	 The sets of finite and infinite words over $\Sigma$ will be written $\Sigma^*$ and $\Sigma^{\oo}$ respectively. We let $\AP""\Sigma^\infty"" = \Sigma^* \cup \Sigma^\oo$. For a word $u\in \Sigma^\infty$ we write $u_i$ to represent the $i$-th letter of $u$.
	 We let $\varepsilon$ denote the ""empty word"". For $u\in \Sigma^*$ and $v\in "\Sigma^\infty"$, the concatenation of these words is written $u\cdot v$, or simply $uv$. If $u=v\cdot w$ for $v\in \Sigma^*, u,w\in "\Sigma^\infty"$, we say that $v$ is a  \AP""prefix"" of $u$ and we write $v "\sqsubseteq" u$ (it induces a partial order on $\Sigma^*$). 
	 
	 For a finite "word" $u\in \Sigma^*$ we write $\AP""\mfirst""(u)=u_0$ and $\AP""\mlast""(u)=u_{|u|-1}$.
	 For a word $u\in "\Sigma^\infty"$, we let 
	 $ \AP""\mathit{Inf}(u)""=\{ a\in \Sigma \; : \; u_i=a \text{ for infinitely many } i\in \oo\}$ and $ \AP\intro*\mocc=\{ a\in \Sigma \; : \; \exists i \in \oo \text{ such that }u_i=a  \}$.
	 
	Given a map $\alpha:A \rightarrow B$, we will extend $\alpha$ to words component-wise, i.e., $\aa: A^\infty \rightarrow B^\infty$ will be defined as $\aa(a_0a_1a_2\dots)=\aa(a_0)\aa(a_1)\aa(a_2)\dots$. We will use this convention throughout the paper without explicitly mention it.
	 
	 A \AP""directed graph"" is a tuple $(V,E,\msource,\mtarget)$ where $V$ is a set of vertices, $E$ a set of edges and $\msource,\mtarget:E\rightarrow V$ are maps indicating the source and target for each edge. A \emph{"path"} from $v_1\in V$ to $v_2\in V$ is a word $\rr\in E^*$ such that $\msource("\mfirst"(\rr))=v_1$, $\mtarget("\mlast"(\rr))=v_2$ and $\msource(\rr_{i})=\mtarget(\rr_{i-1})$ for $1\leq i<|\rr|$. A graph is \AP""strongly connected"" if there is a path connecting each pair of vertices. A \AP""subgraph"" of $(V,E,\msource,\mtarget)$ is a graph $(V',E',\msource',\mtarget')$ such that $V'\subseteq V$, $E'\subseteq E$ and $\msource'$ and $\mtarget'$ are the restriction to $E'$ of $\msource$ and $\mtarget$, respectively. A \AP""strongly connected component"" is a maximal "strongly connected" subgraph.

	\subsection{Automata over infinite words}
	
	A  \AP""non-deterministic automaton"" (which we will simply call an \emph{automaton}) is a tuple $\A=(Q,\Sigma, I_0, \Gamma, \delta, \macc)$ where:
	\begin{itemize}
		\item $Q$ is a set of states.
		\item $\Sigma$ is an input alphabet.
		\item $I_0\subseteq Q$ is a non-empty set of initial states.
		\item $\Gamma$ is an output alphabet.
		\item $\delta: Q\times \Sigma \rightarrow \P(Q \times \Gamma)$ is a transition function.
		\item $\macc\subseteq \Gamma^\oo$ is an acceptance condition.
	\end{itemize}
	
	If for every $q\in Q$, $a\in \Sigma$, $\dd(q,a)\neq \emptyset$ we say that the automaton is \AP""$\Sigma$-complete"". We can always suppose that an automaton is $\Sigma$-complete by adding a ``sink node'' $s$ to $Q$ that receives the not previously defined transitions.
	
	If $I_0$ is a singleton and for every $q\in Q$, $a\in \Sigma$, $\dd(q,a)$ is a singleton, we say that $\A$ is a  \AP""deterministic automaton"" (in particular a deterministic automaton is "$\Sigma$-complete").
	%
	In this case we will split the transition function into $\dd:Q\times \Sigma \rightarrow Q$ and $\gamma: Q\times \Sigma \rightarrow \Gamma$.
	In some cases we will omit the output alphabet $\Gamma$. If so, we will implicitly take as the output alphabet the whole set of transitions, $\Gamma=\{ (q,a,\delta(q,a)) \; : \; q\in Q, \; a\in \Sigma \}$. (See Figure~\ref{Fig_AutomataForL} for examples).

	We extend the definition of $\delta$ to finite words $\delta:Q\times \Sigma^* \rightarrow Q$ inductively: 
	\begin{itemize}
		\item $\dd(q,\varepsilon)=q$, for $q\in Q$
		\item $\delta(q,wa)=\delta(\delta(q,w),a)$, for $w\in \Sigma^*$ and $a\in \Sigma$
	\end{itemize}


	Given an "automaton" $\A$ and a word $u\in \Sigma^\oo$, a  \AP""run over $u$"" in $\A$ is a sequence
	\[ \rr=(q_0,u_0,b_0,q_1)(q_1,u_1,b_1,q_2) \dots  \quad q_i\in Q,\,  b_i\in \Gamma \; \text{for every } i\in \oo\]
	such that $q_0\in I_0$ and $(q_{i+1},b_{i})\in \dd(q_i,u_i)$ for all $i\in \oo$. The  \AP""output"" of the run $\rr$ is the word $\moutput_\A(\rr)=b_0b_1b_2\dots \in \Gamma^\oo$. The word $u$ is  \AP""accepted"" by $\A$ if it exists a $run$ $\rr$ over $u$ such that $\moutput_\A(\rr)\in \macc$.	The  \AP""language accepted"" (or recognised) by an automaton $\A$ is the set 
	\[ "\mathcal{L}(\A)":= \{ u\in \Sigma^\oo \; : \; u \text{ is accepted by } \A \}.\]
	We remark that if $\A$ is "deterministic" then there is a single "run over" $u$ for each $u\in \Sigma^\oo$. We let $ \AP""\A(u)""$ denote the "output" of this run.
	\begin{remark*}
		
		We have defined transition-labelled automata (the acceptance condition is defined over transitions instead of over states). 
		 Transition-labelled automata are easily transformed into state-labelled automata, and vice versa. 
	\end{remark*}
	 
	\subsection{Transition systems}
	
	A  \AP""transition system graph"" $\T_G=(V,E,\msource,\mtarget,I_0)$ is a "directed graph" with a non-empty set of initial vertices $I_0\subseteq V$. We will also refer to vertices and edges as \emph{states} and \emph{transitions}, respectively.

	We will suppose in this work that every vertex of a transition system graph has at least one outgoing edge. \\
	
	A \AP""transition system"" $\T$ is obtained from a "transition system graph" $\T_G$ by adding:
	\begin{itemize}
		\item A function $\gamma: E \rightarrow \GG$. The set $\GG$ will be called a \emph{set of colours} and the function $\gamma$ a \emph{colouring function}.
		\item An \AP""acceptance condition"" $\macc \subseteq \GG^\oo$.
	\end{itemize} 
	We will usually take $\GG=E$ and $\gamma$ the identity function. In that case we will omit the set of colours in the description of $\T$. 
	

	A  \AP""run"" from $q\in V$ on a "transition system graph" $\T$ is a sequence of edges $\rr=e_0e_1\dots \in "E^\infty"$ such that $\msource(e_0)=q$ and $\mathit{Target}(e_{i-1})=\mathit{Source}(e_{i})$ for all $1\leq i<|\rr|$. We emphasize the fact that runs can be finite or infinite.
	
	For $A\subseteq V$ we let $"\mathpzc{Run}"_{\T,A}$ denote the set of runs on $\T$ starting from some $q\in A$ (we omit brackets if $A=\{q\}$), and $\mathpzc{Run}_{\T}=\mathpzc{Run}_{\T,I_0}$ the set of runs starting from some initial vertex.
	
	A "run" $\rr \in \mathpzc{Run}_{\T}$ is ""accepting"" if $\gamma(\rr)\in \macc$, and rejecting otherwise.\\
	
	
	We say that a vertex $v\in V$ is  \AP""accessible"" (or \emph{reachable}) if there exists a finite run $\rr \in \mathpzc{Run}_{\T}$ such that $\mtarget("\mlast"(\rr))=v$. A set of vertices $B\subseteq V$ is accessible if every vertex $v\in B$ is accessible. The \AP""accessible part"" of a "transition system" is the set of accessible vertices.

	Given a transition system $\T$ we let $ \AP""|\T|""$ denote $|V|$ for $V$ its set of vertices.
		For a subset of vertices $A\subseteq V$ we write:
		\begin{itemize}
			\item $""\mIn""(A)=\{ e\in E \; : \; \mtarget(e)\in A \}$,
			\item $""\mout""(A)=\{ e\in E \; : \; \msource(e)\in A \}$.
		\end{itemize}

	We might want to add more information to a transition system. For example we could associate vertices to different players in order to obtain a game, or add an input alphabet to obtain an automaton. A  \AP""labelled transition system"" is a "transition system" $\T$ with labelling functions $l_V: V \rightarrow L_V$, $l_E: E \rightarrow L_E$ into sets of labels for vertices and edges respectively.

	\begin{example}[""Automata as transition systems""]\label{Example_AutomataAsTransSyst}
		An "automaton" $\A=(Q,\Sigma, I_0, \delta, \macc)$ can be seen as a "labelled transition system" $\T=(V,E,\msource,\mtarget,I_0,\macc, l_E)$, taking $V=Q$, $E=\{(q,a,q') \; : \; q\in Q, \, a\in \Sigma,\, \, \dd(q,a)=(q',b) \}$, $\msource$ and $\mtarget$ the projections into the first and last component respectively and  adding labels indicating the input letters:
		\[l_E: E \rightarrow \Sigma \;\; ; \;\; l_E(q,a,b,q')=a  \]
		The automaton $\A$ is "deterministic" if and only if from every vertex $v\in V$ and $a\in \Sigma$ there exists a unique edge $e\in \mout(v)$ such that $l_E(e)=a$ and $I_0$ is a singleton. 
	\end{example}
	Depending on the context, we will use one of the two equivalent formalisms introduced to work with automata.
	
	\begin{definition}[Games]
		A  \AP""game"" $\G=(V, E, \msource, \mtarget, v_0, \macc, l_V)$ is a "transition system" with a single initial vertex $v_0$ and vertices labelled by a function $l_V: V \rightarrow \{Eve, Adam\}$ that induces a partition of $V$ into vertices controlled by a player named Eve and another named Adam.
		
		During a \AP""play"", players move a token from one vertex to another, starting from the initial vertex $v_0$.
		 The player who owns the vertex $v$ where the token is placed chooses an edge in $"\mout"(v)$ and the token travels through this edge to its target. In this way, they produce an infinite "run" $\rr$ over $\G$ (that we also call a \emph{play}). We say that Eve \emph{wins} the play if it belongs to the acceptance condition $\macc$ (and Adam wins in the contrary).
		 
		 A \AP""strategy"" for a player $P\in \{Eve, Adam\}$ is a function $S_P: \mrun_{\G}\cap E^* \rightarrow E$ that tells the player which move to choose after a finite play. We say that a "play" $\rr\in \mrun_{\G}$ is consistent with the strategy $S_P$ for player $P$ if after each finite subplay $\rr' "\sqsubseteq" \rr$ ending in a vertex controlled by $P$, the next edge in $\rr$ is $S_P(\rr')$.
		 We say that Eve \AP""wins"" the game $\G$ if there is a strategy $S_{Eve}$ such that all plays consistent with $S_{Eve}$ for Eve are accepted. Dually, Adam \emph{wins} $\G$ if there is a strategy $S_{Adam}$ such that no play consistent with $S_{Adam}$ for Adam is accepted.
		 
		 Given a game $\G$, the \AP""winning region"" of $\G$ for player $P\in \{Eve, Adam\}$, written $"\W_P(\G)"$, is the set of vertices $v\in V$ such that $P$ wins the game $\G'$ obtained by setting the initial vertex to $v$ in $\G$.
	\end{definition}

%

	\paragraph*{Composition of a transition system and an automaton}
	
	Let $\T=(V,E,\msource,\mtarget,I_0, \ss:E\rightarrow \SS)$ be a "transition system graph" with transitions coloured by colours in a set $\Sigma$, and let $\A=(Q,\Sigma, q_0, \Gamma, \delta, \gamma, \macc)$ be a "deterministic automaton" over the alphabet $\Sigma$. We define the  \AP""composition"" of $\T$ and $\A$ (also called the \emph{product}) as the transition system $\A \lhd \T=(V\times Q, E',\msource',\mtarget',I_0 \times \{q_0\} ,\gamma',\macc)$, where:
	\begin{itemize}
		\item The set of vertices is the cartesian product $V\times Q$.
		\item The set of edges is $E'=E \times Q$.
		
		\item $\msource'(e,q)=(\msource(e),q)$.
		\item $\mtarget'(e,q)=(\mtarget(e),\dd(q,\ss(e)))$.
		\item The initial set is $I_0 \times \{q_0\}$.
		\item The "acceptance condition" is given by the colouring $\gamma': E \times Q \rightarrow \Gamma$, $\gamma'(e,q)= \gamma(q,\ss(e))$ and the set $\macc\subseteq \Gamma^\oo$.
	\end{itemize}

	Intuitively, a computation in $\A \lhd \T$ happens as follows: we start from a vertex $v_0\in I_0$ in $\T$ and from $q_0 \in Q$. Whenever a transition $e$ between $v_1$ and $v_2$ takes places in $\T$, it produces the colour $c(e)\in \Sigma$. Then, the automaton $\A$ makes the transition corresponding to $\ss(e)$, producing an output in $\Gamma$. In this way, a word in $\Gamma^\oo$ is produced and we can use the "acceptance condition" $\macc \subseteq \Gamma^\oo$ of the automaton as the acceptance condition for $\A \lhd \T$.
	
	
	In particular, we can perform this operation if $\T=\B$ is an automaton. We obtain in this way a new automaton $\A \lhd \B$ that uses the acceptance condition of $\A$.
	
	We refer the reader to Figure~\ref{Fig_Product_AxZF} for an example of the composition of two automata.
	
	\begin{proposition}[Folklore]
		Let $\B=(B,\Sigma_1,I_0,\Sigma_2,\delta,\macc_B)$ be an "automaton", and $\A=(A,\Sigma_2,q_0',\Gamma,\delta',\macc_A)$ be a "deterministic" automaton recognising $\L(\A)=\macc_B\subseteq \Sigma_2^\oo$. Then $\L(\A \lhd \B)=\L(\B)$. 
	\end{proposition}
%
%

	\subsection{Classes of acceptance conditions}\label{Section_Acceptingconditions}
	
	The definition of an "acceptance condition" we have used so far is a very general one. In this section we present the main types of representations for the commonly named ""$\oo$-regular conditions"".
	
	
	Let $\Gamma$ be a finite set (whose elements will be called \emph{colours}). The set $\Gamma$ will usually be the set of edges of a "transition system".
	
	\begin{description}

	\item[Büchi] A  \AP""Büchi condition"" $\macc_B$ is represented by a subset $B\subseteq \Gamma$. An infinite word $u\in \Gamma^\oo$ is accepted if some colour from $B$ appears infinitely often in $u$:
	\[ u\in \macc_B \; \Leftrightarrow \; "\minf"(u) \cap B \neq \emptyset. \]
	
	The dual notion is the co-Büchi condition.
	
	\item[co-Büchi] The  \AP""co-Büchi condition"" $\macc_{cB}$ represented by $B\subseteq \Gamma$ is defined as
		\[ u\in \macc_{cB} \; \Leftrightarrow \; "\minf"(u) \cap B = \emptyset. \]

	\item[Rabin] A  \AP""Rabin condition"" is represented by a family of ``Rabin pairs'', $R=\{(E_1,F_1),\dots,(E_r,F_r)\}$, where $E_i,F_i\subseteq \Gamma$. The condition $\macc_R$ is defined as 
		\[ u\in \macc_R \; \Leftrightarrow \; \text{there exists an index } i\in \{1,\dots,r\} \text{ such that } \quad "\minf"(u) \cap E_i \neq \emptyset \; \wedge \; \minf(u) \cap F_i = \emptyset . \]

	The dual notion of a Rabin condition is the Streett condition.
	\item[Streett] The  \AP""Streett condition"" associated to the family $S=\{(E_1,F_1),\dots,(E_r,F_r)\}$, $E_i,F_i\subseteq \Gamma$ is defined as
	\[ u\in \macc_S \; \Leftrightarrow \; \text{for all } i\in \{1,\dots,r\} \quad "\minf"(u) \cap E_i \neq \emptyset \; \rightarrow \; \minf(u) \cap F_i \neq \emptyset .\]
	
	
	\item[Parity] To define a  \AP""parity condition"" we suppose that $\Gamma$ is a finite subset of $\NN$. We define the condition $\macc_P$ as 
	\[ u\in \macc_p \; \Leftrightarrow \; \min  "\minf"(u) \text{ is even} .\]
	
	The elements of $\Gamma$ are called  \AP""priorities"" in this case. Since the expressive power of a parity condition (and the complexity of related algorithms) depends on the number of priorities used, we associate to a parity condition the interval $[\mu,\eta]$, where $\mu= \min \Gamma$ and $\eta = \max \Gamma$. Modulo a normalization (subtracting $\mu$ or $\mu-1$ to all priorities) we can suppose that $\mu=0$ or $\mu=1$. If a parity condition uses priorities in $[\mu,\eta]$ we call it a "$[\mu,\eta]$-parity condition".
	
	We remark that "Büchi conditions" are exactly $[0,1]$-parity conditions and "co-Büchi" are $[1,2]$-parity conditions.
	
	 Parity conditions are also called "Rabin chain conditions" since a parity condition is equivalent to a Rabin condition given by $R=\{(E_1,F_1),\dots,(E_r,F_r)\}$ with $F_1\subseteq E_1 \subseteq F_2  \subseteq \dots \subseteq E_r$. \\
	 
	
	\item[Muller] A  \AP""Muller condition"" is given by a family $\F\subseteq \P(\Gamma)$ of subsets of $\Gamma$. A word $u\in \Gamma^\oo$ is accepted if the colours visited infinitely often form a set of the family $\F$:
		\[ u\in \macc_\F \; \Leftrightarrow \; "\minf"(u) \in \F . \]
	We remark that Muller conditions can express all the previously defined acceptance conditions.
	
	\end{description}
	
	We will also define conditions that depend on the structure of the transition system and not only on the set of colours.
	\begin{description}

	\item[Generalised weak conditions] Let $\T=(V,E,\msource,\mtarget,q_0,\macc)$ be a "transition system". An  \AP""ordered partition"" of $\T$ is a partition of $V$, $V_1,\dots,V_s\subseteq V$ such that for every pair of vertices $p\in V_i$, $q\in V_j$, if there is a transition from $p$ to $q$, then $i\geq j$. We call each subgraph $V_i$ a  \AP""component of the ordered partition"". Every such component must be a union of "strongly connected components" of $\T$, so we can imagine that the partition is the decomposition into strongly connected components suitably ordered. We remark that given an "ordered partition" of $\T$, a "run" will eventually stay in some component $V_i$.
	
	Given different representations of acceptance conditions $\macc_1, \dots, \macc_m$ from some of the previous classes, a  \AP""generalised weak condition"" is a condition for which we allow to use the different conditions in different components of an ordered partition of a transition system. We will mainly use the following type of "generalised weak" condition:
	
	Given a transition system $\T$ and an "ordered partition" $(V_i)_{i=1}^s$, a $"\mathit{Weak}_k"$-condition is a "parity condition" such that in any component $V_i$ there are at most $k$ different priorities associated to transitions between vertices in $V_i$. It is the "generalised weak condition" for $[1,k]$ and $[0,k-1]$.

	The adjective \textit{Weak} has typically been used to refer to the condition $"\mathit{Weak}_1"$. It correspond to a partition of $\T$ into ``accepting'' and ``rejecting'' components. A "run" will be accepted if the component it finally stays in is accepting.
\end{description}
	
	"Transition systems" (resp. "automata", "games") using an acceptance condition of type $\R$ will be called  \AP""$\R$-transition systems"" (resp. \emph{$\R$-automata}, \emph{$\R$-games}). We will also say that they are \emph{labelled with an $\R$-condition}.

	\begin{remark}\label{Remark_Colours=Edges}
		As we have already observed, we can always suppose that $\GG=E$. However, this supposition might affect the size of the representation of the acceptance conditions, and therefore the complexity of related algorithms as shown in~\cite{Horn2008Explicit}.
	\end{remark}
\begin{example}\label{Example_AutomataForL}
In Figure~\ref{Fig_AutomataForL} we show three automata recognising the language
\[ \L = \{ u \in \{0,1\}^\oo \; : \; \minf(u)=\{1\} \text{ or } (\minf(u)=\{0\} \text{ and} \text{ there is an even number of 1's in } u) \} \]
and using different acceptance conditions. We represent Büchi conditions by marking the accepting transitions with a \textbullet \hspace{1mm} symbol. For Muller or parity conditions we write in each transition $\alpha:\textcolor{Green2}{a}$, with $\aa \in \{0,1\}$ the input letter and $\textcolor{Green2}{a}\in \Gamma$ the output letter. The initial vertices are represented with an incoming arrow.
\begin{figure}[ht]
	\scalebox{0.9}{
	\centering 
	\begin{minipage}[b]{0.4\textwidth} 
		\begin{tikzpicture}[square/.style={regular polygon,regular polygon sides=4}, align=center,node distance=2cm,inner sep=2pt]
		
		\node at (0,2) [state, initial] (0) {};
		\node at (2,2) [state] (1) {};
		\node at (0,0) [state] (2) {};
		\node at (2,0) [state] (3) {};

		\path[->] 
		(0)  edge [in=70,out=110,loop] 	node[above] {$0$ }   (0)
		(0)  edge [in=150,out=30] 	node[above] {$1$ }   (1)
		(0)  edge []  node[left] {$0$ }   (2)
		
		(1)  edge [in=-30,out=210]  node[below] {$1$ }   (0)
		(1)  edge [in=70,out=110,loop] 	node[above] {$0$ }   (1)
		(1)  edge [] 	node[right] {$0,1$ }   (3)
		
		(2)  edge[thick] [in=200,out=160,loop]  node[left] {$0$ } node[scale=2] { \textbullet }  (2)
		
		(3)  edge[thick] [in=200,out=160,loop]  node[anchor=south west, pos=0.3] { $1$ } node[scale=2] { \textbullet }  (3);
		
		\end{tikzpicture}
		\caption*{Non-deterministic Büchi  automaton.}
	\end{minipage}
	\begin{minipage}[b]{0.4\textwidth} 
		\begin{tikzpicture}[square/.style={regular polygon,regular polygon sides=4}, align=center,node distance=2cm,inner sep=2pt]
		
		\node at (0,2) [state, initial] (0) {};
		\node at (2,2) [state] (1) {};
		
		\path[->] 
		(0)  edge [in=70,out=110,loop] 	node[above] {$0:\textcolor{Green2}{a}$ }   (0)
		(0)  edge [in=150,out=30] 	node[above] {$1 : \textcolor{Green2}{b}$ }   (1)
		
		(1)  edge [in=-30,out=210]  node[below] {$1: \textcolor{Green2}{b}$ }   (0)
		(1)  edge [in=70,out=110,loop] 	node[above] {$0 : \textcolor{Green2}{c}$ }   (1);
		
		\end{tikzpicture}
		\caption*{  Deterministic Muller  automaton.\\ $\F_1=\{\{a\},\{b\}\}$.}
	\end{minipage}
	\begin{minipage}[b]{0.4\textwidth} 
		\begin{tikzpicture}[square/.style={regular polygon,regular polygon sides=4}, align=center,node distance=2cm,inner sep=2pt]
		
		\node at (0,2) [state, initial] (A1) {};
		\node at (0,0) [state] (A2) {};
		\node at (2,1) [state] (B) {};

		\path[->] 
		(A1)  edge [in=70,out=110,loop] 	node[above] {$0 :\textcolor{Green2}{2}$ }   (A1)
		(A1)  edge 	node[ anchor=south west, pos=0.2] {$1:\textcolor{Green2}{1}$ }   (B)

		(B)  edge[in=70,out=110,loop]  node[above] {$0:\textcolor{Green2}{1}$ }   (B)
		(B)  edge [in=10,out=230]	node[below right] {$1:\textcolor{Green2}{2}$ }   (A2)
		
		(A2)  edge [in=190,out=50]  node[above] {$1:\textcolor{Green2}{2}$ }   (B)
		(A2)  edge []  node[left] {$0:\textcolor{Green2}{1}$ }  (A1);
		
		\end{tikzpicture}
		\caption*{ Deterministic parity automaton.}
	\end{minipage}
}
	\caption{Different types of automata accepting the language $\L$.}
	\label{Fig_AutomataForL}
\end{figure}


\end{example}

In the following we will use a small abuse of notation and speak indifferently of an acceptance condition and its representation. For example, we will sometimes replace the "acceptance condition" of a "transition system" by a family of sets $\F$ (representing a "Muller condition") or by a function assigning priorities to edges.
	
	\subsection*{Equivalent conditions}
	
	Two different representations of acceptance conditions over a set $\Gamma$ are  \AP""equivalent"" if they define the same set $\macc\subseteq \Gamma^\infty$.
	
	Given a "transition system graph" $\T_G$, two representations $\R_1,\R_2$ of acceptance conditions are \emph{\AP""equivalent over"" $\T_G$} if they define the same accepting subset of runs of $"\mathpzc{Run}"_{T}$. We write $(\T_G,\R_1) \simeq(\T_G,\R_2)$. \\

	If $\A$ is the "transition system graph" of an automaton (as in Example~\ref{Example_AutomataAsTransSyst}), and $\R_1,\R_2$ are two representations of acceptance conditions such that $(\A,\R_1) \simeq(\A,\R_2)$, then they recognise the same language: $\L(\A,\R_1)=\L(\A,\R_2)$. However, the converse only holds for "deterministic" automata.
	
	\begin{proposition}\label{Prop_EquivaleceDetAutomata_ImpliesEqCondition}
		Let $\A$ be the  the "transition system graph" of a "deterministic automaton" over the alphabet $\SS$ and let $\R_1,\R_2$ be two representations of acceptance conditions such that $\L(\A,\R_1)=\L(\A,\R_2)$. Then, both conditions are "equivalent over" $\A$, $(\A,\R_1) \simeq(\A,\R_2)$.
	\end{proposition}
	\begin{proof}
		Let $\rr \in \mrun_{T}$ be an infinite "run" in $\A$, and let $u\in \SS^\oo$ be the word in the input alphabet such that $\rr$ is the "run over" $u$ in $\A$. Since $\A$ is deterministic, $\rr$ is the only "run over" $u$, then $\rr$ belongs to the "acceptance condition" of $(\A,\R_i)$ if and only if the word $u$ belongs to $\L(\A,\R_1)=\L(\A,\R_2)$, for $i=1,2$.  
	\end{proof}

	\subsection*{The deterministic parity hierarchy}
	
	As we have mentioned in the introduction, deterministic Büchi automata have strictly less expressive power than deterministic Muller automata. However, every language recognised by a Muller automaton can be recognised by a deterministic parity automaton, but we might require at least some number of priorities to do so. We can assign to each regular language $L\subseteq \Sigma^\oo$ the optimal number of priorities needed to recognise it using a "deterministic automaton". We obtain in this way the  \AP""deterministic parity hierarchy"", first introduced by Mostowski in~\cite{mostowski1984RegularEF}, represented in Figure~\ref{Fig_ParityHierarchy}. In that figure, we denote by $[\mu,\eta]$ the set of languages over an alphabet $\Sigma$ that can be recognised using a deterministic "$[\mu,\eta]$-parity" automaton. The intersection of the levels $[0,k]$ and $[1,k+1]$ is exactly the set of languages recognised using a $"\mathit{Weak}_k"$ deterministic automaton.
	
	 This hierarchy is strict, that is, for each level of the hierarchy there are languages that do not appear in lower levels~\cite{wagner1979omega}.

	\begin{figure}[ht]
		\centering 
			\begin{tikzpicture}[square/.style={regular polygon,regular polygon sides=4}, align=center,node distance=2cm,inner sep=3pt]
			
			\node at (0,0)  (00) {$[0,0]$};
			\node at (4,0)  (11) {$[1,1]$};
			\node at (2,0.75)  (Weak1) {$\mathit{Weak}_1$};
			
			\node at (0,1.5)  (01) {$[0,1]$};
			\node at (4,1.5)  (12) {$[1,2]$};
			\node at (2,2.25)  (Weak2) {$\mathit{Weak}_2$};
			
			\node at (0,3)  (02) {$[0,2]$};
			\node at (4,3)  (13) {$[1,3]$};
			\node at (2,3.75)  (Weak3) {$\mathit{Weak}_3$};			
			
			\node at (0,4.25)  (dotsLeft) {$\vdots$};
			\node at (4,4.25)  (dotsRight) {$\vdots$};
			\node at (2,4.5)  (dotsRight) {$\vdots$};

			\draw   
			(00) edge (Weak1)
			(11) edge (Weak1)
			(Weak1) edge (01)
			(Weak1) edge (12)
			
			(01) edge (Weak2)
			(12) edge (Weak2)
			(Weak2) edge (02)
			(Weak2) edge (13)
			
			(02) edge (Weak3)
			(13) edge (Weak3);
			
			\end{tikzpicture}
			\caption{The deterministic parity hierarchy.}
			\label{Fig_ParityHierarchy}
		\end{figure}
	
	We observe that the set of languages that can be recognised by a deterministic "Rabin" automaton using $r$ Rabin pairs is the level $[1,2r+1]$. Similarly, the languages recognisable by a deterministic "Streett" automaton using $s$ pairs is $[0,2s]$.\\

	For non-deterministic automata the hierarchy collapses for the level $[0,1]$ (Büchi automata).

	\subsection{Trees}\label{Sec_Trees}
	A  \AP""tree"" is a set of sequences of non-negative integers $T\subseteq \oo^*$ that is prefix-closed: if $\tt\cdot i \in T$, for $\tt \in \oo^*, i\in \oo$, then $\tt\in T$. 
	In this report we will only consider finite trees.
	
	The elements of $T$ are called  \AP""nodes"". A ""subtree"" of $T$ is a tree $T'\subseteq T$. The empty sequence~$\varepsilon$ belongs to every non-empty "tree" and it is called the  \AP""root"" of the tree.  
	 A "node" of the form $\tt \cdot i$, $i\in \oo$, is called a  \AP""child"" of $\tt$, and $\tt$ is called its  \AP""parent"". We let $\mathit{Children}(\tt)$ denote the set of children of a node $\tt$. Two different children $\ss_1,\ss_2$ of $\tt$ are called  \AP""siblings"", and we say that $\ss_1$ is  \AP""older"" than $\ss_2$ if $"\mlast"(\ss_1)<"\mlast"(\ss_2)$. We will draw the children of a node from left to right following this order.
	If two "nodes" $\tt,\ss$ verify $\tt \prefix \ss$, then $\tt$ is called an  \AP""ancestor"" of $\ss$, and $\ss$ a  \AP""descendant"" of $\tt$ (we add the adjective ``strict'' if in addition they are not equal).
	
	A "node" is called a  \AP""leaf"" of $T$ if it is a maximal sequence of $T$ (for the "prefix relation" $\prefix$). A  \AP""branch"" of $T$ is the set of prefixes of a "leaf". The set of branches of $T$ is denoted $\AP""\mathit{Branch}""(T)$. We consider the lexicographic order over leaves, that is, for two leaves $\ss_1, \, \ss_2$, $\ss_1<_{\mathit{lex}}\ss_2$ if $\ss_1(k)<\ss_2(k)$, where $k$ is the smallest position such that $\ss_1(k)\neq \ss_2(k)$. We extend this order to $"\mathit{Branch}"(T)$: let $\bb_1$, $\bb_2$ be two branches defined by the leaves $\ss_1$ and $\ss_2$ respectively. We define $\bb_1<\bb_2$ if  $\ss_1 <_{\mathit{lex}}\ss_2$. That is, the set of branches is ordered from left to right.
	
	For a node $\tt \in T$ we define $\AP""\mathit{Subtree}_T""(\tt)$ as the subtree consisting on the set of nodes that appear below $\tt$, or above it in the same branch (they are "ancestors" or "descendants" of $\tt$):
	\[  \mathit{Subtree}_T(\tt)= \{ \ss \in T \; : \; \ss \prefix \tt \text{ or } \tt \prefix \ss \}. \] We omit the subscript $T$ when the tree is clear from the context.
	
	Given a node $\tau$ of a tree $T$, the  \AP""depth"" of $\tt$ in $T$ is defined as the length of $\tt$, $\mdepth(\tt)=|\tt|$ (the root $\varepsilon$ has depth $0$). The  \AP""height of a tree"" $T$, written $\mheight(T)$, is defined as the maximal depth of a "leaf" of $T$ plus $1$. The  \AP""height of the node"" $\tt\in T$ is $\mheight(T)-\mdepth(\tt)$ (maximal leaves have height $1$).
	
	A  \AP""labelled tree"" is a pair $(T,\nu)$, where $T$ is a "tree" and $\nu: T \rightarrow \Lambda$ is a labelling function into a set of labels $\Lambda$.

	\begin{example}
		In Figure~\ref{Fig_Tree} we show a tree $T$ of "height" $4$ and we show $\msubtree(\tt)$ for $\tt=\langle 2 \rangle$. The node $\tt$ has "depth" $1$ and height $3$. The branches $\aa$, $\bb$, $\gamma$ are ordered as $\aa<\bb<\gamma$.

		\begin{figure}[ht]
			\centering 
			\begin{minipage}[b]{0.45\textwidth} 
				
				\begin{tikzpicture}[square/.style={regular polygon,regular polygon sides=4}, align=center,node distance=2cm,inner sep=3pt]
				
				\node at (0,4)  (R) {$\langle \varepsilon\rangle$};
				
				\node at (-1.5,3)  (0) {$ \; \langle  0\rangle$ };
				\node at (0,3)  (1) {$\; \langle 1\rangle$ };
				\node at (1,3)  (2) {$\tt{=}\langle2\rangle$};
				
				\node at (-2,2)  (00) {$\; \langle 0{,}0\rangle$};
				\node at (-1,2)  (01) {$\langle 0{,}1\rangle$};
				\node at (0.5,2)  (20) {$\langle2{,}0\rangle$};
				\node at (1.75,2)  (21) {$\langle 2{,}1\rangle$};
				
				\node at (-1,1)  (010) {$\; \langle  0{,}1{,}0\rangle$};
				\node at (1,1)  (210) {$\langle 2{,}1{,}0\rangle$};
				\node at (2.25,1)  (211) {$\langle 2{,}1{,}1\rangle$};
				
				\node at (-2,1.5)  (aa) {$\aa$};
				\node at (-1,0.6)  (bb) {$\bb$};
				\node at (1,0.6)  (bb) {$\gamma$};
				
				\draw   
				(R) edge (0)
				(R) edge (1)
				(R) edge (2)
				
				(0) edge (00)
				(0) edge (01)
				
				(2) edge (20)
				(2) edge (21)
				
				(01) edge (010)
				
				(21) edge (210)
				(21) edge (211);
				\end{tikzpicture}
				\caption*{\centering Tree $T$.}
				\end{minipage}
			\hspace{10mm}
				\begin{minipage}[b]{0.2\textwidth} 
					
					\begin{tikzpicture}[square/.style={regular polygon,regular polygon sides=4}, align=center,node distance=2cm,inner sep=3pt]
					
					\node at (0,4)  (R) {$\langle \varepsilon\rangle$};
					
					\node at (1,3)  (2) {$\tt$};
					
					\node at (0.5,2)  (20) {$\langle2{,}0\rangle$};
					\node at (1.5,2)  (21) {$\langle 2{,}1\rangle$};
					
					\node at (1,1)  (210) {$\langle 2{,}1{,}0\rangle$};
					\node at (2,1)  (211) {$\langle 2{,}1{,}1\rangle$};				
					
					\draw   
					(R) edge (2)
					
					(2) edge (20)
					(2) edge (21)
					
					(21) edge (210)
					(21) edge (211);
					\end{tikzpicture}
					\caption*{ \centering $\msubtree(\tt)$, $\tt=\langle 2 \rangle$.}
			\end{minipage}
		\caption{Example of a tree.}
		\label{Fig_Tree}
	\end{figure}
		\end{example}

\newpage

%% file: zielonka-tree.tex
In the previous section we have presented different classes of "acceptance conditions" for "transition systems" over infinite words, with "Muller conditions" being the most general kind of "$\oo$-regular conditions".
In order to translate a "Muller condition" $\F$ over $\Gamma$ into a simpler one, the usual procedure is to build
 a "deterministic automaton" over $\Gamma$ using a simpler condition that accepts $\F$, i.e., this automaton will accept the words $u\in \Gamma^\oo$ such that $"\minf"(u)\in \F$. As we have asserted, the simplest condition that we could use in general in such a deterministic automaton is a "parity" one, and the number of priorities that we can use is determined by the position of the "Muller condition" in the "parity hierarchy".
	
	In this section we build a deterministic parity automaton that recognises a given "Muller condition", and we prove that this automaton has minimal size and uses the optimal number of priorities. This construction is based in the notion of the "Zielonka tree", introduced in~\cite{zielonka1998infinite} (applied there to the study of the optimal memory needed to solve a Muller game). In most cases, this automaton strictly improves other constructions such as the LAR~\cite{gurevich1982trees} or its modifications~\cite{kretinski2017IAR}.
	
	All constructions and proofs on this section can be regarded as a special case of those of Section~\ref{sec:acd}. However, we include the proofs for this case here since we think that this will help the reader to understand many ideas that will reappear in Section~\ref{sec:acd} in a more complicated context.

	\subsection{The Zielonka tree automaton}\label{Section_PresentationZielonkaTree}
	In this first section we present the Zielonka tree and the parity automaton that it induces.
	
	\begin{definition}[Zielonka tree of a Muller condition]\label{Def_ZielonkaTree}
		Let $\Gamma$ be a finite set of colours and $\F\subseteq \P(\Gamma)$ a "Muller condition" over $\Gamma$. The  \AP""Zielonka tree"" of $\F$, written $T_\F$, is a tree labelled with subsets of $\Gamma$ via the labelling $\nu:T_\F \rightarrow \P(\Gamma)$, defined inductively as:
		\begin{itemize}
			\item $\nu(\varepsilon)=\Gamma$
			\item If $\tt$ is a node already constructed labelled with $S=\nu(\tt)$, we let $S_1,\dots,S_k$ be the maximal subsets of $S$ verifying the property
			\[ S_i \in \F \; \Leftrightarrow \; S\notin \F \quad \text{ for each } i=1,\dots,k . \]
			For each $i=1,\dots,k$ we add a child to $\tt$ labelled with $S_i$.
		\end{itemize}
	\end{definition}

	\begin{remark*}
		We have not specified the order in which children of a node appear in the Zielonka tree. Therefore, strictly speaking there will be several Zielonka trees of a "Muller condition". The order of the nodes will not have any relevance in this work and we will speak of ``the'' Zielonka tree of $\F$.
	\end{remark*}

	\begin{definition}
		We say that the condition $\F$ and the tree $"T_\F"$ are \intro(Zielonka){even} if $\Gamma\in \F$, and that they are \kl(Zielonka){odd} on the contrary. We associate a priority $\AP ""p_Z(\tt)""$ to each node (to each level in fact) of the "Zielonka tree" as follows:
		\begin{itemize}
			\item If $"T_\F"$ is \kl(Zielonka){even}, then $p_Z(\tt)=\mdepth(\tt)$.
			\item If $"T_\F"$ is \kl(Zielonka){odd}, then $p_Z(\tt)=\mdepth(\tt)+1$.
			
		\end{itemize}
	\end{definition}
	
	In this way, $p_Z(\tt)$ is even if and only if $\nu(\tt)\in \F$. We represent nodes $\tt\in "T_\F"$ such that $"p_Z(\tt)"$ is even as a  \AP""circle"" (\emph{round nodes}), and those for which $p_Z(\tt)$ is odd as a \emph{square}. 
	
	\begin{example}\label{Example_ZielonkaTrees}
		Let $\Gamma_1=\{a,b,c\}$ and $\F_1=\{\{a\},\{b\}\}$ (the "Muller condition" of the automaton of Example~\ref{Example_AutomataForL}). The "Zielonka tree" $T_{\F_1}$ is shown in Figure~\ref{Fig_ZielonkaTree1}. It is \kl(Zielonka){odd}.
		
		Let $\Gamma_2=\{a,b,c,d\}$ and
		\vspace{-2mm} \[\F_2=\{\{a,b,c,d\},\{a,b,d \},\{ a,c,d\},\{ b,c,d \},\{a,b\},\{a,d\},\{b,c\},\{b,d \},\{ a\},\{b \},\{d \}  \} .\]
		The "Zielonka tree" $T_{\F_2}$ is \kl(Zielonka){even} and it is shown on Figure~\ref{Fig_ZielonkaTree2}.
		
		On the right of each tree there are the priorities assigned to the nodes of the corresponding level. We have named the branches of the Zielonka trees with greek letters and we indicate the names of the nodes in \textcolor{Violet2}{violet}.
		\begin{figure}[ht]
			\begin{minipage}[b]{0.5\textwidth} 
				
				\begin{tikzpicture}[square/.style={regular polygon,regular polygon sides=4}, align=center,node distance=2cm,inner sep=3pt]
				
				\node at (0,1.5) [draw, rectangle, text height=0.5cm, text width=1cm] (R) { a,b,c \vspace{1mm}};
				
				\node at (-1,0) [draw, ellipse,text height=0.3cm, text width=0.5cm] (0) {a};
				\node at (1,0) [draw, ellipse,text height=0.3cm, text width=0.5cm] (1) {b};
				
				\node at (1.8,1.5)  (p2) {1};
				\node at (1.8,0)  (p3) {2};
				\node at (-1,-0.8)  (aa) {$\aa$};
				\node at (1,-0.8)  (bb) {$\bb$};
				\node at (-2,0)  (pp) {};
				
				\node at (0.9,1.5)  (neps) {$\textcolor{Violet2}{ \langle \varepsilon \rangle}$};
				\node at (-0.3,-0.2)  (n0) {{\scriptsize$\textcolor{Violet2}{ \langle 0\rangle }$}};
				\node at (0.35,-0.2)  (n1) {{\scriptsize$\textcolor{Violet2}{ \langle 1 \rangle} $}};
				\draw   
				(R) edge (0)
				(R) edge (1);
				\end{tikzpicture}
				\caption{ Zielonka tree $T_{\F_1}$.}
				\label{Fig_ZielonkaTree1}
			\end{minipage}
		\begin{minipage}[b]{0.4\textwidth} 
			\begin{tikzpicture}[square/.style={regular polygon,regular polygon sides=4}, align=center,node distance=2cm,inner sep=3pt]
			
			\node at (0,2.6) [draw,ellipse, minimum width=1.2cm,minimum height=0.8cm] (R) {a,b,c,d};
			
			\node at (-1.5,1.3) [draw, rectangle] (0) {a,b,c};
			\node at (1.5,1.3) [draw, rectangle, minimum width=0.8cm] (1) {c,d};
			
			\node at (-2.5,0) [draw, ellipse , minimum width=0.8cm] (00) {a,b};
			\node at (-0.5,0) [draw, ellipse , minimum width=0.8cm] (01) {b,c};
			\node at (1.5,0) [draw, ellipse , minimum width=0.8cm] (10) {d};
			
			\node at (-0.5,-1.3) [rectangle, draw, text width=0.4cm] (010) {c};

			\node at (3,2.6)  (p0) {0};
			\node at (3,1.3)  (p1) {1};
			\node at (3,0) (p2) {2};
			\node at (3,-1.3) (p3) {3};
			
			\draw   
			(R) edge (0)
			(R) edge (1)
			
			(0) edge (00)
			(0) edge (01)
			(01) edge (010)
			(1) edge (10);

			\node at (-2.5,-0.7)  (aa) {$\aa$};
			\node at (-0.5,-2)  (bb) {$\bb$};
			\node at (1.5,-0.7)  (cc) {$\gamma$};
			
			\node at (1.2,2.6)  (neps) {$\textcolor{Violet2}{ \langle \varepsilon \rangle}$};
			\node at (-2.2,1.3)  (n0) {{\scriptsize$\textcolor{Violet2}{ \langle 0\rangle }$}};
			\node at (0.8,1.3)  (n1) {{\scriptsize$\textcolor{Violet2}{ \langle 1\rangle }$}};
			
			\node at (-3.3,0)  (n00) {{\scriptsize$\textcolor{Violet2}{ \langle 0{,}0\rangle}$}};
			\node at (-1.3,0)  (n01) {{\scriptsize$\textcolor{Violet2}{ \langle 0{,}1\rangle}$}};
			\node at (0.8,-0)  (n10) {{\scriptsize$\textcolor{Violet2}{ \langle 1{,}0\rangle}$}};
			
			\node at (-1.3,-1.3)  (n010) {{\scriptsize$\textcolor{Violet2}{ \langle 0{,}1{,}0\rangle}$}};
			
			\end{tikzpicture}
			\caption{Zielonka tree $T_{\F_2}$.}
			\label{Fig_ZielonkaTree2}
		\end{minipage}

		\end{figure}
	\end{example}
	
	We show next how to use the "Zielonka tree" of $\F$ to build a "deterministic automaton" recognising the "Muller condition" $\F$.

	\begin{definition}
		For a branch $\bb \in \mbranch("T_\F")$ and a colour $a\in \Gamma$ we define $\intro(Zielonka){\mathit{Supp}}(\bb,a)=\tt$ as the "deepest" node (maximal for $\prefix$) in $\bb$ such that $a\in \nu(\tt)$. 
		
	
	\end{definition}
	
	\begin{definition}\label{Def_Nextbranch_Zielonka}
		Given a tree $T$, a branch $\bb\in \mbranch(T)$ and a node $\tt\in \bb$, if $\tt$ is not a "leaf" then it has a unique "child" $\ss_\bb$ such that $\ss_\bb \in \bb$. In this case, we let $\AP""\mathit{Nextchild}""(\bb,\tt)$ be the next "sibling" of $\ss_\bb$ on its right, that is:
		\[ \mathit{Nextchild}(\bb,\tt)=\begin{cases}
			\text{ Smallest child of } \tt  \text{ if } \ss_\bb \text{ is the greatest child of } \tt.\\[2mm]
			\text{ Smallest older sibling of } \ss_\bb  \text{ if not.}
		\end{cases} \]
		
		 We define $\AP""\mathit{Nextbranch}""(\bb,\tt)$ as the leftmost branch in $T$ (smallest in the order defined in Section~\ref{Sec_Trees}) below $"\mathit{Nextchild}"(\bb,\tt)$, if $\tt$ is not a "leaf", and we let $\mathit{Nextbranch}(\bb,\tt)= \bb$ if $\tt$ is a leaf of $T$.
		 
	\end{definition}

	\begin{example}
		In the previous example, on the tree $"T_{\F_2}"$ of Figure~\ref{Fig_ZielonkaTree2}, we have that $\kl(Zielonka){\mathit{Supp}}(\aa,c)=\langle 0 \rangle$, $\mnextc(\bb,\langle \varepsilon \rangle)=\langle 1 \rangle$, $"\mathit{Nextbranch}"(\bb,\langle \varepsilon \rangle)=\gamma$, $\mnextc(\bb,\langle 0 \rangle)=\langle 0 {,} 0 \rangle$ and  $"\mathit{Nextbranch}"(\bb,\langle 0 \rangle)=\aa$.
	\end{example}

	\begin{definition}[Zielonka tree automaton]\label{Def_ZielonkaAutomaton}
		Given a "Muller condition" $\F$ over $\Gamma$ with "Zielonka tree" $T_\F$, we define the  \AP""Zielonka tree automaton"" $\ZF=(Q,\Gamma,q_0, [\mu,\eta],\dd, p:Q\times \Gamma \rightarrow [\mu,\eta])$ as a "deterministic automaton" using a "parity" acceptance condition given by $p$, where
		\begin{itemize}
		\item $Q=\mbranch(T_\F)$, the set of states is the set of branches of $T_\F$.
		\item The initial state $q_0$ is irrelevant, we pick the leftmost branch of $T_\F$.
		\item $\dd(\bb,a)= \mnextb(\bb,\kl(Zielonka){\mathit{Supp}}(\bb,a))$. 
		
		\item $\mu=0, \; \eta=\mheight(T_\F)-1$ if $\F$ is \kl(Zielonka){even}.
		\item $\mu=1, \; \eta=\mheight(T_\F)$ if $\F$ is \kl(Zielonka){odd}.
		
		\item $p(\bb,a)="p_Z"(\kl(Zielonka){\mathit{Supp}}(\bb,a))$.
		\end{itemize}
	\end{definition}
	
	The transitions of the automaton are determined as follows: if we are in a branch $\bb$ and we read a colour $a$, then we move up in the branch $\bb$ until we reach a node $\tt$ that contains the colour $a$ in its label. Then we pick the child of $\tt$ just on the right of the branch $\bb$ (in a cyclic way) and we move to the leftmost branch below it. We produce the priority corresponding to the depth of $\tt$.
	
	
	\begin{example}\label{Example_ZielonkaTreeAutomata}
		Let us consider the conditions of Example~\ref{Example_ZielonkaTrees}. The "Zielonka tree automaton" for the "Muller condition" $\F_1$ is shown in Figure~\ref{Fig_ZielonkaTreeAutomaton1}, and that for $\F_2$ in Figure~\ref{Fig_ZielonkaTreeAutomaton2}. States are the branches of the respective "Zielonka trees".
\begin{figure}[ht]
	
		\begin{minipage}[b]{0.3\textwidth} 
		
		\begin{tikzpicture}[square/.style={regular polygon,regular polygon sides=4}, align=center,node distance=2cm,inner sep=2pt]
		
		\node at (0,2) [state, initial] (0) {$\aa$};
		\node at (2,2) [state] (1) {$\bb$};
		
		\path[->] 
		(0)  edge [in=70,out=110,loop] 	node[above] {$a:\textcolor{Green2}{2}$ }   (0)
		(0)  edge [in=150,out=30] 	node[above] {$b, c : \textcolor{Green2}{1}$ }   (1)
		
		(1)  edge [in=-30,out=210]  node[below,pos=0.3] {$a,c: \textcolor{Green2}{1}$ }   (0)
		(1)  edge [in=70,out=110,loop] 	node[above] {$b : \textcolor{Green2}{2}$ }   (1);
		
		\node at (0,-0.3) (space) {};
		\end{tikzpicture}
		\caption{ The "Zielonka tree automaton" $\Z_{\F_1}$.}
		\label{Fig_ZielonkaTreeAutomaton1}
	\end{minipage}
	\hspace{15mm}
	\begin{minipage}[b]{0.3\textwidth} 
		\begin{tikzpicture}[square/.style={regular polygon,regular polygon sides=4}, align=center,node distance=2cm,inner sep=2pt]
		
		\node at (0,2) [state, initial] (0) {$\aa$};
		\node at (3,2) [state] (1) {$\bb$};
		\node at (1.5,0) [state] (2) {$\gamma$};
		
		\path[->] 
		(0)  edge [in=70,out=110,loop] 	node[above] {$a,b:\textcolor{Green2}{2}$ }   (0)
		(0)  edge [out=20,in=160] 	node[above] {$c:\textcolor{Green2}{1}$ }   (1)
		(0)  edge [out=-50,in=130] 	node[right,pos=0.5] {$d : \textcolor{Green2}{0}$ }   (2)
		
		(1)  edge [out=110,in=70,loop]  node[above] {$b: \textcolor{Green2}{2}$ }   (1)
		(1)  edge [out=20,in=-20,loop] 	node[right] {$c : \textcolor{Green2}{3}$ }   (1)
		(1)  edge [out=200,in=-20]  node[above,pos=0.5] {$a: \textcolor{Green2}{1}$ }   (0)
		(1)  edge [out=-90,in=30] 	node[right] {$d : \textcolor{Green2}{0}$ }   (2)
		
		(2)  edge [out=230,in=270,loop]  node[left] {$c: \textcolor{Green2}{1}$ }   (2)
		(2)  edge [out=280,in=320,loop] 	node[right] {$d : \textcolor{Green2}{2}$ }   (2)
		(2)  edge [out=170,in=270]  node[left] {$a,b: \textcolor{Green2}{0}$ }   (0);
		
		\end{tikzpicture}
		\caption{The "Zielonka tree automaton" $\Z_{\F_2}$.}
		\label{Fig_ZielonkaTreeAutomaton2}
	\end{minipage}
\end{figure}
	\end{example}

	\begin{proposition}[Correctness]\label{Prop_Correctness-Zielonka}
		Let $\F\subseteq \P(\Gamma)$ be a "Muller condition" over $\Gamma$. Then, a word $u\in \Gamma^\oo$ verifies $"\minf(u)"\in \F$ ($u$ belongs to the Muller condition) if and only if $u$ is "accepted by" $"\ZF"$. 
	\end{proposition}
	
	\begin{proof}
		Let us first remark that we can associate to each input word $u\in \GG^\oo$ an infinite sequence of nodes in the Zielonka tree $\{\tt_{u,i}\}_{i=0}^{\infty}$ as follows: let $\bb_i$ be the state of the Zielonka tree automaton (the branch of the tree $"T_\F"$) reached after reading $u_0u_1\dots u_{i-1}$ ($\bb_0$ being the leftmost branch), then
		\[  \tt_{u,i}=\kl(Zielonka){\mathit{Supp}}(\bb_i,u_i)  \]
		 The sequence of priorities produced by the automaton $\ZF$ when reading $u$ is given by the priorities associated to $\tt_{u,i}$, that is, $"\moutput"_\ZF(u)=\{"p_Z"(\tt_{u,i})\}_{i=0}^{\infty}$.
		
		Let $p_{\min}$ be the minimal priority produced infinitely often in $"\moutput"_\ZF(u)$. We first show that there is a unique node appearing infinitely often in $\{\tt_{u,i}\}_{i=0}^{\infty}$ such that $"p_Z"(\tt_{u,i})=p_{\min}$. Indeed, transitions of $\ZF$ verify that $\dd(\bb,a)$ is a branch in the subtree under $\kl(Zielonka){\mathit{Supp}}(\bb,a)$. However, subtrees below two different "siblings" have disjoint sets of branches, so if $\tt_{u,i}$ and $\tt_{u,k}$ are siblings, for some $k>i$, then there must exist some transition at position $j$, $i<j<k$ such that $\kl(Zielonka){\mathit{Supp}}(\bb_j,u_j)$ is a strict "ancestor" of $\tt_{u,i}$ and $\tt_{u,k}$. Therefore,  $"p_Z"(\tt_{u,j})<"p_Z"(\tt_{u,i})$, what cannot happen infinitely often since $p_{\min}="p_Z"(\tt_{u,i})$.
		
		We let $\tt_p$ be the highest node visited infinitely often. The reasoning above also proves that all nodes appearing infinitely often in $\{\tt_{u,i}\}_{i=0}^{\infty}$ are descendants of $\tt_p$, and therefore the states appearing infinitely often in the "run over $u$" in $\ZF$ are branches in $\msubtree(\tt_p)$. We will prove that
		
		\begin{itemize}
			\item $\minf(u)\subseteq \nu(\tt_p)$.
			\item For every child $\ss$ of $\tt_p$, $\minf(u)\nsubseteq \nu(\ss)$.
		\end{itemize}
	
		Therefore, by the definition of the "Zielonka tree", $\minf(u)$ is accepted if and only if $\nu(\tt_p)\in \F$ and thus
		$$ \minf(u)\in \F \quad \Leftrightarrow \quad \nu(\tt_p)\in \F \quad \Leftrightarrow \quad "p_Z"(\tt_p)=p_{\min} \, \text{ is even.} $$
		
		In order to see that $\minf(u)\subseteq \nu(\tt_p)$, it suffices to remark that for every branch $\bb$ of  $\msubtree(\tt_p)$ and for every $a\notin \nu(\tt_p)$, we have that $\kl(Zielonka){\mathit{Supp}}(\bb, a)$ is a strict "ancestor" of $\tt_p$. Since the nodes $\tt_{u,i}$ appearing infinitely often are all descendants of $\tt_p$, the letter $a$ cannot belong to $\minf(u)$ if $a\notin \nu(\tt_p)$.
		
		Finally, let us see that $\minf(u)\nsubseteq \nu(\ss)$ for every child of $\tt_p$. Suppose that $\minf(u)\subseteq \nu(\ss)$ for some child $\ss$. Since we visit $\tt_p$ infinitely often, transitions of the form $\dd(\bb,a)$ such that $\tt_p = \msupp(\bb,a)$ take place infinitely often. By definition of $\mnextb(\bb,a)$, after each of these transitions we move to a branch passing through the next child of $\tt_p$, so we visit all children of $\tt_p$ infinitely often. Eventually we will have $\ss\in \dd(\bb,a)$ (the state reached in $\ZF$ will be some branch $\bb'$ below $\ss$).
		However, since $\minf(u)\subseteq \nu(\ss)$, for every $a\in \minf(u)$ and every $\bb'\in \msubtree(\ss)$, we would have that $\kl(Zielonka){\mathit{Supp}}(\bb',a)$ is a descendant of $\ss$, and therefore we would not visit again $\tt_p$ and the priority $p_{\min}$ would not be produced infinitely often, a contradiction.
	\end{proof}
	
	\subsection{Optimality of the Zielonka tree automaton}\label{Section_OptimalityZielonkaTree}
	We prove in this section the strong optimality of the "Zielonka tree automaton", both for the number of priorities (Proposition~\ref{Prop_OptimalPrioritiesZielonka}) and for the size (Theorem~\ref{Th_OptimalityZielonkaTree}).
	
	Proposition~\ref{Prop_OptimalPrioritiesZielonka} can be proved easily applying the results of~\cite{niwinskiwalukievicz1998Relating}. We present here a self-contained proof. 
	
	\begin{lemma}\label{Lemma_PrioritiesInRange}
		Let $\P$ be a parity "transition system" with set of edges $E$ and priorities given by $p:E \rightarrow [\mu,\eta]$ such that the minimal priority it uses is $\mu$ and the maximal one is $\eta$. If the number of different priorities used in $\P$ ($|p(E)|$) is smaller or equal than $\eta-\mu$, then we can relabel $\P$ with a parity condition that is "equivalent over" $\P$ that uses priorities in $[\mu', \eta']$ and $\eta'-\mu' < \eta-\mu$.
	\end{lemma}
	\begin{proof}
		If $\P$ uses less priorities than the length of the interval $[\mu, \eta]$, that means that there is some priority $d$, $\mu<d< \eta$ that does not appear in $\P$. Then, we can relabel $\P$ with the parity condition given by:
		\[ p'(e)=\left\{ \begin{array}{c}
		p(e) \text{ if } p(e)<d\\
		p(e)-2 \text{ if } d <p(e) \end{array} \right. \]
		that is clearly an "equivalent condition over" $\P$ that uses priorities in $[\mu, \eta-2]$.
	\end{proof}
	
		\begin{proposition}[Optimal number of priorities]\label{Prop_OptimalPrioritiesZielonka}
			The "Zielonka tree" gives the optimal number of priorities recognising a "Muller condition" $\F$. More precisely, if $[\mu,\eta]$ are the priorities used by $\ZF$ and $\P$ is another "parity automaton" recognising $\F$, it uses at least $\eta-\mu+1$ priorities. Moreover, if it uses priorities in $[\mu',\eta']$ and $\eta-\mu = \eta'-\mu'$, then $\mu$ and $\mu'$ have the same parity.
		\end{proposition}

		\begin{proof}
			Let $\P$ be a "deterministic" "parity automaton" recognising $\F$ using priorities in $[\mu',\eta']$. After Lemma~\ref{Lemma_PrioritiesInRange}, we can suppose that $\P$ uses all priorities in this interval. Let $\bb$ be a "branch" of $\TF$ of maximal length $h=\eta-\mu+1$, and let $S_0\subseteq S_{1} \subseteq \dots \subseteq S_{h-1}=\Gamma$ be the labellings of the nodes of this branch from bottom to top. Let us suppose $S_0\in \F$, the case $S_0\notin \F$ being symmetric.
			Let $a_i$ be the finite word formed concatenating the colours of $S_i$. In particular $a_i$ is accepted if and only if $i$ is even. Let $\eta'$ be the greatest priority appearing in the automaton $\P$. We prove by induction on $j$ that, for every $v\in \Gamma^*$, the "run over"
			$(a_0a_1\dots a_jv)^\oo$ in $\P$
			produces a priority smaller than or equal to $\eta'-j$, if $\eta'$ even, and smaller than or equal to $\eta'-j-1$ if $\eta'$ is odd. 
			We do here the case $\eta'$ even, the case $\eta'$ odd being symmetric.
			
			 For $j=0$ this is clear, since $\eta'$ is the greatest priority. For $j>0$, if it was not true, the smallest priority produced infinitely often reading $(a_0a_1\dots a_jv)^\oo$ would be strictly greater than $\eta'-j$ for some $v\in \Gamma^*$. Since $\eta'-j$ has the same parity as $j$ and $S_j\in \F$ if and only if $j$ is even, then the smallest priority produced infinitely often reading $(a_0a_1\dots a_jv)^\oo$ must have the same parity than $j$ and cannot be $\eta'-j+1$, so it is greater than $\eta'-j+2$. However, by induction hypothesis, the run over $(a_0a_1\dots a_{j-1}w)^\oo$ produces a priority smaller than or equal to $\eta'-(j-1)$ for every $w$, in particular for $w=a_jv$, contradicting the induction hypothesis.
			 
			 In particular, taking $v=\varepsilon$, we have proved that the "run over"
			 $(a_0a_1\dots a_{h-1})^\oo$ in $\P$ produces a priority smaller than or equal to $\eta'-(h-1)$ that has to be even if and only if $\mu$ is even. Therefore, $\P$ must use all priorities in $[\eta'-(h-1),\eta']$, that is, at least $h$ priorities.
			 
			
		\end{proof}

	
	In order to prove Theorem~\ref{Th_OptimalityZielonkaTree} we introduce the definition of an \emph{$X$-strongly connected component} and we present two key lemmas.
	
	\begin{definition}
		Let $\A=(Q,\Sigma, q_0, \Gamma, \delta,\macc)$ be a "deterministic automaton" and $X\subseteq \Sigma$ a subset of letters of the input alphabet. An  \AP""X-strongly connected component"" (abbreviated $X$-SCC) is a non-empty subset of states $S\subseteq Q$ such that:
		\begin{itemize}
			\item For every state $q\in S$ and every letter $x\in X$, $\delta(q,x)\in S$.
			\item For every pair of states $q,q'\in S$ there is a finite word $w\in X^*$ such that $\delta(q,w)=q'$.
		\end{itemize}
	\end{definition}
	
	That is, an $X$-SCC of $\A$ are the states of an "$X$-complete" part of $\A$ that forms a "strongly connected subgraph".
	
	
	\begin{lemma}\label{Lemma_ExistenceSCC}
		For every "deterministic automaton" $\A=(Q,\Sigma, q_0, \Gamma, \delta,\macc)$ and every subset of letters $X\subseteq \Sigma$ there is an "accessible" "$X$-SCC" in $\A$.
	\end{lemma}
	\begin{proof}

		Restricting ourselves to the set of "accessible" states of $\A$ we can suppose that every state of the automaton is accessible.
		
		We prove the lemma by induction on $|\A|$. For $|\A|=1$, the state of the automaton forms an $X$-SCC. For $|\A|>1$, if  $Q$ is not an "$X$-SCC", there are $q,q'\in Q$ such that there does not exist a word $w\in X^*$ such that $\delta(q,w)=q'$. Let 
		\[ Q_q=\{ p\in Q \; : \; \exists u\in X^* \text{ such that } p= \dd(q,u) \} .\]
		Since $q'\notin Q_q$, the set $Q_q$ is strictly smaller than $Q$. The set $Q_q$ is non-empty and closed under transitions labelled by letters of $X$, so the restriction of $\A$ to this set of states and the alphabet $X$ forms an automaton $\A_{Q_q,X}=(Q_q,X, q, \Gamma, \delta')$ (where $\dd'$ is the restriction of $\dd$ to these states and letters). By induction hypothesis, $\A_{Q_q,X}$ contains an $X$-SCC that is also an $X$-SCC for $\A$.	
	\end{proof}
	
	

	\begin{lemma}\label{Lemma_Disjoint_X-SCC}
		Let $\F$ be a "Muller condition" over $\Gamma$, $T_\F$ its "Zielonka tree" and $\P=\ab(P,\Gamma,p_0,\ab [\mu',\eta'],\ab \delta_P,p':P\rightarrow [\mu',\eta'])$ a "deterministic" parity automaton recognising $\F$. Let $\tt$ be a node of $T_\F$ and $C=\nu(\tt)\subseteq \Gamma$ its label. Finally, let $A,B\subseteq C$ be two different subsets maximal such that $C\in \F \, \Leftrightarrow \, A \notin \F$, $C\in \F \, \Leftrightarrow \, B \notin \F$ (they are the labels of two different children of $\tt$). Then, if $P_A$ and $P_B$ are two "accessible" "$A$-SCC" and "$B$-SCC" of $\P$ respectively, they satisfy $P_A \cap P_B=\emptyset$.
	\end{lemma}
	
	\begin{proof}
		We can suppose that $C\in \F$ and $A,B \notin \F$. Suppose that there is a state $q\in P_A \cap P_B$. Let $A=\{a_1,\dots,a_l\}$, $B=\{b_1,\dots,b_r\}$ and $q_1=\delta(q,a_1\cdots a_l)\in A$, $q_2=\delta(q,b_1\cdots b_r)\in B$. By definition of an $X$-SCC, there are words $u_1 \in A^*$, $u_2\in B^*$ such that $\dd(q_1,u_1)=q$ and $\dd(q_2,u_2)=q$. Since $A,B \notin \F$, the minimum priorities $p_1$ and $p_2$ produced by the "runs over" $(a_1\cdots a_lu_1)^\oo$ and $(b_1\cdots b_r u_2)^\oo$ starting from $q$ are odd. However, the run over $(a_1\cdots a_lu_1b_1\cdots b_r u_2)^\oo$ starting from $q$ must produce an even minimum priority (since $A\cup B \in \F$), but the minimum priority visited in this run is $\min\{p_1,p_2 \}$, odd, which leads to a contradiction.
	\end{proof}
	
%
	\begin{theorem}[Optimal size of the Zielonka tree automaton]\label{Th_OptimalityZielonkaTree}
		Every "deterministic" "parity automaton" $\P=(P,\Gamma,p_0,[\mu',\eta'],\delta_P,p':P\times \Gamma\rightarrow [\mu',\eta'])$ accepting a "Muller condition" $\F$ over $\Gamma$ verifies
			$ |"\ZF"|\leq |\P| .$
	\end{theorem}
	\begin{proof}
		Let $\P$
		 be a deterministic parity automaton accepting $\F$. To show $|\ZF|\leq |\P|$ we proceed by induction on the number of colours $|\Gamma|$. For $|\Gamma|=1$ the two possible "Zielonka tree automata" have one state, so the result holds. Suppose $|\Gamma|>1$ and consider the first level of $\TF$.

		Let $n$ be the number of children of the root of $\TF$. For $i=1,...,n$, let $A_i=\nu(\tt_i)\subseteq C$ be the label of the $i$-th child of the root of $\TF$, $\tt_i$, and let $n_i$ be the number of branches of the subtree under $\tt_i$, $\msubtree(\tt_i)$. We remark that $|\ZF|=\sum_{i=1}^{n}n_i$. Let $\F\upharpoonright A_i:=\{F\in \F \; : \; F\subseteq A_i \}$. Since each $A_i$ verifies $|A_i|<|C|$ and the "Zielonka tree" for $\F \rest A_i$ is the subtree of $\TF$ under the node $\tt_i$, every "deterministic parity automaton" accepting $\F \rest A_i$ has at least $n_i$ states, by induction hypothesis. 

		Thanks to Lemma~\ref{Lemma_ExistenceSCC}, for each $i=1,\dots,n$ there is an accessible "$A_i$-SCC" in $\P$, called $P_i$. Therefore, the restriction to $P_i$ (with an arbitrary initial state) is an automaton recognising $F\rest A_i$. By induction hypothesis, for each $i=1,...,n $, $|P_i|\geq n_i$. Thanks to Lemma~\ref{Lemma_Disjoint_X-SCC}, we know that for every $i,j\in \{1,\dots,n\},\; i \neq j$, $P_i \cap P_j = \emptyset$. We deduce that 
		\begin{equation*}\label{Eq_NumberStates}
		|\P|\geq \sum\limits_{i=1}^{n}|P_i|\geq \sum\limits_{i=1}^{n}n_i=|\ZF| .\qedhere
		\end{equation*}

	\end{proof}
	
	\subsection{The Zielonka tree of some classes of acceptance conditions}\label{Section_ZielonkaTree-SomeTypes}
	 In this section we present some results proven by Zielonka in~\cite[Section 5]{zielonka1998infinite} that show how we can use the "Zielonka tree" to deduce if a "Muller condition" is representable by a "Rabin", "Streett" or "parity" condition. These results are generalised to "transition systems" in Section~\ref{Section_4.2.StructuralParity}.
	 
	 We first introduce some definitions. The terminology will be justified by the upcoming propositions.
	 
	 \begin{definition}\label{Def_ShapesOfTrees}
	 	Given a tree $T$ and a function assigning priorities to nodes, $p:T\rightarrow \NN$, we say that $(T,p)$ has
	 	\begin{itemize}
	 		\item  \AP""Rabin shape"" if every node with an even priority assigned ("round" node) has at most one child.
	 		
	 		\item  \AP""Streett shape"" if every node with an odd priority assigned ("square" node) has at most one child.
	 		
	 		\item  \AP""Parity shape"" if every node has at most one child.
	 	\end{itemize}
	 \end{definition}

	\begin{proposition}\label{Prop_RabinZielonkaTree}
		
		Let $\F\subseteq \P( \Gamma)$ be a "Muller condition". The following conditions are equivalent:
		\begin{enumerate}
			\item $\F$ is "equivalent" to a "Rabin condition".
			\item The complement of the family $\F$ is closed under union.
			\item $"T_\F"$ has "Rabin shape". 
		\end{enumerate}
		
%
	\end{proposition}

	\begin{proposition}\label{Prop_StreettZielonkaTree}
		
		Let $\F\subseteq \P( \Gamma)$ be a "Muller condition". The following conditions are equivalent:
		\begin{enumerate}
			\item $\F$ is "equivalent" to a "Streett condition".
			\item The family $\F$ is closed under union.
			\item $"T_\F"$ has "Streett shape". 
		\end{enumerate}
		
	\end{proposition}

\begin{proposition}\label{Prop_ParityZielonkaTree}
	Let $\F\subseteq \P( \Gamma)$ be a "Muller condition". The following conditions are equivalent:
	\begin{enumerate}
		\item $\F$ is "equivalent" to a "parity condition".
		\item The family $\F$ and its complement are closed under union.
		\item $"T_\F"$ has "parity shape". 
	\end{enumerate}
	Moreover, if some of these conditions is verified, $\F$ is "equivalent" to a "$[1,\eta]$-parity condition" (resp. "$[0,\eta-1]$-parity condition") if and only if  $\mheight("T_\F")\leq \eta $ and in case of equality $"T_\F"$ is \kl(Zielonka){odd} (resp. \kl(Zielonka){even}).
\end{proposition}
%
%

%
\begin{corollary}
	A "Muller condition" $\F \subseteq \P(\GG)$ is equivalent to a parity condition if and only if it is equivalent to both Rabin and Streett conditions.
\end{corollary}
	
	\begin{example}
		In Figures~\ref{Fig_ZielonkaTreeParity} and~\ref{Fig_ZielonkaTreeRabin} we represent "Zielonka trees" for some examples of "parity" and "Rabin" conditions.

		We remark that for a fixed number of "Rabin" (or Street) pairs we can obtain "Zielonka trees" of very different shapes that range from a single branch (for "Rabin chain conditions") to a tree with a branch for each "Rabin pair" and height $3$. 

	\begin{figure}[ht]
		\centering
		\begin{minipage}[b]{0.3\textwidth} 
			
			\begin{tikzpicture}[square/.style={regular polygon,regular polygon sides=4}, align=center,node distance=2cm,inner sep=3pt]
			
			\node at (0,3) [draw, rectangle, text height=0.3cm, text width=2cm] (R) { 1, 2, 3, 4};
			
			\node at (0,2) [draw, ellipse,text height=0.2cm, text width=1.2cm] (0) {2, 3, 4};
			
			\node at (0,1) [draw, rectangle, text height=0.3cm, text width=1cm] (1) { 3, 4};
			
			\node at (0,0) [draw, ellipse,text height=0.2cm, text width=0.5cm] (2) {4};
			
			\draw   
			(R) edge (0)
			(0) edge (1)
			(1) edge (2);
			
			\node at (-2,0) {};
			
			\end{tikzpicture}
			\caption{ Zielonka tree of a parity condition.}
			\label{Fig_ZielonkaTreeParity}
		\end{minipage}
		\hspace{10mm}
		\begin{minipage}[b]{0.3\textwidth} 
			\begin{tikzpicture}[square/.style={regular polygon,regular polygon sides=4}, align=center,node distance=2cm,inner sep=3pt]
			
			\node at (0,3) [draw,rectangle, minimum width=0.7cm,, minimum height=0.7cm,,scale=0.8] (R) { }; 
			
			\node at (-1.6,2) [draw, ellipse, minimum width=0.7cm,, minimum height=0.7cm,, scale=0.8] (0) {}; 
			\node at (1.6,2) [draw, ellipse, minimum width=0.7cm,, minimum height=0.7cm,, scale=0.8] (2) {}; 
			
			\node at (-1.6,1) [draw, rectangle, minimum width=0.7cm,, minimum height=0.7cm,, scale=0.8] (00) {}; 
			\node at (1.6,1) [draw, rectangle, minimum width=0.7cm,,minimum height=0.7cm,, scale=0.8] (20) {}; 
			
			\node at (-2.2,0) [draw, ellipse , minimum width=0.7cm,,minimum height=0.7cm,, scale=0.8] (000) {}; 
			\node at (-1,0) [draw, ellipse ,minimum width=0.7cm,, minimum height=0.7cm,, scale=0.8] (001) {}; 
			
			\node at (-1,-1) [draw, rectangle , minimum width=0.7cm,,minimum height=0.7cm,, scale=0.8] (0010) {}; 
			
			\node at (1,0) [draw, ellipse , minimum width=0.7cm,,minimum height=0.7cm,, scale=0.8] (200) {}; 
			\node at (2.2,0) [draw, ellipse, minimum width=0.7cm,,minimum height=0.7cm,, scale=0.8] (201) {}; 
			\node at (1,-1) [draw, rectangle ,minimum width=0.7cm,,minimum height=0.7cm,, scale=0.8] (2000) {}; 
			
			\node at (2.2,-1) [draw, rectangle ,minimum width=0.7cm,,minimum height=0.7cm,, scale=0.8] (2010) {}; 
			
			\draw   
			(R) edge (0)
			(R) edge (2)
			
			(0) edge (00)
			(2) edge (20)
			
			(00) edge (000)
			(00) edge (001)
			
			(20) edge (200)
			(20) edge (201)
			
			(200) edge (2000)
			
			(001) edge (0010)
			(201) edge (2010);
			
			\end{tikzpicture}
			\caption{Zielonka tree of a Rabin condition.}
			\label{Fig_ZielonkaTreeRabin}
		\end{minipage}
		
	\end{figure}
	
	\end{example}

\newpage

%% file: acd.tex
In this section we present our main contribution: an optimal transformation of "Muller" "transition systems" into "parity" transition systems. Firstly, we formalise what we mean by ``a transformation'' using the notion of "locally bijective morphisms" in Section~\ref{Section_LocallyBijectiveMorphism}.
Then, we describe a transformation from a "Muller transition system" to a parity one. Most transformations found in the literature use the "composition" of the "transition system" by a parity automaton recognising the Muller condition (such as $"\Z_\F"$). In order to achieve optimality this does not suffice, we need to take into account the structure of the transition system. Following ideas already present in~\cite{wagner1979omega}, we analyse the alternating chains of accepting and rejecting cycles of the transition system. We arrange this information in a collection of "Zielonka trees" obtaining a data structure, the "alternating cycle decomposition", that subsumes all the structural information of the transition system necessary to determine whether a "run" is accepted or not. We present the "alternating cycle decomposition" in Section~\ref{Subsection_ACD} and we show how to use this structure to obtain a parity transition system that mimics the former Muller one in Section~\ref{Subsection_ACD-transformation}.
	
	In Section~\ref{Section_OptimalityACD} we prove the optimality of this construction. More precisely, we prove that if $\P$ is a parity transition system that admits a "locally bijective morphism" to a Muller transition system $\T$, then the transformation of $\T$ using the alternating cycle decomposition provides a smaller transition system than $\P$ and using less priorities.
	
	\subsection{Locally bijective morphisms as witnesses of transformations}\label{Section_LocallyBijectiveMorphism}
	We start by defining locally bijective morphisms.
	
	\begin{definition}
		Let $\T=(V,E,\msource,\mtarget,I_0,\macc)$, $\T'=(V',E',\msource',\mtarget',I_0',\macc')$ be two "transition systems". A  \AP""morphism of transition systems"", written $\pp: \T \rightarrow \T'$, is a pair of maps $(\pp_V: V \rightarrow V', \pp_E: E \rightarrow E')$ such that:
		\begin{itemize}
			\item $\pp_V(v_0)\in I_0'$ for every $v_0\in I_0$ (initial states are preserved).
			\item $\msource'(\pp_E(e))=\pp_V(\msource(e))$ for every $e\in E$ (origins of edges are preserved).

			\item $\mtarget'(\pp_E(e))=\pp_V(\mtarget(e))$ for every $e\in E$ (targets of edges are preserved).
			
			\item For every "run" $\rr \in \mrun_{\T}$, 
			$\rr \in \macc \; \Leftrightarrow \; \pp_E(\rr) \in \macc'$ (acceptance condition is preserved).
		\end{itemize}
	
	If $(\T,l_V,l_E)$, $(\T',l_V',l_E')$ are "labelled transition systems", we say that $\pp$ is a  \AP""morphism of labelled transition systems"" if in addition it verifies 
	\begin{itemize}
		\item $l_V'(\pp_V(v))=l_V(v)$ for every $v\in V$ (labels of states are preserved).
		\item $l_E'(\pp_E(e))=l_E(e)$ for every $e\in V$ (labels of edges are preserved).
	\end{itemize}
	\end{definition}

	We remark that it follows from the first three conditions that if $\rr \in \mrun_{\T}$ is a "run" in $\T$, then $\pp_E(\rr)\in \mrun_{T'}$ (it is a "run" in $\T'$ starting from some initial vertex).

	Given a "morphism of transition systems" $(\pp_V,\pp_E)$, we will denote both maps by $\pp$ whenever no confusion arises. We extend $\pp_E$ to $E^*$ and $E^\oo$ component wise.
	
	\begin{remark*}
		A "morphism of transition systems" $\pp=(\pp_V, \pp_E)$ is unequivocally characterised by the map $\pp_E$. Nevertheless, it is convenient to keep the notation with both maps.
	\end{remark*}
	
	\begin{definition}
	Given two "transition systems" $\T=(V,E,\msource,\mtarget,I_0,\macc)$, $\T'=(V',E',\msource',\mtarget',I_0',\macc')$, a "morphism of transition systems" $\pp: \T \rightarrow \T'$ is called
	
	\begin{itemize}
		\item  \AP""Locally surjective"" if 
		\begin{itemize}
		\item For every $v_0'\in I_0'$ there exists $v_0\in I_0$ such that $\pp(v_0)=v_0'$.
		\item 
		For every $v\in V$ and every $ e'\in E'$ such that $ \msource'(e')=\pp(v)$
		there exists $e\in E $ such that $ \pp(e)=e' $ and $ \msource(e)=v$. 

		\end{itemize}
		
		\item "Locally injective" if 
		\begin{itemize}
			\item For every $v_0'\in I_0'$, there is at most one $v_0\in I_0$ such that $\pp(v_0)=v_0'$. 
			
			\item For every $ v\in V$ and every $ e'\in E' $ such that $ \msource'(e')=\pp(v) $
			if there are $ e_1,e_2\in E $ such that $ \pp(e_i)=e'$ and $ \msource(e_i)=v$, for $ i=1,2 $, then $ e_1=e_2  $.

		\end{itemize}
		
		\item "Locally bijective" if it is both "locally surjective" and "locally injective".
	\end{itemize}

	\end{definition}

	\begin{remark*}
		Equivalently, 
		a "morphism of transition systems" $\pp$ is "locally surjective" (resp. injective) if the restriction of $\pp_E$ to $"\mout"(v)$ is a surjection (resp. an injection) into $"\mout"(\pp(v))$ for every $v\in V$ and the restriction of $\pp_V$ to $I_0$ is a surjection (resp. an injection) into $I_0'$. 
		
		If we only consider the "underlying graph" of a "transition system", without the "accepting condition", the notion of "locally bijective morphism" is equivalent to the usual notion of bisimulation. However, when considering the accepting condition, we only impose that the acceptance of each "run" must be preserved (and not that the colouring of each transition is preserved). This allows us to compare transition systems using different classes of accepting conditions. 
	\end{remark*}

	We state two simple, but key facts.
	\begin{fact}\label{Fact_LocBijMorph_BijectionRuns}
		If  $\pp: \T \rightarrow \T'$ is a "locally bijective morphism", then $\pp$ induces a bijection between the runs in $\mrun_{\T}$ and $\mrun_{\T'}$ that preserves their acceptance.
	\end{fact}

	
	\begin{fact}\label{Fact_LocSurjMorph_OntoAccessible}
		If $\pp$ is a "locally surjective morphism", then it is onto the "accessible part" of $\T'$. That is, for every "accessible" state $v'\in \T'$, there exists some state $v\in \T$ such that $\pp_V(v)=v'$. In particular if every state of $\T'$ is "accessible", $\pp$ is surjective.
	\end{fact}

%

	Intuitively, if we transform a "transition system" $\T_1$ into $\T_2$ ``without adding non-determinism'', we will have a locally bijective morphism $\pp: \T_2 \rightarrow \T_1$. In particular, if we consider the "composition" $\T_2=\B \lhd \T_1$ of $\T_1$ by some "deterministic automaton" $\B$, as defined in Section~\ref{sec:notations}, the projection over $\T_1$ gives a "locally bijective morphism" from $\T_2$ to $\T_1$.

	\begin{example}\label{Example_ProductAndLocBijMorphism}
		Let $\A$ be the "Muller automaton" presented in the Example~\ref{Fig_AutomataForL}, and $\Z_{\F_1}$ the "Zielonka tree automaton" for its Muller condition $\F_1=\{\{a\},\{b\}\}$ as in the Figure~\ref{Fig_ZielonkaTreeAutomaton1}. We show them in Figure~\ref{Fig_MullerAndZielonkaAutomata} and their "composition" $\ZF \lhd \A$ in Figure~\ref{Fig_Product_AxZF}.	If we name the states of $\A$ with the letters $A$ and $B$, and those of $\Z_{\F_1}$ with $\aa,\bb$, there is a locally bijective morphism $\pp: \ZF \lhd \A \rightarrow \A$ given by the projection on the first component
		\[ \pp_V((X,y))=X \; \text{ for } X\in \{A,B\},\, y\in \{\aa,\bb\} \] 
		and $\pp_E$ associates to each edge $e\in \mout(X,y)$ labelled by $a\in \{0,1\}$ the only edge in $\mout(X)$ labelled with $a$. 
	\begin{figure}[ht]
		\centering
			\begin{minipage}[b]{0.4\textwidth} 
			\begin{tikzpicture}[square/.style={regular polygon,regular polygon sides=4}, align=center,node distance=2cm,inner sep=2pt]
			
			\node at (0,2) [state, initial] (0) {A};
			\node at (2,2) [state] (1) {B};
			
			\path[->] 
			(0)  edge [in=70,out=110,loop] 	node[above] {$0:\textcolor{Green2}{a}$ }   (0)
			(0)  edge [in=150,out=30] 	node[above] {$1 : \textcolor{Green2}{b}$ }   (1)
			
			(1)  edge [in=-30,out=210]  node[below] {$1: \textcolor{Green2}{b}$ }   (0)
			(1)  edge [in=70,out=110,loop] 	node[above] {$0 : \textcolor{Green2}{c}$ }   (1);
			
			\end{tikzpicture}
			\caption*{Muller automaton $\A$ with accepting condition $\F_1=\{\{a\},\{b\}\}$.}
		\end{minipage}
		\hspace{0mm}
		\begin{minipage}[b]{0.3\textwidth} 
			\begin{tikzpicture}[square/.style={regular polygon,regular polygon sides=4}, align=center,node distance=2cm,inner sep=2pt]
			
			\node at (0,2) [state, initial] (0) {$\aa$};
			\node at (2,2) [state] (1) {$\bb$};
			
			\path[->] 
			(0)  edge [in=70,out=110,loop] 	node[above] {$a:\textcolor{Green2}{2}$ }   (0)
			(0)  edge [in=-70,out=-110,loop] 	node[below] {$c:\textcolor{Green2}{1}$ }   (0)
			(0)  edge [in=150,out=30] 	node[above] {$b : \textcolor{Green2}{1}$ }   (1)
			
			(1)  edge [in=-30,out=210]  node[below,pos=0.3] {$a,c: \textcolor{Green2}{1}$ }   (0)
			(1)  edge [in=70,out=110,loop] 	node[above] {$b : \textcolor{Green2}{2}$ }   (1);
			
			\end{tikzpicture}
			\caption*{\centering $\Z_{\F_1}$. }
		\end{minipage}
		\caption{Muller automaton and the Zielonka tree automaton of its acceptance condition.}
		\label{Fig_MullerAndZielonkaAutomata}
		\end{figure}
	
		\begin{figure}[ht]
		\centering
			\begin{tikzpicture}[square/.style={regular polygon,regular polygon sides=4}, align=center,node distance=2cm,inner sep=2pt]
			\node at (0,3) [state, initial] (Aa) {A,$\aa$};
			\node at (3,3) [state] (Bb) {B,$\bb$};
			\node at (0,0) [state] (Ab) {A,$\bb$};
			\node at (3,0) [state] (Ba) {B,$\aa$};
			
			\path[->] 
			(Aa)  edge [in=70,out=110,loop] 	node[above] {$0:\textcolor{Green2}{2}$ }   (Aa)
			(Aa)  edge [] 	node[above] {$1 : \textcolor{Green2}{1}$ }   (Bb)
			
			(Ab)  edge []  node[left] {$0: \textcolor{Green2}{1}$ }   (Aa)
			(Ab)  edge [out=60,in=210] 	node[left,pos=0.6] {$1 : \textcolor{Green2}{2}$ }   (Bb)
			
			(Ba)  edge [out=20,in=-20,loop] 	node[right] {$0:\textcolor{Green2}{1}$ }   (Ba)
			(Ba)  edge [] 	node[above, pos=0.4] {$1 : \textcolor{Green2}{1}$ }   (Ab)
			
			(Bb)  edge []  node[right] {$0: \textcolor{Green2}{1}$ }   (Ba)
			(Bb)  edge [out=240,in=30] 	node[right] {$1 : \textcolor{Green2}{2}$ }   (Ab);
			
%
%
%
%
			\end{tikzpicture}
			\caption{The composition $\ZF \lhd \A$.}
			\label{Fig_Product_AxZF}
	\end{figure}
	\end{example}

	\begin{remark*}
		We know that $\ZF$ is a minimal automaton recognising the "Muller condition" $\F$ (Theorem~\ref{Th_OptimalityZielonkaTree}). However, the "composition" $Z_{\F_1} \lhd \A$ has $4$ states, and in the Example~\ref{Example_AutomataForL} (Figure~\ref{Fig_AutomataForL}) we have shown a parity automaton recognising $\L(\A)$ with only $3$ states. Moreover, there is a "locally bijective" morphism from this smaller parity automaton to $\A$ (we only have to send the two states on the left to $A$ and the state on the right to $B$). In the next section we will show a transformation that will produce the parity automaton with only $3$ states starting from $\A$.
	\end{remark*}

	\paragraph*{Morphisms of automata and games}
	Before presenting the optimal transformation of Muller transition systems, we will state some facts about "morphisms" in the particular case of "automata" and "games". When we speak about a "morphism" between two automata, we always refer implicitly to the morphism between the corresponding "labelled transition systems", as explained in Example~\ref{Example_AutomataAsTransSyst}.

	\begin{fact}\label{Fact_MorphismDetAutomata-LocBijective}
		A "morphism" $\pp=(\pp_V,\pp_E)$  between two "deterministic automata" is always "locally bijective" and it is completely characterised by the map $\pp_V$.  
	\end{fact}
	\begin{proof}
		For each letter of the input alphabet and each state, there must be one and only one outgoing transition labelled with this letter.
	\end{proof}

	
	\begin{proposition}\label{Prop3.6_MorphImpliesLanguages}
		Let $\A=(Q,\Sigma, I_0, \Gamma, \delta, \macc)$, $\A'=(Q',\Sigma, I_0', \Gamma, \delta', \macc')$ be two (possibly non-deterministic) "automata". If there is a "locally surjective morphism"  $\pp: \A \rightarrow \A'$, then $"\L(\A)"="\L(\A')"$.
	\end{proposition}
	
	\begin{proof}
		Let $u\in \Sigma^\oo$. If $u\in \L(\A)$ there is an accepting run, $\rr$, over $u$ in $\A$. By the definition of a "morphism of labelled transition systems", $\pp(\rr)$ is also an accepting "run over $u$" in $\A'$.

		Conversely, if $u\in \L(\A')$ there is an accepting "run over $u$" $\rr'$ in $\A'$. Since $\pp$ is locally surjective there is a run $\rr$ in $\A$, such that $\pp(\rr)=\rr'$, and therefore $\rr$ is an accepting run over $u$.
	\end{proof}
	
%
%
	\begin{remark*}
		The converse of the previous proposition does not hold: $"\L(\A)"=\L(\A')$ does not imply the existence of morphisms $\pp: \A \rightarrow \A'$ or $\pp: \A' \rightarrow \A$, even if $\A$ has minimal size among the Muller automata recognising $\L(\A)$. 
	\end{remark*}

	If $\A$ and $\A'$ are "non-deterministic automata" and $\pp:\A \rightarrow \A'$ is a "locally bijective morphism", then $\A$ and $\A'$ have to share some other important semantic properties. Two classes of automata that have been extensively studied are (strongly) unambiguous and good-for-games automata. An automaton is \AP""unambiguous"" if for every input word $w\in \SS^\oo$ there is at most one accepting "run over" $w$, and it is ""strongly unambiguous"" if there is at most one "run over" $w$. \AP""Good-for-games"" automata (GFG), first introduced by Henzinger and Piterman in~\cite{Henzinger2006SolvingGW}, are automata that can resolve the non-determinism depending only in the prefix of the word read so far. These types of automata have many good properties and have been used in different contexts (as for example in the model checking of LTL formulas~\cite{Couvreur2003AnOA} or in the theory of cost functions~\cite{Colcombet2009CostFunctions}). "Strongly unambiguous" automata can recognise $\oo$-regular languages using a "Büchi" condition (see~\cite{Carton2003UnambiguousBuchi}) and GFG automata have strictly more expressive power than deterministic ones, being in some cases exponentially smaller (see~\cite{Boker2013NondetUnknownFuture, KuperbergGFG}).
	
	We omit the proof of the next proposition, being a consequence of Fact~\ref{Fact_LocBijMorph_BijectionRuns} and of the argument from the proof of Proposition~\ref{Prop_GamesLocBijMorph}.

\begin{proposition}
	Let $\A$ and $\A'$ be two "non-deterministic automata". If $\pp:\A \rightarrow \A'$ is a "locally bijective morphism", then 
	\begin{itemize}
		
		\item $\A$ is unambiguous if and only if $\A'$ is unambiguous.
		\item $\A$ is strongly unambiguous if and only if $\A'$ is strongly unambiguous.
		\item $\A$ is GFG if and only if $\A'$ is GFG.
	\end{itemize} 
\end{proposition}
	
	Having a "locally bijective morphism" between two games implies that the "winning regions" of the players are preserved.
	
	\begin{proposition}\label{Prop_GamesLocBijMorph}
		Let $\G=(V, E, \msource, \mtarget, v_0, \macc, l_V)$ and $\G'=\ab (V',\ab E', \ab \msource', \ab \mtarget', v_0', \macc', l_V')$ be two "games" such that there is a "locally bijective morphism" $\pp:\G \rightarrow \G'$. Let $P\in \{Eve, Adam\}$ be a player in those games. Then, $P$ wins $\G$ if and only if she/he wins $\G'$. Moreover, if $\pp$ is surjective, the "winning region" of $P$ in $\G'$ is the image by $\pp$ of her/his winning region in $\G$, $"\W_P"(\G')=\pp("\W_P"(\G))$. 
	\end{proposition}
	\begin{proof}
	Let $S_P: \mrun_{\G}\cap E^* \rightarrow E$ be a winning "strategy" for player $P$ in $\G$. Then, it is easy to verify that the strategy $S_P': \mrun_{\G'}\cap E'^* \rightarrow E'$ defined as
	$ S_P'(\rr') = \pp_E ( S_P(\pp^{-1}(\rr')))  $
	is a winning "strategy" for $P$ in $\G'$. (Remark that thanks to Fact~\ref{Fact_LocBijMorph_BijectionRuns}, the morphism $\pp$ induces a bijection over "runs", allowing us to use $\pp^{-1}$ in this case).
	
	Conversely, if $S_P': \mrun_{\G'}\cap E'^* \rightarrow E'$ is a winning "strategy" for $P$ in $\G'$, then $ S_P(\rr) = \pp_E^{-1} ( S_P'(\pp(\rr)))  $
	is a winning "strategy" for $P$ in $\G$. Here $ \pp_E^{-1} (e')$ is the only edge $e\in E$ in $"\mout"(\mtarget("\mlast"(\rr)))$ such that $\pp_E(e)=e'$.
	
	The equality $"\W_P"(\G')=\pp("\W_P"(\G))$ stems from the fact that if we choose a different initial vertex $v_1$ in $\G$, then $\pp$ is a "locally bijective morphism" to the game $\G'$ with initial vertex $\pp(v_1)$. Conversely, if we take a different initial vertex $v_1'$ in $\G'$, since $\pp$ is surjective we can take a vertex $v_1\in \pp^{-1}(v_1')$, and $\pp$ remains a locally bijective morphism between the resulting games.
	\end{proof}

	\subsection{The alternating cycle decomposition}\label{Subsection_ACD}
	 Most transformations of "Muller" into "parity" "transition systems" are based on the "composition" by some automaton converting the Muller condition into a parity one. These transformations act on the totality of the system uniformly, regardless of the local structure of the system and the "acceptance condition". 
	  The transformation we introduce in this section takes into account the interplay between the particular "acceptance condition" and the "transition system", inspired by the \emph{alternating chains} introduced in~\cite{wagner1979omega}.
	  
	  In the following we will consider "Muller transition systems" with the Muller acceptance condition using edges as colours. We can always suppose this, since given a transition system $\T$ with edges coloured by $\gamma: E \rightarrow C$ and a Muller condition $\F \subseteq \P(C)$, the condition $\widetilde{\F}\subseteq \P(E)$ defined as $A\in \widetilde{\F} \; \Leftrightarrow \gamma(A)\in \F$ is an "equivalent condition over" $\T$. However, the size of the representation of the condition $\F$ might change. Making this assumption corresponds to consider what are called \emph{explicit Muller conditions}. In particular, solving Muller games with explicit Muller conditions is in $\mathrm{PTIME}$~\cite{Horn2008Explicit}, while solving general Muller games is $\mathrm{PSPACE}$-complete~\cite{Dawar2005ComplexityBounds}.

	 \begin{definition}
	 	Given a "transition system" $\T=(V,E,\msource,\mtarget,I_0, \macc)$, a  \AP\intro{loop} is a subset of edges $l\subseteq E$ such that it exists $v\in V$ 
	 	 and a finite "run" $\rr\in \mrun_{T,v}$ such that $"\mfirst"(\rr)="\mlast"(\rr)=v$ and $\mocc(\rr)=l$. The set of "loops" of $\T$ is denoted $\mloop(\T)$.
	 	For a "loop" $l\in \mloop(\T)$ we write
	 	\[  ""\mathit{States}""(l):= \{ v\in V \; : \; \exists e\in l, \; \msource(e)=v \}.\]
	 \end{definition}
	
	Observe that there is a natural partial order in the set $\mloop(\T)$ given by set inclusion.

		\begin{remark*}
		If $l$ is a "loop" in $\mloop(\T)$, for every $q\in \mstate(l)$ there is a run $\rr \in \mrun_{\T,q}$ such that $"\mathit{Occ}"(\rr)=l$.
	\end{remark*}
	
	\begin{remark*}
		The maximal loops of $\mloop(\T)$ (for set inclusion) are disjoint and in one-to-one correspondence with the "strongly connected components" of $\T$.
	\end{remark*}
	
		\begin{definition}[Alternating cycle decomposition]\label{Def_ACD}
		Let $\T=(V,E,\msource,\mtarget,I_0, \F)$ be a "Muller transition system" with "acceptance condition" given by $\F\subseteq \P(E)$. The  \AP\intro{alternating cycle decomposition} (abbreviated ACD) of $\T$, noted $\ACD (T)$, is a family of "labelled trees" $(t_1, \nu_1),\dots, (t_r,\nu_r)$ with nodes labelled by "loops" in $\mloop(\T)$, $\nu_i: t_i\rightarrow \mloop(\T)$. We define it inductively as follows:
		\begin{itemize}
			\item Let $\{l_1,\dots, l_r\}$ be the set of maximal loops of $\mloop(\T)$. For each $i\in \{1,\dots, r\}$ we consider a "tree" $t_i$ and define $\nu_i(\varepsilon)=l_i$. 
			
			\item Given an already defined node $\tau$ of a tree $t_i$ we consider the maximal loops of the set 
			\[\{ l\subseteq \nu_i(\tau) \; : \; l\in \mloop(\T) \text{ and } l \in \F \; \Leftrightarrow \; \nu_i(\tau) \notin \F \}\]
			and for each of these loops $l$ we add a child to $\tau$ in $t_i$ labelled by $l$. 	
		\end{itemize}
	
		For notational convenience we add a special "tree" $(t_0,\nu_0)$ with a single node $\varepsilon$ labelled with the edges not appearing in any other tree of the forest, i.e., $\nu_0(\varepsilon)=E \setminus \bigcup_{i=1}^{r}l_i$ (remark that this is not a "loop").

		 We define $\mathit{States}(\nu_0(\varepsilon)):= V\setminus \bigcup_{i=1}^{r}\mathit{States}(l_i)$ (remark that this does not follow the general definition of $"\mathit{States}"()$ for loops).
		 
		 We call the trees $t_1,\dots, t_r$ the  \AP""proper trees"" of the "alternating cycle decomposition" of $\T$.   
		 Given a node $\tt$ of $t_i$, we note $\mstate_i(\tt):=\mstate(\nu_i(\tau))$.
		
	\end{definition}
	
	\begin{remark*}
		As for the "Zielonka tree", the "alternating cycle decomposition" of $\T$ is not unique, since it depends on the order in which we introduce the children of each node. This will not affect the upcoming results, and we will refer to it as ``the'' alternating cycle decomposition of $\T$. 
	\end{remark*}
	
	For the rest of the section we fix a "Muller transition system" $\T=(V,E,\msource,\mtarget, I_0, \F)$ with the "alternating cycle decomposition" given by $(t_0,\nu_0), (t_1,\nu_1),\dots, (t_r,\nu_r)$.
	
	\begin{remark*}
		The "Zielonka tree" for a "Muller condition" $\F$ over the set of colours $C$ can be seen as a special case of this construction, for the automaton with a single state, input alphabet $C$, a transition for each letter in $C$ and "acceptance condition" $\F$.
	\end{remark*}
	
	\begin{remark*}
		Each state and edge of $\T$ appears in exactly one of the "trees" of $\ACD (\T)$.
	\end{remark*}
	
	\begin{definition}
		The  \AP""index"" of a state $q\in V$  (resp. of an edge $e\in E$) in $\ACD (\T)$ is the only number $j\in \{0,1,\dots,r\}$ such that $q\in "\mathit{States}"_j(\varepsilon)$ (resp. $e \in \nu_j(\varepsilon)$).
	\end{definition}
	
	\begin{definition}
		For each state $q\in V$ of "index" $j$ we define the \AP""subtree associated to the state $q$"" as the "subtree" $t_q$ of $t_j$ consisting in the set of nodes $\{\tt \in t_j \; : \; q\in \mstate_j(\tt) \}$.
	\end{definition}
	
	We refer to Figures~\ref{Fig_ACD} and~\ref{Fig_Subtree_q4} for an example of $"t_q"$.
	

	\begin{definition}
		For each "proper tree" $t_i$ of $"\ACD" (\T)$ we say that $t_i$ is \intro(ACD){even} if $\nu_i(\varepsilon)\in \F$ and that it is \emph{\kl(ACD){odd}} if $\nu_i(\varepsilon)\notin \F$.
		
		We say that the "alternating cycle decomposition" of $\T$ is \emph{even} if all the trees of maximal "height" of $\ACD (\T)$ are even; that it is \emph{odd} if all of them are odd, and that it is \emph{\kl(ACD){ambiguous}} if there are even and odd trees of maximal "height".
	\end{definition}

	\begin{definition}
		For each $\tau \in t_i$, $i=1,\dots,r$, we define the \AP\intro(node){priority} of $\tt$ in $t_i$, written $p_i(\tau)$ as follows: 
		\begin{itemize}
			\item If $\ACD (\T)$ is \kl(ACD){even} or \kl(ACD){ambiguous}
			\begin{itemize}
				\item If $t_i$ is \kl(ACD){even} ($\nu_i(\varepsilon)\in \F$), then $p_i(\tau):=\mdepth(\tau)=|\tt|$.
				\item If $t_i$ is \kl(ACD){odd} ($\nu_i(\varepsilon)\notin \F$), then $p_i(\tau):=\mdepth(\tau)+1=|\tt|+1$.
			\end{itemize} 
			\item If $\ACD (\T)$ is \kl(ACD){odd}
			\begin{itemize}
				\item If $t_i$ is \kl(ACD){even} ($\nu_i(\varepsilon)\in \F$), then $p_i(\tau):=\mdepth(\tau)+2=|\tt|+2$.
				\item If $t_i$ is \kl(ACD){odd} ($\nu_i(\varepsilon)\notin \F$), then $p_i(\tau):=\mdepth(\tau)+1=|\tt|+1$.	
			\end{itemize} 
		\end{itemize}
	
		For $i=0$, we define $p_0(\varepsilon)=0$ if $"\ACD" (\T)$ is \kl(ACD){even} or \kl(ACD){ambiguous} and $p_0(\varepsilon)=1$ if $"\ACD" (\T)$ is \kl(ACD){odd}.
	\end{definition}
	
	The assignation of priorities to nodes produces a labelling of the levels of each tree. It will be used to determine the priorities needed by a parity "transition system" to simulate $\T$. The distinction between the cases $\ACD (\T)$ even or odd is added only to obtain the minimal number of priorities in every case.

	\begin{example}\label{Example_TransitionSystem}
		In Figure~\ref{Fig_TransitionSystem} we represent a "transition system" $\T=(V,E,\msource,\mtarget,q_0,\F)$ with $V=\{q_0,q_1,q_2,q_3,q_4,q_5\}$, $E=\{a,b,\dots,j,k\}$ and using the "Muller condition"
		\begin{flalign*}
			\F=\{\{c,d,e \},\{e \},
			\{ g,h,i \},\{l \},
			\{h,i,j,k \},\{j,k \}	\}.
		\end{flalign*}It has $2$ strongly connected components (with vertices $S_1=\{q_1,q_2\}, S_2=\{q_3,q_4,q_5\}$), and a vertex $q_0$ that does not belong to any strongly connected component.
		
		The "alternating cycle decomposition" of this transition system is shown in Figure~\ref{Fig_ACD}. It consists of two proper "trees", $t_1$ and $t_2$, corresponding to the strongly connected components of $\T$ and the tree $t_0$ that corresponds to the edges not appearing in the strongly connected components.
		
		We observe that $\ACD (\T)$ is \kl(ACD){odd} ($t_2$ is the highest tree, and it starts with a non-accepting "loop"). It is for this reason that we start labelling the levels of $t_1$ from $2$ (if we had assigned priorities $0,1$ to the nodes of $t_1$ we would have used $4$ priorities, when only $3$ are strictly necessary).
		
		In Figure~\ref{Fig_Subtree_q4} we show the "subtree associated to" $q_4$.
			\begin{figure}[ht]
			\centering 
			
			\begin{tikzpicture}[square/.style={regular polygon,regular polygon sides=4}, align=center,node distance=2cm,inner sep=2pt]
			
			\node at (0,1.5) [state, initial] (0) {$q_0$};
			\node at (3,1.5) [state] (1) {$q_1$};
			\node at (6,1.5) [state] (2) {$q_2$};
			\node at (0,0) [state] (3) {$q_3$};
			\node at (3,0) [state] (4) {$q_4$};
			\node at (6,0) [state] (5) {$q_5$};
			
			\path[->] 
			
			(0)  edge [] 	node[above] {$a$ }   (1)
			(0)  edge []  node[left] {$b$ }   (3)
			
			(1)  edge [in=160,out=20]  node[above] {$c$ }   (2)
			(1)  edge [] 	node[left] {$f$ }   (4)
			
			(2)  edge [in=-20,out=200]  node[above] {$d$ }   (1)
			(2)  edge [in=-30,out=30,loop] 	node[right] {$e$ }   (2)
			
			(3)  edge [in=210,out=150,loop]  node[left] {$g$ }   (3)
			(3)  edge [in=160,out=20]  node[above] {$h$ }   (4)
			
			(4)  edge [in=-20,out=200]  node[above] {$i$ }   (3)
			(4)  edge [in=160,out=20] 	node[above] {$j$ }   (5)
			
			(5)  edge [in=-20,out=200]  node[above] {$k$ }   (4)
			(5)  edge [in=-30,out=30,loop]  node[right] {$l$ }   (5);
			
			\end{tikzpicture}
			\caption{Transition system $\T$.}
			\label{Fig_TransitionSystem} 
		\end{figure}
		
		\begin{figure}[ht]
			\scalebox{1}{
			\hspace{-2mm}
				\begin{minipage}[b]{0.25\textwidth} 
				
				\begin{tikzpicture}[square/.style={regular polygon,regular polygon sides=4}, align=center,node distance=2cm,inner sep=3pt]
				
				\node at (0,2) [draw, diamond] (R) {a,b,f\\ $q_0$};
				
				\node at (1.5,2)  (p1) {1};
				\node at (0,0)  (pp) {};
				
				\node at (-1,1.5)  (neps) {$\textcolor{Violet2}{\langle \varepsilon \rangle}$};

				\end{tikzpicture}
				\caption*{\centering Tree $t_0$.}
			\end{minipage}
			\begin{minipage}[b]{0.22\textwidth} 
				
				\begin{tikzpicture}[square/.style={regular polygon,regular polygon sides=4}, align=center,node distance=2cm,inner sep=3pt]
				
				\node at (0,1.8) [draw, ellipse] (R) {c,d,e\\ $q_1,q_2$};
				
				\node at (0,0) [draw, rectangle] (0) {c,d \\ $q_1,q_2$};
				
				\node at (1,1.8)  (p2) {2};
				\node at (1,0)  (p3) {3};
				\node at (1,-1)  (pp) {};
				
				\node at (-1,1.5)  (neps) {$\textcolor{Violet2}{\langle \varepsilon \rangle}$};
				\node at (-0.9,-0.2)  (n0) {{\scriptsize$\textcolor{Violet2}{\langle 0 \rangle}$}};
				
				\draw   
				(R) edge (0);
				\end{tikzpicture}
				\caption*{\centering Tree $t_1$.}
			\end{minipage}
			\begin{minipage}[b]{0.4\textwidth} 
			\begin{tikzpicture}[square/.style={regular polygon,regular polygon sides=4}, align=center,node distance=2cm,inner sep=3pt]
			
			\node (rect) at (0,3.6) [draw,minimum width=1.2cm,minimum height=0.8cm] (R) {g,h,i,j,k,l \\ $q_3,q_4,q_5$};
			
			\node at (-2,1.8) [draw, ellipse] (0) {g,h,i \\ $q_3,q_4$};
			\node at (0,1.8) [draw, ellipse, minimum width=0.8cm] (1) {l \\ $q_5$};
			\node at (2,1.8) [draw, ellipse , minimum width=0.8cm] (2) {h,i,j,k \\ $q_3,q_4,q_5$};
			
			\node at (-2.5,0) [rectangle, draw, text width=0.8cm] (00) {g \\ $q_3$};
			\node at (-1,0) [rectangle, draw,minimum height=0.6cm] (01) {h,i \\ $q_3,q_4$};
			
			\node at (2,0) [rectangle, draw, text width=0.8cm] (20) {h,i \\ $q_3,q_4$};
			
			\node at (3.5,3.6)  (p1) {1};
			\node at (3.5,1.8)  (p2) {2};
			\node at (3.5,0) (p3) {3};
			
			\node at (-1.1,3.4)  (neps) {$\textcolor{Violet2}{\langle \varepsilon \rangle}$};
			\node at (-3,1.6)  (n0) {{\scriptsize$\textcolor{Violet2}{\langle 0 \rangle}$}};
			\node at (-0.6,1.6)  (n1) {{\scriptsize$\textcolor{Violet2}{\langle 1 \rangle}$}};
			\node at (0.9,1.3)  (n2) {{\scriptsize$\textcolor{Violet2}{\langle 2 \rangle}$}};
			
			\node at (-3.3,-0.1)  (n00) {{\scriptsize$\textcolor{Violet2}{\langle 0{,}0 \rangle}$}};
			\node at (-0.2,-0.1)  (n01) {{\scriptsize$\textcolor{Violet2}{\langle 0{,}1 \rangle}$}};
			
			\node at (1.1,-0.1)  (n20) {{\scriptsize$\textcolor{Violet2}{\langle 2{,}0 \rangle}$}};

			\draw   
			(R) edge (0)
			(R) edge (1)
			(R) edge (2)
			(0) edge (00)
			(0) edge (01)
			(2) edge (20);

			\end{tikzpicture}
			\caption*{\centering Tree $t_2$.}
          \end{minipage}
     }	
	\caption{"Alternating cycle decomposition" of $\T$. The priority assigned to the nodes of each level of the trees is indicated on the right. Nodes with an even priority are drawn as circles and those with an odd priority as rectangles (excepting the special node forming the root of $t_0$). Each node $\tt$ is labelled with $\nu_i(\tt)$ and with $\mstate_i(\tt)$. In \textcolor{Violet2}{violet} the names of the nodes.}
		\label{Fig_ACD} 
	\end{figure}

\begin{figure}[ht!]
	\centering 
	\hspace{-2mm}
	\scalebox{0.95}{
		\begin{tikzpicture}[square/.style={regular polygon,regular polygon sides=4}, align=center,node distance=2cm,inner sep=3pt]
		
		\node (rect) at (0,3.6) [draw,minimum width=1.3cm,minimum height=1.1cm] (R) {g,h,i,j,k,l};
		
		\node at (-2,1.8) [draw, ellipse,minimum height=1cm] (0) {g,h,i };
		\node at (2,1.8) [draw, ellipse , minimum width=0.8cm,minimum height=1cm] (2) {h,i,j,k};
		
		\node at (-2,0) [rectangle, draw,minimum width=1cm, minimum height=1cm] (01) {h,i };
		
		\node at (2,0) [rectangle, draw, minimum width=1cm, minimum height=1cm] (20) {h,i };
		
		\node at (3.5,3.6)  (p1) {1};
		\node at (3.5,1.8)  (p2) {2};
		\node at (3.5,0) (p3) {3};
		
		\node at (-1.1,3.4)  (neps) {$\textcolor{Violet2}{\langle \varepsilon \rangle}$};
		\node at (-3,1.6)  (n0) {{\scriptsize$\textcolor{Violet2}{\langle 0 \rangle}$}};
		\node at (1.1,1.4)  (n2) {{\scriptsize$\textcolor{Violet2}{\langle 2 \rangle}$}};
		
		\node at (-1,-0.1)  (n01) {{\scriptsize$\textcolor{Violet2}{\langle 0{,}1 \rangle}$}};
		
		\node at (1.1,-0.1)  (n20) {{\scriptsize$\textcolor{Violet2}{\langle 2{,}0 \rangle}$}};

		\draw   
		(R) edge (0)
		(R) edge (2)
		(0) edge (01)
		(2) edge (20);
		
		\end{tikzpicture}
	}
	\caption{"Subtree associated to" $q_4$, noted $t_{q_4}$.}
	\label{Fig_Subtree_q4} 
\end{figure}

\end{example}
	

	\subsection{The alternating cycle decomposition transformation}\label{Subsection_ACD-transformation}
	
	We proceed to show how to use the "alternating cycle decomposition" of a "Muller transition system" to obtain a "parity transition system". Let $\T=(V,E,\msource,\mtarget,I_0, \F)$ be a "Muller transition system" and $(t_0,\nu_0), (t_1, \nu_1),\dots, (t_r,\nu_r)$, its "alternating cycle decomposition". 
	
	First, we adapt the definitions of $\mathit{Supp}$ and $\mathit{Nextbranch}$ to the setting with multiple trees.

	\begin{definition}
		For an edge $e\in E$ such that $\mtarget(e)$ has "index" $j$, for $i\in \{0,1,\dots,r\}$  and a branch $\bb$ in some subtree of $t_i$, we define the  \AP""support"" of $e$ from $\tau$ as: 
		\[ \msupp(\bb,i,e)=
		\begin{cases}

		\text{The maximal node (for } \prefix \text{) } \tt\in \bb \text{ such that } e\in \nu_i(\tt),  \text{ if } i= j .\\[2mm]
		
		\text{The root } \varepsilon \text{ of } t_j,  \text{ if } i\neq j.
		
		\end{cases}  \]

	\end{definition}
	Intuitively, $\msupp(\bb,i,e)$ is the highest node we visit if we want to go from the bottom of the branch $\bb$ to a node of the tree that contains $e$ ``in an optimal trajectory'' (going up as little as possible). If we have to jump to another tree, we define $\msupp(\bb,i,e)$ as the root of the destination tree.

	\begin{definition}
		Let $i\in \{0,1,\dots,r\}$, $q$ be a state of "index" $i$, $\bb$ be a branch of some "subtree" of $t_i$ and $\tt\in \bb$ be a node of $t_i$ such that $q\in \mstate_i(\tt)$. If $\tt$ is not the deepest node of $\bb$, let $\ss_\bb$ be the unique child of $\tt$ in $t_i$ such that $\ss_\bb \in \bb$. We define:

		\[ \intro{\mathit{Nextchild}_{t_q}}(\bb,\tt)=
		\begin{cases}
		\tt , \text{ if } \tt \text{ is a leaf in } \kl{t_q}.\\[3mm]
		\parbox{8cm}{Smallest older sibling of $ \ss_\bb $ in $ \kl{t_q}, $ if $ \ss_\bb $ is defined and there is any such older sibling.}\\[3mm]
		\text{Smallest child of } \tt \text{ in } \kl{t_q} \text{ in any other case}.
			\end{cases} \]
		
	\end{definition}
	
	\begin{definition}
		Let $i\in \{0,1,\dots,r\}$ and $\bb$ be a branch of some "subtree" of $t_i$. For a state $q$ of "index" $j$ and a node $\tt$ such that $q\in \mstate_j(\tt)$ and such that $\tt\in \bb$ if $i=j$, we define:
		
		\[ \intro{\mathit{Nextbranch}_{t_q}}(\bb,i,\tt)= 
		\begin{cases}

		\text{ Leftmost branch in } \kl{t_q} \text{ below } \kl{\mathit{Nextchild}_{t_q}}(\bb,\tt),  \text{ if } i= j .\\[3mm]
		
		\text{The leftmost branch in  } \msubtree_{t_q}(\tt),  \text{ if } i\neq j.
		
		\end{cases}  \]

	\end{definition}
	
	
	\begin{definition}[ACD-transformation]\label{Def_TransformationMullerAutomata} 
		Let  $\T=(V,E,\msource,\mtarget,I_0, \F)$ be a "Muller transition system" with "alternating cycle decomposition" $\ACD (\T)= \{(t_0,\nu_0),\ab (t_1,\nu_1),\dots,(t_r,\nu_r)\}$. We define its  \AP""ACD-parity transition system"" (or \emph{ACD-transformation}) $\kl{\P_{\mathcal{ACD}(\T)}}=\ab (V_P,E_P,\ab \msource_P,\mtarget_P,I_0', p:E_P\rightarrow \NN)$ as follows: 
		
		\begin{itemize}
			\item $V_P=\{ (q,i,\bb ) \; : \;  q\in V \text{ of "index" } i \text{ and } \bb\in \mbranch(\kl{t_q}) \} $.
			
			\item For each node $(q,i,\bb)\in V_P$ and each edge $e\in \kl{\mout}(q)$ we define an edge $e_{i,\bb}$ from $(q,i,\bb )$. We set 
			\begin{itemize}
				\item  $\msource_P(e_{i,\bb})=(q,i,\bb)$, where $q=\msource(e)$.
				
				\item  $\mtarget_P(e_{i,\bb})=(q',k,\kl{\mathit{Nextbranch}_{t_{q'}}}(\bb,i,\tt))$, where $q'=\mtarget(e)$, $k$ is its "index" and $\tt=\msupp(\bb,i,e)$.
			
				\item $p(e_{i,\tt})=\kl(node){p_j}(\msupp(\bb,i,e))$, where $j$ is the "index" of $\msupp(\bb,i,e)$.
			\end{itemize} 	
				\item $I_0'=\{(q_0,i,\bb_0) \; : \; q_0\in I_0, \, i \text{ the index of } q_0$ and $\bb_0$ the leftmost branch in $\kl{t_{q_0}}\}$.
			\end{itemize} 
		If $\T$ is labelled by $l_V:V\rightarrow L_V$, $l_E:E\rightarrow L_E$, we label $\kl{\P_{\mathcal{ACD}(\T)}}$ by $l_V'((q,i,\bb))=l_V(q)$ and $l_E'(e_{i,\bb})=l_E(e)$.
	\end{definition}
	
	The set of states of $\P_{\mathcal {ACD}(\T )}$ is build as follows: for each state $q\in \T$ we consider the subtree of $\ACD(\T)$ consisting of the nodes with $q$ in its label, and we add a state for each branch of this subtree.
	
	Intuitively, to define transitions in the transition system $\kl{\P_{\mathcal{ACD}(\T)}}$ we move simultaneously in $\T$ and in $\ACD(\T)$. We start from $q_0\in I_0$ and from the leftmost branch of $\kl{t_{q_0}}$. When we take a transition $e$ in $\T$ while being in a branch $\bb$, we climb the branch $\bb$ searching a node $\tt$ with $q'=\mtarget(e)$ and $e$ in its label, and we produce the priority corresponding to the level reached. If no such node exists, we jump to the root of the tree corresponding to $q'$. Then, we move to the next child of $\tt$ on the right of $\bb$ in the tree $\kl{t_{q'}}$, and we pick the leftmost branch under it in $\kl{t_{q'}}$. If we had jumped to the root of $\kl{t_{q'}}$ from a different tree, we pick the leftmost branch of $\kl{t_{q'}}$. 
	
	\begin{remark*}
		The size of $\kl{\P_{\mathcal{ACD}(\T)}}$ is 
		\[ |\P_{\mathcal {ACD}(\T )}|=\sum\limits_{q\in V} |\mbranch(\kl{t_q})|. \]
		The number of priorities used by $\kl{\P_{\mathcal{ACD}(\T)}}$ is the "height" of a maximal tree of $\ACD(\T)$ if $\ACD(\T)$ is \kl(ACD){even} or \kl(ACD){odd}, and the "height" of a maximal tree plus one if $\ACD(\T)$ is \kl(ACD){ambiguous}.
	\end{remark*}
	
\begin{example}\label{Example_ParityTransitionSystem}
	In Figure~\ref{Fig_ParityTransitionSystem} we show the "ACD-parity transition system" $\kl{\P_{\mathcal{ACD}(\T)}}$ of the transition system of Example~\ref{Example_TransitionSystem} (Figure~\ref{Fig_TransitionSystem}). States are labelled with the corresponding state $q_j$ in $\T$, the tree of its "index" and a node $\tt\in t_i$ that is a leaf in $\kl{t_{q_j}}$ (defining a branch of it).
	
	We have tagged the edges of $\kl{\P_{\mathcal{ACD}(\T)}}$ with the names of edges of $\T$ (even if it is not an automaton). These indicate the image of the edges by the "morphism" $\pp: \P_{\mathcal{ACD}(\T)} \rightarrow \T$, and make clear the bijection between "runs" in $\T$ and in $\kl{\P_{\mathcal{ACD}(\T)}}$.
	
	In this example, we create one ``copy'' of states $q_0,q_1$ and $q_2$, three ``copies'' of the state $q_3$ and two``copies'' of states $q_4$ and $q_5$. The resulting "parity transition system" $\kl{\P_{\mathcal{ACD}(\T)}}$ has therefore $10$ states.
	
	\begin{figure}[ht]
		\centering 
		\begin{tikzpicture}[square/.style={regular polygon,regular polygon sides=4}, align=center,node distance=2cm,inner sep=2pt]
		
		\node[scale=0.8] at (0,1.5) [state, initial] (0) { $q_0,t_0$\\$\langle \varepsilon \rangle$};
		\node[scale=0.8] at (3,1.5) [state] (1) {$q_1,t_1$\\$\langle 0\rangle $};
		\node[scale=0.8] at (6,1.5) [state] (2) {$q_2,t_1$\\$\langle 0\rangle $};
		
		\node[scale=0.8] at (0,0) [state] (3n00) { $q_3,t_2$\\ $\langle 0{,}0\rangle $};
		\node[scale=0.8] at (1.5,-1.3) [state] (3n01) {$q_3,t_2$\\$\langle 0{,}1\rangle $};
		\node[scale=0.8] at (0,-2.5) [state] (3n20) {$q_3,t_2$\\$\langle 2{,}0\rangle $};
		
		\node[scale=0.8] at (3,0) [state] (4n01) {$q_4,t_2$\\$\langle 0{,}1\rangle $};
		\node[scale=0.8] at (3,-2.5) [state] (4n20) {$q_4,t_2$\\$\langle 2{,}0\rangle $};
		
		\node[scale=0.8] at (6,0) [state] (5n1) {$q_5,t_2$\\$\langle 1\rangle $};
		\node[scale=0.8] at (6,-2.5) [state] (5n2) {$q_5,t_2$\\$\langle 2\rangle $};
		\path[->] 
		
		(0)  edge [] 	node[above] {$a: \textcolor{Green2}{2}$ }   (1)
		(0)  edge []  node[left] {$b: \textcolor{Green2}{2}$ }   (3n00)
		
		(1)  edge [in=160,out=20]  node[above] {$c: \textcolor{Green2}{3}$ }   (2)
		(1)  edge [] 	node[left] {$f: \textcolor{Green2}{1}$ }   (4n01)
		
		(2)  edge [in=-20,out=200]  node[above] {$d: \textcolor{Green2}{3}$ }   (1)
		(2)  edge [in=-30,out=30,loop] 	node[right] {$e: \textcolor{Green2}{2}$ }   (2)
		
		(3n00)  edge [in=210,out=150,loop]  node[left] {$g: \textcolor{Green2}{3}$ }   (3n00)
		(3n00)  edge [in=160,out=20,pos=0.4]  node[above] {$h: \textcolor{Green2}{2}$ }   (4n01)
		
		(3n01)  edge []  node[] {$g: \textcolor{Green2}{2}$ }   (3n00)
		(3n01)  edge [out=20,in=250]  node[right, pos=0.3] {$\;\;h: \textcolor{Green2}{3}$ }   (4n01)
		
		(3n20)  edge []  node[left] {$g: \textcolor{Green2}{1}$ }   (3n00)
		(3n20)  edge [in=170,out=10]  node[above] {$h: \textcolor{Green2}{3}$ }   (4n20)
		
		(4n01)  edge [out=200,in=70]  node[above] {$i: \textcolor{Green2}{3}$ }   (3n01)
		(4n01)  edge [] 	node[above] {$j: \textcolor{Green2}{1}$ }   (5n1)
		
		(4n20)  edge [in=-20,out=200]  node[below] {$i: \textcolor{Green2}{3}$ }   (3n20)
		(4n20)  edge [out=20,in=160] 	node[above,pos=0.5] {$j: \textcolor{Green2}{2}$ }   (5n2)
		
		(5n2)  edge [out=200,in=-20]  node[below] {$k: \textcolor{Green2}{2}$ }   (4n20)
		(5n2)  edge []  node[right] {$l: \textcolor{Green2}{1}$ }   (5n1)
		
		(5n1)  edge []  node[left,pos=0.4] {$k: \textcolor{Green2}{1}$ }   (4n20)
		(5n1)  edge [in=-30,out=30,loop]  node[right] {$l: \textcolor{Green2}{2}$ }   (5n1);
		
		\end{tikzpicture}
		\caption{Transition system $\kl{\P_{\mathcal{ACD}(\T)}}$.}
		\label{Fig_ParityTransitionSystem} 
	\end{figure}
\end{example}

\begin{example}
	
	Let $\A$ be the "Muller automaton" of Example~\ref{Example_AutomataForL}. Its "alternating cycle decomposition" has a single tree that coincides with the "Zielonka tree" of its "Muller acceptance condition" $\F_1$ (shown in Figure~\ref{Fig_ZielonkaTree1}). However, its "ACD-parity transition system" has only $3$ states, less than the "composition" $\Z_{\F_1} \lhd \A $ (Figure~\ref{Fig_Product_AxZF}), as shown in Figure~\ref{Fig_ACDTransformationMullerAutomA}.

\begin{figure}[ht]
	\centering
	\begin{minipage}[b]{0.3\textwidth} 
		\begin{tikzpicture}[square/.style={regular polygon,regular polygon sides=4}, align=center,node distance=2cm,inner sep=2pt]
		
		\node at (0,2) [state, initial] (0) {A};
		\node at (2,2) [state] (1) {B};
		
		\path[->] 
		(0)  edge [in=70,out=110,loop] 	node[above] {$0:\textcolor{Green2}{a}$ }   (0)
		(0)  edge [in=150,out=30] 	node[above] {$1 : \textcolor{Green2}{b}$ }   (1)
		
		(1)  edge [in=-30,out=210]  node[below] {$1: \textcolor{Green2}{b}$ }   (0)
		(1)  edge [in=70,out=110,loop] 	node[above] {$0 : \textcolor{Green2}{c}$ }   (1);

		\end{tikzpicture}
		\caption*{ Muller automaton $\A$.\\  $\F_1=\{\{a\},\{b\}\}$.}
	\end{minipage}
	\hspace{0mm}
		\begin{minipage}[b]{0.32\textwidth} 
		
		\begin{tikzpicture}[square/.style={regular polygon,regular polygon sides=4}, align=center,node distance=2cm,inner sep=3pt]
		
		\node at (0,1.5) [draw, rectangle, text height=0.3cm, text width=1cm] (R) { a,b,c \\ A,B};
		
		\node at (-1,0) [draw, ellipse,text height=0.1cm, text width=0.8cm] (0) {a\\ A};
		\node at (1,0) [draw, ellipse,text height=0.1cm, text width=0.8cm] (1) {b\\ A,B};
		
		\node at (2,1.5)  (p2) {1};
		\node at (2,0)  (p3) {2};
		\draw   
		(R) edge (0)
		(R) edge (1);
		\end{tikzpicture}
		\caption*{\centering $\ACD(\A)$.}
	\end{minipage}
\begin{minipage}[b]{0.33\textwidth} 
	\begin{tikzpicture}[square/.style={regular polygon,regular polygon sides=4}, align=center,node distance=2cm,inner sep=2pt]
	
	\node at (0,2) [state, initial] (A1) {A,0};
	\node at (0,0) [state] (A2) {A,1};
	\node at (2,1) [state] (B) {B,1};

	\path[->] 
	(A1)  edge [in=70,out=110,loop] 	node[above] {$0 :\textcolor{Green2}{2}$ }   (A1)
	(A1)  edge 	node[ anchor=south west, pos=0.2] {$1:\textcolor{Green2}{1}$ }   (B)

	(B)  edge[in=70,out=110,loop]  node[above] {$0:\textcolor{Green2}{1}$ }   (B)
	(B)  edge [in=10,out=230]	node[below right] {$1:\textcolor{Green2}{2}$ }   (A2)
	
	(A2)  edge [in=190,out=50]  node[above] {$1:\textcolor{Green2}{2}$ }   (B)
	(A2)  edge []  node[left] {$0:\textcolor{Green2}{1}$ }  (A1);
	
	\end{tikzpicture}
	\caption*{\centering $\P_{\mathcal{ACD}(\A)}$.}
\end{minipage}
	\caption{Muller automaton $\A$, its "alternating cycle decomposition" and its "ACD-transformation".}
	\label{Fig_ACDTransformationMullerAutomA}
\end{figure}
	
\end{example}
%
	\begin{proposition}[Correctness]\label{Prop_Correctness-ACD}
		Let  $\T=(V, E, \msource, \mtarget, I_0,  \F)$ be a "Muller transition system" and $\kl{\P_{\mathcal{ACD}(\T)}}=(V_P, E_P, \msource_P, \mtarget_P, I_0',  p:E_P\rightarrow \NN)$ its "ACD-transition system". Then, there exists a "locally bijective morphism" $ \pp : \kl{\P_{\mathcal{ACD}(\T)}} \rightarrow \T$.
		Moreover, if $\T$ is a "labelled transition system", then  $\pp$ is a "morphism of labelled transition systems".
	\end{proposition}
	\begin{proof}
	We define $\pp_V : V_P \rightarrow V$ by $\pp_V((q,i,\bb))=q$ and $\pp_E : E_P \rightarrow E$ by $\pp_E(e_{i,\tt})=e$.
	It is clear that this map preserves edges, initial states and labels. It is also clear that it is "locally bijective", since we have defined one initial state in $\kl{\P_{\mathcal{ACD}(\T)}}$ for each initial state in $\T$, and by definition the edges in $\kl{\mout}((q,i,\bb))$ are in bijection with $\mout(q)$. It induces therefore a bijection between the runs of the transition systems
	 (Fact~\ref{Fact_LocBijMorph_BijectionRuns}).
		
		Let us see that a "run" $\rr$ in $\T$ is accepted if and only if $\pp^{-1}(\rr)$ is accepted in $\P_{\mathcal {ACD}(\T )}$. First, we remark that any infinite run $\rr$ of $\T$ will eventually stay in a "loop" $l\in \mloop(\T)$ such that $\minf(\rr)=l$, and therefore we will eventually only visit states corresponding to the tree $t_i$ such that $l\subseteq \nu_i(\varepsilon)$ in the "alternating cycle decomposition". Let $p_{\min}$ be the smallest priority produced infinitely often in the run $\pp^{-1}(\rr)$ in $\P_{\mathcal {ACD}(\T )}$. As in the proof of Proposition~\ref{Prop_Correctness-Zielonka}, there is a unique node $\tt_p$ in $t_i$ visited infinitely often such that $\kl(node){p_i}(\tt_p)=p_{\min}$. Moreover, the states visited infinitely often in $\P_{\mathcal {ACD}(\T )}$ correspond to branches below $\tt_p$, that is, they are of the form
		$ (q,i,\bb)$, with $\bb \in \msubtree_{t_q}(\tt_p), \text{ for } q\in \mstate_i(\tt_p)$.
		 We claim that $\tt_p$ verifies:
		 
		\begin{itemize}
			\item $l\subseteq \nu_i(\tt_p)$.
			\item $l\nsubseteq \nu_i(\sigma)$ for every "child" $\sigma$ of $\tt_p$.
		\end{itemize}
	
	By definition of $\ACD(\T)$ this implies
	$$ l\in \F \; \Longleftrightarrow \; \nu_i(\tt_p)\in \F \; \Leftrightarrow \; p_{\min} \text{ is even.}$$
	
	We show that $l\subseteq \nu_i(\tt_p)$. For every edge $e\notin \nu_i(\tt_p)$ of "index" $i$ and for every branch $\bb \in \msubtree_{t_q}(\tt_p), \text{ for } q\in \mstate_i(\tt_p)$, we have that $\tt'=\msupp(\bb,i,e)$ is a strict "ancestor" of $\tt_p$ in $t_i$. Therefore, if $l$ was not contained in $\nu_i(\tt_p)$ we would produce infinitely often priorities strictly smaller than $p_{\min}$.
	
	Finally, we show that $l\nsubseteq \nu_i(\sigma)$ for every "child" $\sigma$ of $\tt_p$. Since we reach $\tt_p$ infinitely often, we take transitions $e_{i,\bb}$ such that $\tt_p=\msupp(\bb,i,e)$ infinitely often. Let us reason by contradiction and let us suppose that there is some child $\ss$ of $\tt_p$ such that $l\subseteq \nu_i(\sigma)$. Then for each edge $e\in l$, $\mtarget(e)\in \mstate_i(\ss)$, and therefore $\ss\in t_q$ for all $q\in \mstate(l)$ and for each transition $e_{i,\bb}$ such that $\tt_p=\msupp(\bb,i,e)$, some branches passing through $\ss$ are considered as destinations. Eventually, we will go to some state $(q,i,\bb')$, for some branch $\bb'\in \msubtree_{t_q}(\ss)$. But since $l\subseteq \nu_i(\sigma)$, then for every edge $e\in l$ and branch $\bb'\in \msubtree_{t_q}(\ss)$ it is verified that $\msupp(\bb',i,e)$ is a "descendant" of $\ss$, so we would not visit again $\tt_p$ and all priorities produced infinitely often would be strictly greater than $p_{\min}$.  
	\end{proof}

	From the remarks at the end of Section~\ref{Section_LocallyBijectiveMorphism}, we obtain:
	\begin{corollary}
		If $\A$ is a "Muller automaton" over $\Sigma$, the automaton $\kl{\P_{\mathcal{ACD}(\A)}}$ is a "parity automaton" recognising $\L(\A)$. Moreover,
		\begin{itemize}
			\item $\A$ is "deterministic" if and only if $\P_{\mathcal{ACD}(\A)}$ is deterministic.
			\item $\A$ is "unambiguous" if and only if $\P_{\mathcal{ACD}(\A)}$ is unambiguous.
			\item $\A$ is "GFG" if and only if $\P_{\mathcal{ACD}(\A)}$ is GFG.
		\end{itemize} 
	\end{corollary}

	\begin{corollary}
		If $\G$ is a "Muller game", then $\P_{\mathcal {ACD}(\G )}$ is a "parity game" that has the same winner than $\G$.
		The "winning region" of $\G$ for a player $P\in \{Eve, Adam\}$ is $\kl{\W_P}(\G)=\pp("\W_P"(\P_{\mathcal {ACD}(\G )}))$, being $\pp$ the morphism of the proof of Proposition~\ref{Prop_Correctness-ACD}.
	\end{corollary}

	\subsection{Optimality of the alternating cycle decomposition transformation}\label{Section_OptimalityACD}
	
	In this section we prove the strong optimality of the "alternating cycle decomposition transformation", both for number of priorities (Proposition~\ref{Prop_OptimalityACD_Priorities}) and for size (Theorem~\ref{Th_OptimalityACDTransformation}). We use the same ideas as for proving the optimality of the "Zielonka tree automaton" in Section~\ref{Section_OptimalityZielonkaTree}.

	\begin{proposition}[Optimality of the number of priorities]\label{Prop_OptimalityACD_Priorities}
		Let $\T$ be a "Muller transition system" such that all its states are "accessible" and let $\P_{\mathcal{ACD}(\T)}$ be its "ACD-transition system". If $\P$ is another "parity transition system" such that there is a "locally bijective morphism" $\pp:\P \rightarrow \T$, then $\P$ uses at least the same number of priorities than $\kl{\P_{\mathcal{ACD}(\T)}}$.
	\end{proposition}
	\begin{proof}
		We distinguish 3 cases depending on whether $\ACD(\T)$ is \kl(ACD){even}, \kl(ACD){odd} or \kl(ACD){ambiguous}.
		
		We treat simultaneously the cases $\ACD(\T)$  \kl(ACD){even} and \kl(ACD){odd}. In these cases, the number $h$ of priorities used by $\kl{\P_{\mathcal{ACD}(\T)}}$ coincides with the maximal "height" of a tree in $\ACD(\T)$. Let $t_i$ be a tree of maximal "height" $h$ in $\ACD(\T)$, $\bb=\{\tt_1,\dots,\tt_{h}\}\in \mbranch(t_i)$ a branch of $t_i$ of maximal length (ordered as $\tt_1 \kl{\sqsupseteq} \tt_2 \kl{\sqsupseteq} \dots \kl{\sqsupseteq} \tt_{h}=\varepsilon$)  and $l_j=\nu_i(\tt_j)$, $j=1,\dots, h$. We fix $q\in \mstate_i(\tt_1)$, where $\tt_1$ is the leaf of $\bb$, and we write \[\mathit{Loop}_\T(q)=\{w\in \mrun_{T,q}\cap E^* \; : \; \kl{\mfirst}(w)=\kl{\mlast}(w)=q \},\] and for each $j=1,\dots, h$ we choose $w_j \in \mathit{Loop}_\T(q)$ such that $\mocc(w_j)=l_j$. Let $\eta'$ be the maximal priority appearing in $\P$.  We show as in the proof of Proposition~\ref{Prop_OptimalPrioritiesZielonka} that for every $v\in \mathit{Loop}_\T(q)$, the "run" $\pp^{-1}((w_1\dots w_k v)^\oo)$ must produce a priority smaller or equal to $\eta'-k+1$. Taking $k=h$, the "run" $\pp^{-1}((w_1\dots w_h)^\oo)$ produces a priority smaller or equal to $\eta'-h+1$ and even if and only if $\ACD(\T)$ is \kl(ACD){even}. By Lemma~\ref{Lemma_PrioritiesInRange} we can suppose that $\P$ uses all priorities in $[\eta'-h+1, \eta']$. We conclude that $\P$ uses at least $h$ priorities, so at least as many as $\kl{\P_{\mathcal{ACD}(\T)}}$.

		In the case $\ACD(\T)$ \kl(ACD){ambiguous}, if $h$ is the maximal "height" of a tree in $\ACD(\T)$, then $\kl{\P_{\mathcal{ACD}(\T)}}$ uses $h+1$ priorities. We can repeat the previous argument with two different maximal branches of respective maximal \kl(ACD){even} and \kl(ACD){odd} trees. We conclude that $\P$ uses at least priorities in a range $[\mu,\mu+h]\cup [\eta,\eta+h]$, with $\mu$ even and $\eta$ odd, so it uses at least $h+1$ priorities.
	\end{proof}
	
	A similar proof, or an application of the results from~\cite{niwinskiwalukievicz1998Relating} gives the following result:
	\begin{proposition}\label{Prop_OptimalityACD_Priorities-Languages}
		If $\A$ is a deterministic automaton, the accessible part of $\kl{\P_{\mathcal{ACD}(\A)}}$ uses the optimal number of priorities to recognise $\L(\A)$.
	\end{proposition}

	Finally, we state and prove the optimality of $\kl{\P_{\mathcal {ACD}(\A )}}$ for size.
	
		\begin{theorem}[Optimality of the number of states]\label{Th_OptimalityACDTransformation}
		Let $\T$ be a (possibly "labelled") "Muller transition system" such that all its states are "accessible" and let $\P_{\mathcal{ACD}(\T)}$ be its "ACD-transition system". If $\P$ is another "parity transition system" such that there is a "locally bijective morphism" $\pp:\P \rightarrow \T$, then
		$ |"\P_{\mathcal{ACD}(\T)}"|\leq |\P| $.
	\end{theorem}
	

	\paragraph*{Proof of Theorem~\ref{Th_OptimalityACDTransformation}}
	We follow the same steps as for proving Theorem~\ref{Th_OptimalityZielonkaTree}. We will suppose that all states of the "transition systems" considered are "accessible".
	

	\begin{definition}\label{Def_l-SCC}
		Let $\T_1$, $\T_2$ be "transition systems" such that there is a "morphism of transition systems" $\varphi: \T_1 \rightarrow \T_2$. Let $l\in \kl{\mathpzc{Loop}}(\T_2)$ be a "loop" in $\T_2$. An  \AP""$l$-SCC"" of $\T_1$ (with respect to $\pp$) is a non-empty "strongly connected subgraph" $(V_l,E_l)$ of the subgraph $(\pp_V^{-1}(\mstate(l)),\pp_E^{-1}(l) )$ such that 
		\begin{align}\label{Eq_l-SCC_property-star}
			 \nonumber& \text{for every } q_1\in V_l \text{ and every } e_2\in "\mout"(\pp(q_1))\cap l\\ 
			 \tag{$\star$} &\text{there is an edge } e_1\in \pp^{-1}(e_2)\cap "\mout"(q_1) \text{ such that } e_1\in E_l. 
		\end{align}
	\end{definition}
	
	That is, an $l$-SCC is a "strongly connected subgraph" of $\T_1$ in which all states and transitions correspond via $\pp$ to states and transitions appearing in the "loop" $l$. Moreover, given a "run" staying in $l$ in $\T_2$ we can simulate it in the $l$-SCC of $\T_1$ (property \eqref{Eq_l-SCC_property-star}).
%
	
	\begin{lemma}\label{Lemma_ExistanceSCC_SurjMorph}	
		Let $\T_1$
		 and $ \T_2$
		  be two "transition systems" such that there is a "locally surjective" "morphism" $\pp: \T_1 \rightarrow \T_2$. Let $l\in "\mathpzc{Loop}"(\T_2)$ and $C_l=(V_l,E_l)$ be a non-empty "$l$-SCC" in $\T_1$. Then, for every "loop" $l'\in "\mathpzc{Loop}"(\T_2)$ such that $l'\subseteq l$ there is a non-empty "$l'$-SCC" in $C_l$.
	\end{lemma}

	\begin{proof}
		Let $(V',E')=(V_l,E_l)\cap (\pp_V^{-1}(\mstate(l')),\pp_E^{-1}(l'))$. We first prove that $(V',E')$ is non-empty. Let $q_1\in V_l \subseteq \mstate(l)$. Let $\rr \in \mrun_{T_2,\pp(q)}$ be a finite run in $\T_1$ from $\pp(q_1)$, visiting only edges in $l$ and ending in $q_2\in \mstate(l')$. From the local surjectivity, we can obtain a run in $\pp^{-1}(\rr)$ that will stay in $(V',E')$ and that will end in a state in $\pp_V^{-1}(\mstate(l'))$. The subgraph $(V',E')$ clearly has property \eqref{Eq_l-SCC_property-star} (for $l'$).
		
		We prove by induction on the size that any non-empty subgraph $(V',E')$ verifying the property \eqref{Eq_l-SCC_property-star} (for $l'$) admits an $l'$-SCC. If $|V'|=1$, then $(V',E')$ forms by itself a "strongly connected graph". If $|V'|>1$ and $(V',E')$ is not strongly connected, then there are vertices $q,q'\in V'$ such that there is no path from $q$ to $q'$ following edges in $E'$. We let
		\[ V'_q=\{p\in V' \; : \; \text{there is a path from } q \text{ to } p \text{ in  } (V',E')\} \; ; \; E'_q=E'\cap "\mout"(V'_q)\cap "\mIn"(V'_q) .\]
		
		Since $q'\notin V'_q$, the size $|V'_q|$ is strictly smaller than $|V'|$. 
		Also, the subgraph $(V'_q,E'_q)$ is non-empty since $q\in V'_q$.
		The property \eqref{Eq_l-SCC_property-star} holds from the definition of $(V'_q,E'_q)$. We conclude by induction hypothesis.
		\end{proof}

	\begin{lemma}\label{Lemma_Disjoint_SCC_BijMorphism}
		Let $\T$ be a "Muller transition system" with acceptance condition $\F$ and let $\P$ be a "parity transition system" such that there is a "locally bijective morphism" $\pp: \P \rightarrow \T$. Let $t_i$ be a "proper tree" of $\ACD(\T)$ and $\tt,\ss_1,\ss_2\in t_i$ nodes in $t_i$ such that $\ss_1,\ss_2$ are different "children" of $\tt$, and let $l_1=\nu_i(\ss_1)$ and $l_2=\nu_i(\ss_2)$. If $C_1$ and $C_2$ are two "$l_1 $-SCC" and "$l_2$-SCC" in $\P$, respectively, then $C_1\cap C_2= \emptyset$. 
	\end{lemma}
	
	\begin{proof}
		Suppose there is a state $q\in C_1\cap C_2$. Since $\pp_V(q)\in \mstate(l_1)\cap \mstate(l_2)$, and $l_1, l_2$ are "loops" there are finite "runs" $\rr_1,\rr_2 \in \mrun_{\T,\pp_V(q)}$ such that $\mocc(\rr_1)=l_1$ and $"\mathit{Occ}"(\rr_2)=l_2$. We can ``simulate'' these runs in $C_1$ and $C_2$ thanks to property \eqref{Eq_l-SCC_property-star}, producing runs $\pp^{-1}(\rr_1)$ and $\pp^{-1}(\rr_2)$ in $\mrun_{\P,q}$ and arriving to $q_1="\mlast"(\pp^{-1}(\rr_1))$ and $q_2=\mlast(\pp^{-1}(\rr_1))$. 
		Since $C_1, C_2$ are "$l_1,l_2$-SCC", there are finite runs $w_1\in \mrun_{\P,q_1}$, $w_2\in \mrun_{\T,q_2}$ such that $\mathit{Last}(w_1)=\mathit{Last}(w_2)=q$, so the runs $\pp^{-1}(\rr_1)w_1$ and $\pp^{-1}(\rr_2)w_2$ start and end in $q$. We remark that in $\T$ the runs  $\pp(\pp^{-1}(\rr_1)w_1)=\rr_1\pp_E(w_1)$ and $\pp(\pp^{-1}(\rr_2)w_2)=\rr_2\pp_E(w_2)$
		start and end in $\pp_V(q)$ and
		visit, respectively, all the edges in $l_1$ and $l_2$. 
		From the definition of $\ACD(\T)$ we have that $l_1\in \F \; \Leftrightarrow \; l_2 \in \F \; \Leftrightarrow \; l_1 \cup l_2 \notin \F$. Since $\pp$ preserves the "acceptance condition", the minimal priority produced by $\pp^{-1}(\rr_1)w_1$ has the same parity than that of $\pp^{-1}(\rr_2)w_2$, but concatenating both runs we must produce a minimal priority of the opposite parity, arriving to a contradiction.
	\end{proof}
	
	\begin{definition}
		 Let $\T$ be a "Muller transition system" and $"\P_{\mathcal{ACD}(\T)}"$	
		its "ACD-parity transition system". For each tree $t_i$ of $\ACD(\T)$, each node $\tt\in t_i$ and each state $q\in \mstate_i(\tt)$ we write:
		\[ \psi_{\tt,i,q}=|\mbranch(\msubtree_{t_q}(\tt))|=|\{ (q,i,\bb)\in "\P_{\mathcal{ACD}(\T)}" \; : \; \bb  \, \text{ passes through } \tt \}|.\] 
		\vspace{-3mm}
		\[\Psi_{\tt,i}=\sum\limits_{q\in \mstate_i(\tt)}\psi_{\tt,i,q}=|\{ (q,i,\bb)\in "\P_{\mathcal{ACD}(\T)}" \; : \;q\in V \text{ of "index" } i \text{ and } \bb  \, \text{ passes through } \tt \}| .\]
	\end{definition}

	\begin{remark*}
		If we consider the root of the trees in $\ACD(\T)$, then each $\Psi_{\varepsilon,i}$ is the number of states in $"\P_{\mathcal{ACD}(\T)}"$ associated to this tree, i.e., $\Psi_{\varepsilon,i}=|\{ (q,i,\bb)\in "\P_{\mathcal{ACD}(\T)}" \; : \; q\in V,\; \bb\in \mbranch("t_q")\}|$. Therefore
		\[ |"\P_{\mathcal{ACD}(\T)}"|=\sum\limits_{i=0}^{r}\Psi_{\varepsilon,i} .\]
	\end{remark*}
	
	\begin{proof}[\textbf{Proof of Theorem~\ref{Th_OptimalityACDTransformation}}]
		Let $\T=(V,E,\msource,\mtarget,I_0,\F)$ be a "Muller transition system", $"\P_{\mathcal{ACD}(\T)}"$ the "ACD-parity transition system" of $\T$ and $\P=(V',E',\msource',\mtarget',I_0',p':E':\rightarrow \NN)$ a parity transition system such that there is a "locally bijective morphism" $\pp: \P \rightarrow \T$.
		
		First of all, we construct two modified transition systems $\widetilde{\T}=(V,\widetilde{E},\mathit{So\hspace{-0.3mm}\widetilde{urc}\hspace{-0.3mm}e},\mathit{Ta\hspace{-0.3mm}\widetilde{rge}\hspace{-0.3mm}t},I_0,\widetilde{\F})$ and $\widetilde{\P}=(V',\widetilde{E'},\mathit{So\hspace{-0.3mm}\widetilde{urc}\hspace{-0.3mm}e}',\mathit{Ta\hspace{-0.3mm}\widetilde{rge}\hspace{-0.3mm}t}', I_0', \widetilde{p'}:\widetilde{E'}:\rightarrow \NN)$, such that
		
		\begin{enumerate}
			\item Each vertex of $V$ belongs to a "strongly connected component".
			\item All leaves $\tt\in t_i$ verify $|\mstate_i(\tt)|=1$, for every $t_i\in \ACD(\widetilde{\T})$.
			\item Nodes $\tt\in t_i$ verify $\mstate_i(\tau)=\bigcup_{\ss \in \mathit{Children}(\tau)}\mathit{States}_i(\ss)$, for every $t_i\in \ACD(\widetilde{\T})$.
			\item There is a "locally bijective morphism" $\widetilde{\pp}: \widetilde{\P} \rightarrow \widetilde{\T}$.
			\item $|\P_{\mathcal {ACD}(\widetilde{\T} )}|\leq |\widetilde{\P}| \; \Rightarrow \; |\P_{\mathcal {ACD}(\T )}|\leq |\P|$.
		\end{enumerate}

		We define the transition system $\widetilde{\T}$ by adding for each $q\in V$ two new edges, $e_{q,1}, e_{q,2}$ with $\mathit{So\hspace{-0.3mm}\widetilde{urc}\hspace{-0.3mm}e}(e_{q,j})=\mathit{Ta\hspace{-0.3mm}\widetilde{rge}\hspace{-0.3mm}t}(e_{q,j})=q$, for $j=1,2$. 
		The modified "acceptance condition" $\widetilde{\F}$ is given by: let $C\subseteq \widetilde{E}$
		\begin{itemize}
			\item If $C\cap E\neq \emptyset$, then $C\in \widetilde{\F} \; \Leftrightarrow \; C\cap E \in \F$ (the occurrence of edges $e_{q,j}$ does not change the "acceptance condition").
			\item If $C\cap E = \emptyset$, if there are edges of the form $e_{q,1}$ in $C$, for some $q\in V$, then $C\in \widetilde{\F}$. If all edges of $C$ are of the form $e_{q,2}$, $C\notin \F$.
		\end{itemize} 
		It is easy to verify that the "transition system" $\widetilde{\T}$ and $\ACD(\widetilde{\T})$ verify conditions 1,2 and 3.
		We perform equivalent operations in $\P$, obtaining $\widetilde{\P}$: 
		we add a pair of edges $e_{q,1}, e_{q,2}$ for each vertex in $\P$, and we assign them priorities $\widetilde{p}(e_{q,1})=\eta+\epsilon$ and $\widetilde{p}(e_{q,2})=\eta+\epsilon+1$, where $\eta$ is the maximum of the priorities in $\P$ and $\epsilon=0$ if $\eta$ is even, and $\epsilon=1$ if $\eta$ is odd.  We extend the "morphism" $\pp$ to $\widetilde{\pp}: \widetilde{\P} \rightarrow \widetilde{\T} $ conserving the "local bijectivity" by setting $\widetilde{\pp}_E(e_{q,j})=e_{\pp(q),j}$ for $j=1,2$. Finally, it is not difficult to verify that the underlying graphs of $\P_{\mathcal {ACD}(\widetilde{\T} )}$ and $\widetilde{\P}_{\mathcal {ACD}(\T )}$ are equal (the only differences are the priorities associated to the edges $e_{q,j}$), so in particular $|\P_{\mathcal {ACD}(\widetilde{\T} )}|=|\widetilde{\P}_{\mathcal {ACD}(\T )}|=|\P_{\mathcal {ACD}(\T )}|$. Consequently,  $|\P_{\mathcal {ACD}(\widetilde{\T} )}|\leq |\widetilde{\P}|$ implies $|\P_{\mathcal {ACD}(\T)}|\leq |\widetilde{\P}|=|\P|$.

		Therefore, it suffices to prove the theorem for the modified systems $\widetilde{\T}$ and $\widetilde{\P}$. From now on, we take $\T$ verifying the conditions 1, 2 and 3 above. In particular, all trees are "proper trees" in $\ACD(\T) $. It also holds that for each $q\in V$ and $\tt\in t_i$ that is not a leaf, $\psi_{\tt,i,q}=\sum\limits_{\ss\in \mathit{Children}(\tt)}\psi_{\ss,i,q}$. Therefore, for each $\tt\in t_i$ that is not a leaf $$\Psi_{\tt,i}=\sum\limits_{\ss\in \mathit{Children}(\tt)}\Psi_{\ss,i},$$ and for each leaf $\ss\in t_i$ we have $\Psi_{\ss,i}=1$.
		
		 Vertices of $V'$ are partitioned in the equivalence classes of the preimages by $\pp$ of the roots of the trees $\{t_1,\dots,t_r\}$ of $\ACD(\T)$:
		\[ V'= \bigcup\limits_{i=1}^r\pp_V^{-1}( \mstate_i(\varepsilon)) \quad \text{ and } \quad \pp_V^{-1}( \mstate_i(\varepsilon)) \cap \pp_V^{-1}( \mstate_j(\varepsilon))=\emptyset \text{ for } i\neq j .\]

		\begin{claim*}\label{Lemma_C_tau_bigger_Psi}
			For each $i=1,\dots,r$ and each $\tt\in t_i$, if $C_\tt$ is a non-empty "$\nu_i(\tt)$-SCC", then
			$|C_\tt|\geq \Psi_{\tt,i}$.
		\end{claim*}
		
		
		 Let us suppose this claim holds. In particular $(\pp_V^{-1}( \mstate_i(\varepsilon),\pp_E^{-1}( \nu_i(\varepsilon))$ verifies the property \eqref{Eq_l-SCC_property-star} from Definition~\ref{Def_l-SCC}, so from the proof of Lemma~\ref{Lemma_ExistanceSCC_SurjMorph} we deduce that it contains a $\nu_i(\varepsilon)$-SCC and therefore $|\pp_V^{-1}( \mstate_i(\varepsilon))| \geq \Psi_{\varepsilon,i}$, so \[|\P|=\sum\limits_{i=1}^r |\pp_V^{-1}( \mstate_i(\varepsilon))|\geq \sum\limits_{i=1}^r \Psi_{\varepsilon,i}=|"\P_{\mathcal{ACD}(\T)}"|,\] concluding the proof.

	\begin{claimproof}[Proof of the claim]\renewcommand{\qedsymbol}{} 
		
		Let $C_\tt$ be a "$\nu_i(\tt)$-SCC". Let us prove $|C_\tt|\geq \Psi_{\tt,i}$ by induction on the "height of the node" $\tt$. If $\tt$ is a leaf (in particular if its height is $1$), $\Psi_{\tt,i}=1$ and the claim is clear.
		 If $\tt$ of height $h>1$ is not a leaf, then it has children $\ss_1,\dots, \ss_k$, all of them of height $h-1$. Thanks to Lemmas~\ref{Lemma_ExistanceSCC_SurjMorph} and~\ref{Lemma_Disjoint_SCC_BijMorphism}, for $j=1,\dots,k$, there exist disjoint "$\nu_i(\ss_j)$-SCC" included in $C_\tt$, named $C_1,\dots,C_k$, so by induction hypothesis
		\[ |C_\tt| \geq \sum\limits_{j=1}^k |C_j|  \geq \sum\limits_{j=1}^k \Psi_{\ss_j,i}= \Psi_{\tt,i}. \qedhere\]
	
	\end{claimproof}
	\end{proof}

		\begin{remark*}
			From the hypothesis of Theorem~\ref{Th_OptimalityACDTransformation} we cannot deduce that there is a "morphism" from $\P$ to $"\P_{\mathcal{ACD}(\T)}"$ or vice-versa. To produce a counter-example it is enough to remember the ``non-determinism'' in the construction of $\P_{\mathcal{ACD}(\T )}$. Two different orderings in the nodes of the trees of $\ACD(\T)$ will produce two incomparable, but minimal in size parity transition systems that admit a "locally bijective morphism" to $\T$.		
		\end{remark*}
	However, we can prove the following result:
	\begin{proposition}
		If $\pp_1: "\P_{\mathcal{ACD}(\T)}" \rightarrow \T$ is the "locally bijective morphism" described in the proof of Proposition~\ref{Prop_Correctness-ACD}, then for every state $q$ in $\T$ of "index" $i$:
		\[ |\pp_1^{-1}(q)|=\psi_{\varepsilon,i,q}\leq |\pp^{-1}(q)| \;, \; \text{ for every "locally bijective morphism" } \pp: \P \rightarrow \T .\]
	\end{proposition}

	\begin{proof}
		It is enough to remark that if $q \in \mstate_i(\tt)$, then any "$\nu_i(\tt)$-SCC" $C_\tt$ of $\P$ will contain some state in $\pp^{-1}(q)$. We prove by induction as in the proof of the claim that $\psi_{\tt,i,q} \leq |C_\tt \cap \pp^{-1}(q)|$.
	\end{proof}

\newpage

%% file: applications.tex
\subsection{Determinisation of Büchi automata}\label{Section_DeterminisationBuchi}
In many applications, such as the synthesis of reactive systems for $LTL$-formulas, we need to have "deterministic" automata. For this reason, the determinisation of automata is usually a crucial step. Since McNaughton showed in~\cite{mcNaughton1966Testing} that Büchi automata can be transformed into Muller deterministic automata recognising the same language, much effort has been put into finding an efficient way of performing this transformation. The first efficient solution was proposed by Safra in~\cite{safra1988onthecomplexity}, producing a deterministic automaton using a "Rabin condition". Due to the many advantages of "parity conditions" (simplicity, easy complementation of automata, they admit memoryless strategies for games, closeness under union and intersection...), determinisation constructions towards parity automata have been proposed too. In~\cite{piterman2006fromNDBuchi}, Piterman provides a construction producing a parity automaton that in addition improves the state-complexity of Safra's construction. In~\cite{schewe2009tighter}, Schewe breaks down Piterman's construction in two steps: the first one from a non deterministic Büchi automaton $\B$ towards a "Rabin automaton" ($"\R_\B"$) and the second one gives Piterman's parity automaton ($"\P_\B"$). 

In this section we prove that there is a "locally bijective morphism" from $"\P_\B"$ to $"\R_\B"$, and therefore we would obtain a smaller parity automaton applying the "ACD-transformation" in the second step. We provide an example (Example~\ref{Example_Buchi_betterACD}) in which the "ACD-transformation" provides a strictly better parity automaton.

\paragraph*{From non-deterministic Büchi to deterministic Rabin automata}

In~\cite{schewe2009tighter}, Schewe presents a construction of a deterministic "Rabin automaton" $""\R_\B""$ from a non-deterministic "Büchi automaton" $\B$. The set of states of the automaton $\R_\B$ is formed of what he calls ""history trees"". The number of history trees for a Büchi automaton of size $n$ is given by the function $\mathit{hist}(n)$, that is shown to be in $o((1.65n)^n)$ in~\cite{schewe2009tighter}. This construction is presented starting from a state-labelled Büchi automaton. A construction starting from a transition-labelled Büchi automaton can be found in~\cite{varghese2014PhD}. In~\cite{colcombetz2009tight}, Colcombet and Zdanowski proved the worst-case optimality of the construction.

\begin{proposition}[\cite{schewe2009tighter}]
	Given a non-deterministic "Büchi automaton" $\B$ with $n$ states, there is an effective construction of a deterministic "Rabin automaton" $\R_\B$ with $\mathit{hist}(n)$ states and using $2^{n-1}$ Rabin pairs that recognises the language $"\L(\B)"$.
\end{proposition}

\begin{proposition}[\cite{colcombetz2009tight}]
	For every $n\in \NN$ there exists a non-deterministic Büchi automaton $B_n$ of size $n$ such that every deterministic Rabin automaton recognising $\L(\B_n)$ has at least $\mathit{hist}(n)$ states.
\end{proposition}

\paragraph*{From non-deterministic Büchi to deterministic parity automata}

In order to build a deterministic "parity automaton" $""\P_\B""$ that recognises the language of a given "Büchi automaton" $\B$, Schewe transforms the automaton $"\R_\B"$ into a parity one using what he calls a \emph{later introduction record} (LIR). The LIR construction can be seen as adding an ordering (satisfying some restrictions) to the nodes of the "history trees". States of $\P_\B$ are therefore pairs of history trees with a LIR. In this way we obtain a similar parity automaton that with the Piterman's determinisation procedure~\cite{piterman2006fromNDBuchi}.
The worst-case optimality of this construction was proved in~\cite{Varghese14,varghese2014PhD}, generalising the methods of~\cite{colcombetz2009tight}.

\begin{proposition}[\cite{schewe2009tighter}]
	Given a non-deterministic "Büchi automaton" $\B$ with $n$ states, there is an effective construction of a deterministic "parity automaton" $"\P_\B"$ with $O(n!(n-1)!)$ states and using $2n$ priorities that recognises the language $"\L(\B)"$.
\end{proposition}

\begin{proposition}[\cite{Varghese14,varghese2014PhD}]
	For every $n\in \NN$ there exists a non-deterministic Büchi automaton $B_n$ of size $n$ such that $"\P_\B"$ has less than $1.5$ times as many states as a minimal deterministic parity automaton recognising $\L(\B_n)$.
\end{proposition}

\paragraph*{A locally bijective morphism from $\P_\B$ to $\R_\B$ }

\begin{proposition}\label{Prop_LocBijMorphFromP_B-to-R_B}
	Given a "Büchi automaton" $\B$ and its determinisations to Rabin and parity automata $"\R_\B"$ and $"\P_\B"$, there is a "locally bijective morphism" $\pp: \P_\B \rightarrow \R_\B$.
\end{proposition}

\begin{proof}
	Observing the construction of $\R_\B$ and $\P_\B$ in~\cite{schewe2009tighter}, we see that the states of $\P_\B$ are of the form $(T,\chi)$ with $T$ an state of $\R_B$ (a "history tree"), and $\chi: T \rightarrow \{1,\dots,|B|\}$ a LIR (that can be seen as an ordering of the nodes of $T$).
	
	It is easy to verify that the mapping $\pp_V((T,\chi))=T$ defines a morphism $\pp: \R_\B \rightarrow \P_\B$ (from Fact~\ref{Fact_MorphismDetAutomata-LocBijective} there is only one possible definition of $\pp_E$). Since the automata are deterministic, $\pp$ is a "locally bijective morphism".
	
\end{proof}

\begin{theorem}\label{Th_ACD-Better than Schewe}
	Let $\B$ be a "Büchi automaton" and $"\R_\B"$, $"\P_\B"$ the deterministic Rabin and parity automata obtained by applying the Piterman-Schewe construction to $\B$. Then, the parity automaton $\P_{\mathcal{ACD}(\R_B)}$ verifies
	$ |\P_{\mathcal{ACD}(\R_B)}| \leq |\P_\B| $
	and  $"\P_{\mathcal{ACD}(\R_B)}"$ uses a smaller number of priorities than $"\P_\B"$.
\end{theorem}
\begin{proof}
	It is a direct consequence of Propositions~\ref{Prop_LocBijMorphFromP_B-to-R_B},~\ref{Prop_OptimalityACD_Priorities} and Theorem~\ref{Th_OptimalityACDTransformation}.
\end{proof}

\begin{remark*}
	Furthermore, after  Proposition~\ref{Prop_OptimalityACD_Priorities-Languages}, $"\P_{\mathcal{ACD}(\R_B)}"$ uses the optimal number of priorities to recognise $"\L(\B)"$, and we directly obtain this information from the "alternating cycle decomposition" of $\R_\B$, $\ACD(\R_\B)$.
	
\end{remark*}

In Example~\ref{Example_Buchi_betterACD} we show a case in which $|\P_{\mathcal{ACD}(\R_B)}| < |\P_\B|$ and for which the gain in the number of priorities is clear.

\begin{remark*}
	In~\cite{colcombetz2009tight} and~\cite{varghese2014PhD}, the lower bounds for the determinisation of "Büchi automata" to Rabin and parity automata where shown using the family of  \AP""full Büchi automata"", $\{\B_n\}_{n\in \NN}$, $|\B_n|=n$. The automaton $\B_n$ can simulate any other Büchi automaton of the same size. For these automata, the constructions $\P_{\B_n}$ and $\P_{\mathcal{ACD}(\R_{B_n})}$ coincide.
\end{remark*}
\newcommand*{\boldcheckmark}{%
	\textpdfrender{
		TextRenderingMode=FillStroke,
		LineWidth=.5pt, 
	}{\checkmark}%
}
\begin{example}\label{Example_Buchi_betterACD}
	We present a non-deterministic Büchi automaton $\B$ such that the "ACD-parity automaton" of $"\R_\B"$ has strictly less states and uses strictly less priorities than $"\P_\B"$.
	
	In Figure~\ref{Fig_Buchi_betterACD} we show the automaton $\B$ over the alphabet $\Sigma=\{a,b,c\}$. Accepting transitions for the "Büchi condition" are represented with a black dot on them. An accessible "strongly connected component" $\R_\B'$ of the determinisation to a "Rabin automaton" $"\R_\B"$ is shown in Figure~\ref{Fig_Buchi_betterACD-Rabin}. It has 2 states that are "history trees" (as defined in~\cite{schewe2009tighter}). There is a "Rabin pair" $(E_\tt,F_\tt)$ for each node appearing in some "history tree" (four in total), and these are represented by an array with four positions. We assign to each transition and each position $\tt$ in the array the symbols $\textcolor{Green2}{\boldcheckmark}$, $\textcolor{Red2}{\mathbf{X}}$ ,  or $\textcolor{Orange2}{\bullet}$ depending on whether this transition belongs to $E_\tt$, $F_\tt$ or neither of them, respectively (we can always suppose $E_\tt\cap F_\tt = \emptyset$).

	In Figure~\ref{Fig_Buchi_betterACD-ACD} there is the "alternating cycle decomposition" corresponding to $\R_\B'$. We observe that the tree of $\ACD(\R_\B')$ has a single branch of height $3$. 
	This is, the Rabin condition over $\R_\B'$ is already a "$[1,3]$-parity condition" and $\P_{\mathcal{ACD}(\R_B')}=\R_\B'$. In particular it has $2$ states, and uses priorities in $[1,3]$.
	
	On the other hand, in Figure~\ref{Fig_Buchi_betterACD-parity} we show the automaton $\P_\B'$, that has $3$ states and uses priorities in $[3,7]$. The whole automata $\R_\B$ and $\P_\B$ are too big to be pictured in these pages, but the three states shown in Figure~\ref{Fig_Buchi_betterACD-parity} are indeed accessible from the initial state of $\P_\B$.
	\begin{figure}[h]
		\centering 
		\scalebox{0.93}{
		\begin{tikzpicture}[square/.style={regular polygon,regular polygon sides=4}, align=center,node distance=2cm,inner sep=2pt]
		
		\node at (0,0) [state, initial] (0) {$q_0$};
		\node at (2,1) [state] (1) {$q_1$};
		\node at (2,-1) [state] (2) {$q_2$};
		\node at (4,0) [state] (3) {$q_3$};

		\path[->] 
		
		(0)  edge [out=115,in=65,loop] 	node[above] { a,b,c }   (0)
		(0)  edge []  node[above] {a \\[-3mm]} node[scale=2]{ \textbullet }   (1)
		(0)  edge []  node[above] {a,c}   (2)
		
		(1)  edge [out=105,in=55,loop, distance=10mm] 	node[left, pos=0.2] { a,b,c  }    (1)
		(1)  edge [thick]  node[above] { b \\[-3mm] } node[scale=2]{ \textbullet }  (3)
		
		(2)  edge []  node[right,pos=0.4] {c}   (1)
		(2)  edge [out=165,in=200,loop] node[left] { a }	node[scale=2] {\textbullet }   (2)
		(2)  edge [thick,out=20,in=-20,loop] node[right] { \hspace{1mm}b }	   (2)
		
		(3)  edge [thick,out=115,in=65,loop] node[above] { a,b,c }   (3);
		
		\end{tikzpicture}
	}
		\caption{Büchi automaton $\B$ over the alphabet $\Sigma=\{a,b,c\}$.}
		\label{Fig_Buchi_betterACD}
	\end{figure}

	\begin{figure}[ht]
				\centering 
		\scalebox{0.93}{

		\begin{tikzpicture}[square/.style={regular polygon,regular polygon sides=4}, align=center,node distance=2cm,inner sep=3pt]
		
		\node (STATE 1) at (0,0) [draw,circle]{
			\begin{tikzpicture}[scale=0.7, every node/.style={scale=0.9}]
			
			\node at (0,2) [draw,ellipse,text width=1.4cm,minimum height=0.6cm] (aR) {$\hspace{-2mm}q_0,q_1,q_2,q_3$};
		
			\node at (-0.8,1) [draw, ellipse,text width=0.55cm,minimum height=0.6cm] (a0) {$\hspace{-2mm}q_1,q_3$ };
			\node at (0.5,1) [draw, ellipse,minimum width=0.7cm,minimum height=0.6cm] (a1) {$q_2$};
		
			\node at (-0.8,0) [ellipse, draw,minimum width=0.7cm,minimum height=0.6cm] (a00) {$q_3$};

			\draw   
			(aR) edge (a0)
			(aR) edge (a1)
			(a0) edge (a00);
			\end{tikzpicture}
	};
		\node (STATE 2) at (6,0) [draw,circle]{
			\begin{tikzpicture}[scale=0.7, every node/.style={scale=0.9}]
		\node at (5,2) [draw,ellipse,text width=1.4cm,minimum height=0.6cm] (bR) {$\hspace{-2mm}q_0,q_1,q_2,q_3$};
		
		\node at (5,1) [draw, ellipse,text width=0.55cm,minimum height=0.6cm] (b0) {$\hspace{-2mm}q_1,q_3$ };

		\node at (5,0) [ellipse, draw,minimum width=0.7cm,minimum height=0.6cm] (b00) {$q_3$};
		
		
		\draw
		(bR) edge (b0)
		(b0) edge (b00);
			\end{tikzpicture}
	};

	\node at (0.5,-0.8) [scale=1.5, color=red] (state1) {1};
	\node at (6.8,-0.8) [scale=1.5, color=red] (state2) {2};

		\path[->] 
		
		(STATE 1) edge [out=145,in=170,loop, distance=1.5cm] node[above,pos=0.2]{a}  node[left,pos=0.3,scale=0.8]{$\left(
			\begin{array}{c}
			\textcolor{Orange2}{\bullet} \\
			\textcolor{Orange2}{\bullet} \\
			\textcolor{Green2}{\boldcheckmark} \\
			\textcolor{Orange2}{\bullet} 
			\end{array} \right) \hspace{2mm}$} (STATE 1)
		
		(STATE 1) edge [out=190,in=215,loop, distance=1.5cm] node[below,pos=0.8]{b} node[left,pos=0.7,scale=0.8]{$\left(
			\begin{array}{c}
			\textcolor{Orange2}{\bullet} \\
			\textcolor{Orange2}{\bullet} \\
			\textcolor{Orange2}{\bullet} \\
			\textcolor{Orange2}{\bullet} 
			\end{array} \right) \hspace{2mm}$} (STATE 1)
		
		(STATE 1) edge [out=30,in=150] node[above,pos=0.4]{c} node[above,pos=0.6,scale=0.8]{$\left(
			\begin{array}{c}
			\textcolor{Orange2}{\bullet} \\
			\textcolor{Orange2}{\bullet} \\
			\textcolor{Red2}{\mathbf{X}} \\
			\textcolor{Orange2}{\bullet} 
			\end{array} \right) $} (STATE 2)
		
		(STATE 2) edge [out=35,in=10,loop, distance=1.5cm] node[above,pos=0.2]{b} node[right,pos=0.3,scale=0.8]{$\hspace{2mm} \left(
			\begin{array}{c}
			\textcolor{Orange2}{\bullet} \\
			\textcolor{Orange2}{\bullet} \\
			\textcolor{Red2}{\mathbf{X}} \\
			\textcolor{Orange2}{\bullet}
			\end{array} \right) \hspace{2mm}$} (STATE 2)
		
		(STATE 2) edge [out=-10,in=-35,loop , distance=1.5cm] node[below,pos=0.8]{c} node[right,pos=0.7,scale=0.8]{$\hspace{2mm} \left(
			\begin{array}{c}
			\textcolor{Orange2}{\bullet} \\
			\textcolor{Orange2}{\bullet} \\
			\textcolor{Red2}{\mathbf{X}} \\
			\textcolor{Orange2}{\bullet}
			\end{array} \right) \hspace{2mm}$} (STATE 2)
		
		(STATE 2) edge [out=210,in=-30, scale=0.6] node[above,pos=0.6]{a} node[above,pos=0.4,scale=0.8]{$\left(
			\begin{array}{c}
			\textcolor{Orange2}{\bullet} \\
			\textcolor{Orange2}{\bullet} \\
			\textcolor{Orange2}{\bullet} \\
			\textcolor{Orange2}{\bullet} 
			\end{array} \right) $} (STATE 1);

		\end{tikzpicture}
	}
		\caption{ Component $\R_\B'$ of the automaton $\R_\B$. States are "history trees" and transitions are labelled with the input letter and an array representing the Rabin condition. There is one component in the arrays for each node appearing in some history tree, taking the order  \resizebox{17mm}{!}{$\left(
				\begin{array}{c}
				\text{Root} \; \varepsilon \\
				\text{Node }0\\
				\text{Node }1 \\
				\text{Node }00 
				\end{array} \right)$}.\\			
		The components of the arrays can host a green checkmark $\textcolor{Green2}{\boldcheckmark}$ (if the node is in $E_\tt$), a red cross $\textcolor{Red2}{\mathbf{X}}$ (if the node is in $F_\tt$) or an  orange dot $\textcolor{Orange2}{\bullet}$ (if the node is not in $E_\tt$ or $F_\tt$). } 
		\label{Fig_Buchi_betterACD-Rabin}
	\end{figure}
	
	\begin{figure}[ht!]
		\centering
		
		\begin{tikzpicture}[square/.style={regular polygon,regular polygon sides=4}, align=center,node distance=2cm,inner sep=3pt]
		
		\node at (0,3) [draw, rectangle, text height=0.3cm, text width=3cm] (R) { $a_1,b_1,c_1,a_2,b_2,c_2$};
		
		\node at (0,2) [draw, ellipse,text height=0.2cm, text width=0.8cm] (0) {$a_1,b_1$};
		
		\node at (0,1) [draw, rectangle, text height=0.3cm, text width=0.5cm] (1) { $b_1$};

		\draw   
		(R) edge (0)
		(0) edge (1);
		
		\node at (2,3) {$1$};
		\node at (1.2,2) {$2$};
		\node at (1.2,1) {$3$};
		
		\end{tikzpicture}
		\caption{The "alternating cycle decomposition" of $\R_\B'$, $\ACD(\R_\B')$. The labels $a_1,b_1,c_1$ correspond to transitions leaving the state 1 on Figure~\ref{Fig_Buchi_betterACD-Rabin} and $a_2,b_2,c_2$ those leaving state 2. We observe that this allows to substitute the "Rabin condition" on $\R_\B'$ for an equivalent $[1,3]$-parity condition on the same underlying automaton.}
		\label{Fig_Buchi_betterACD-ACD}
		
	\end{figure}
	\begin{figure}[ht]
		\hspace{-10mm} 
		\begin{tikzpicture}[square/.style={regular polygon,regular polygon sides=4}, align=center,node distance=2cm,inner sep=3pt]
		
		\node (STATE 1) at (0,0) [draw,circle]{
			\begin{tikzpicture}[scale=0.7, every node/.style={scale=0.82}]
			
			\node at (0,2.2) [draw,ellipse,text width=1.4cm,text height=0.12cm] (aR) {$\hspace{-2mm}q_0,q_1,q_2,q_3$\\[-0.5mm] \textcolor{Violet2}{$\quad \,$ 0}\vspace{-2mm}};
			
			\node at (-0.8,1) [draw, ellipse,text width=0.55cm,minimum height=0.6cm] (a0) {$\hspace{-2mm}q_1,q_3$\\[-0.5mm]\textcolor{Violet2}{$\;\;$1}\vspace{-2mm} };
			\node at (0.5,1) [draw, ellipse,minimum width=0.7cm,minimum height=0.6cm] (a1) {$q_2$\\[-0.5mm]\textcolor{Violet2}{2}\vspace{-2mm}};
			
			\node at (-0.8,-0.2) [ellipse, draw,minimum width=0.7cm,minimum height=0.6cm] (a00) {$q_3$\\[-0.5mm]\textcolor{Violet2}{3}\vspace{-2mm}};
		\draw   
		(aR) edge (a0)
		(aR) edge (a1)
		(a0) edge (a00);
		\end{tikzpicture}
	};
		\node (STATE 2) at (4.6,0) [draw,circle]{
			\begin{tikzpicture}[scale=0.7, every node/.style={scale=0.82}]
			\node at (4,2.2) [draw,ellipse,text width=1.4cm,minimum height=0.6cm] (bR) {$\hspace{-2mm}q_0,q_1,q_2,q_3$\\[-0.5mm]\textcolor{Violet2}{$\quad \,$ 0}\vspace{-2mm}};
			
			\node at (4,1) [draw, ellipse,text width=0.55cm,minimum height=0.6cm] (b0) {$\hspace{-2mm}q_1,q_3$ \\[-0.5mm]\textcolor{Violet2}{$\;$ 1}\vspace{-2mm}};

			\node at (4,-0.2) [ellipse, draw,minimum width=0.7cm,minimum height=0.6cm] (b00) {$q_3$\\[-0.5mm]\textcolor{Violet2}{2}\vspace{-2mm}};
			
	\draw   
	(bR) edge (b0)
	(b0) edge (b00);
	\end{tikzpicture}
};
		
		\node (STATE 3) at (9.4,0) [draw,circle]{
		\begin{tikzpicture}[scale=0.7, every node/.style={scale=0.82}]
			\node at (8.5,2.2) [draw,ellipse,text width=1.4cm,minimum height=0.6cm] (cR) {$\hspace{-2mm}q_0,q_1,q_2,q_3$\\[-0.5mm]\textcolor{Violet2}{$\quad \,$ 0}\vspace{-2mm}};
			
			\node at (7.7,1) [draw, ellipse,text width=0.55cm,minimum height=0.6cm] (c0) {$\hspace{-2mm}q_1,q_3$\\[-0.5mm]\textcolor{Violet2}{$\;$ 1}\vspace{-2mm} };
			\node at (9,1) [draw, ellipse,minimum width=0.7cm,minimum height=0.6cm] (c1) {$q_2$\\[-0.5mm]\textcolor{Violet2}{3}\vspace{-2mm}};
			
			\node at (7.7,-0.2) [ellipse, draw,minimum width=0.7cm,minimum height=0.6cm] (c00) {$q_3$\\[-0.5mm]\textcolor{Violet2}{2}\vspace{-2mm}};
	
		\draw   
		(cR) edge (c0)
		(cR) edge (c1)
		(c0) edge (c00);
		\end{tikzpicture}
	};
		
%
%
		
		\path[->] 
		
		(STATE 1) edge [out=145,in=170,loop, distance=1.5cm] node[above,pos=0.2]{a : \textcolor{Green2}{$4$}}  (STATE 1)
		
		(STATE 1) edge [out=190,in=215,loop, distance=1.5cm] node[above,pos=0.2]{b : \textcolor{Green2}{$7 \;\quad$}} (STATE 1)
		
		(STATE 1) edge [] node[above,pos=0.4]{c : \textcolor{Green2}{$3$}} (STATE 2)
		
		(STATE 2) edge [out=105,in=75,loop, distance=1.5cm] node[above,pos=0.5]{b,c : \textcolor{Green2}{$5$}} (STATE 2)
		
		(STATE 2) edge [out=30,in=150] node[above,pos=0.5]{a : \textcolor{Green2}{$7$}} (STATE 3)
		
		(STATE 3) edge [out=210,in=-30, scale=0.6] node[below,pos=0.5]{c : \textcolor{Green2}{$5$}}  (STATE 2)
		
		(STATE 3) edge [out=35,in=10,loop, distance=1.5cm] node[above,pos=0.2]{a : \textcolor{Green2}{$6$}} (STATE 3)
		
		(STATE 3) edge [out=-10,in=-35,loop, distance=1.5cm] node[above,pos=0.2]{$\quad$ b : \textcolor{Green2}{$7$}} (STATE 3);

%
%
%
%
%
%
%
%

		\end{tikzpicture}
		\caption{Component $\P_\B'$ of the automaton $\P_\B$. States are ordered "history trees", with the order labels in \textcolor{Violet2}{violet}. It has three different states, since two different orders for the same "history tree" occur. } 
		\label{Fig_Buchi_betterACD-parity}
	\end{figure}
	
\end{example}

\subsection{On relabelling of transition systems by acceptance conditions}\label{Section_4.2.StructuralParity}

 In this section we use the information given by the "alternating cycle decomposition" to provide characterisations of "transition systems" that can be labelled with "parity", "Rabin", "Streett" or $"\mathit{Weak}_k"$ conditions, generalising the results of~\cite{zielonka1998infinite}.
 As a consequence, these yield simple proofs of two results about the possibility to define different classes of acceptance conditions in a deterministic automaton. Theorem~\ref{Th_Rabin+Streett->Parity}, first proven in~\cite{kupferman2010ParityizingRA}, asserts that if we can define a Rabin and a Streett condition on top of an underlying automaton $\A$ such that it recognises the same language $L$ with both conditions, then we can define a parity condition in $\A$ recognising $L$ too. Theorem~\ref{Th_Buchi+co-B->Weak} states that if we can define Büchi and co-Büchi conditions on top of an automaton $\A$ recognising the language $L$, then we can define a $"\mathit{Weak}"$ condition over $\A$ such that it recognises $L$.

%
%

First, we extend the Definition~\ref{Def_ShapesOfTrees} of Section~\ref{Section_ZielonkaTree-SomeTypes} to the "alternating cycle decomposition".

\begin{definition}
	\begin{itemize}
	Given a Muller transition system $\T=(V,E,\msource,\mtarget,I_0,\F)$, we say that its "alternating cycle decomposition" $\ACD(\T)$ is a
		
		\item  \AP""Rabin ACD"" if for every state $q\in V$, the tree $"t_q"$ has "Rabin shape".
		
		\item  \AP""Streett ACD"" if for every state $q\in V$, the tree $"t_q"$ has "Streett shape".
		
		\item  \AP""parity ACD"" if for every state $q\in V$, the tree $"t_q"$ has "parity shape".
		
		\item  \AP""$[1,\eta]$-parity ACD"" (resp. $[0,\eta-1]$-parity ACD) if it is a parity ACD, every tree has "height" at most $\eta$ and trees of height $\eta$ are \kl(ACD){odd} (resp. \kl(ACD){even}).
		
		\item  \AP""Büchi ACD"" if it is a $[0,1]$-parity ACD. \item  \AP""co-Büchi ACD"" if it is a $[1,2]$-parity ACD.
		
		
		\item  \AP""$\mathit{Weak}_k$ ACD"" if it is a parity ACD and every tree $(t_i,\nu_i) \in \ACD(\T)$ has "height" at most $k$.
\end{itemize}
\end{definition}

The next proposition follows directly from the definitions.

\begin{proposition}\label{Prop_ACD_equivalentTypes}
	Let $\T$ be a Muller transition system. Then:
	\begin{itemize}
		\item $\ACD(\T)$ is a "parity ACD" if and only if it is a "Rabin ACD" and a "Streett ACD".
		\item $\ACD(\T)$ is a "$\mathit{Weak}_k$ ACD" if and only if it is a "$[0,k]$-parity ACD" and a "$[1,k+1]$-parity ACD".
	\end{itemize}
\end{proposition}

\begin{proposition}\label{Prop_RelabellinRabin}
	Let $\T=(V,E,\msource,\mtarget,I_0,\F)$ be a Muller transition system. The following conditions are equivalent:
	\begin{enumerate}
		\item We can define a "Rabin condition" 
		over $\T$ that is "equivalent to" $\F$ over $\T$.
		\item For every pair of "loops" $l_1,l_2\in \mloop(\T)$ such that $l_1\cup l_2$ is a "loop", if $l_1\notin \F $ and $l_2\notin \F$, then $l_1\cup l_2 \notin \F$. 
		\item $\ACD(\T)$ is a "Rabin ACD".
	\end{enumerate}
\end{proposition}
\begin{proof}
	\begin{description}
		\item[($1 \Rightarrow 2$)] 
		Suppose that $\T$ uses a Rabin condition with Rabin pairs $(E_1,F_1),\dots,(E_r,F_r)$. Let $l_1$ and $l_2$ be two rejecting "loops". If $l_1\cup l_2$ was accepting, then there would be some Rabin pair $(E_j,F_j)$ and some edge $e\in l_1\cup l_2$ such that $e\in E_j$ and $e\notin F_j$. However, the edge $e$ belongs to $l_1$ or to $l_2$, and the "loop" it belongs to should be accepting too.
		
		\item[($2 \Rightarrow 3$)] Let $q\in V$ of "index" $i$, and $"t_q"$ be the subtree of $\ACD(\T)$ associated to $q$. Suppose that there is a node $\tt \in t_q$ such that $\kl(node){p_i}(\tt)$ is even ("round" node) and that it has two different children $\ss_1$ and $\ss_2$. The "loops" $\nu_i(\ss_1)$ and $\nu_i(\ss_2)$ are maximal rejecting "loops" contained in $\nu_i(\tt)$, and since they share the state $q$, their union is also a "loop" that must verify $ \nu_i(\ss_1) \cup \nu_i(\ss_2)\in \F$, contradicting the hypothesis.
		
		\item[($3 \Rightarrow 1$)] We define a "Rabin condition" over $\T$. For each tree $t_i$ in $\ACD(\T)$ and each "round" node $\tt\in t_i$ ($\kl(node){p_i}(\tt)$ even) we define the Rabin pair $(E_{i,\tt},F_{i,\tt})$ given by:
		\[ E_{i,\tt}=\nu_i(\tt)\setminus \bigcup_{\ss \in "\mathit{Children}"(\tau)}\nu_i(\ss) \quad, \qquad  F_{i,\tt}=E \setminus \nu_i(\tt). \]
			
		Let us show that this condition is "equivalent to" $\F$ over the transition system $\T$. We begin by proving the following consequence of being a "Rabin ACD":
		\begin{claim*}
			If $\tt$ is a "round" node in the tree $t_i$ of $\ACD(\T)$, and $l\in \mloop(\T)$ is a "loop" such that $l\subseteq \nu_i(\tt)$ and $l\nsubseteq \nu_i(\ss)$ for any child $\ss$ of $\tt$, then there is some edge $e\in l$ such that $e\notin \nu_i(\ss)$ for any child $\ss$ of $\tt$.
		\end{claim*}
		\begin{claimproof}
			Since for each state $q\in V$ the tree $"t_q"$ has "Rabin shape", it is verified that $\mstate_i(\ss)\cap \mstate_i(\ss')=\emptyset$ for every pair of different children $\ss, \ss'$ of $\tt$. Therefore, the union of $\nu_i(\ss)$ and $\nu_i(\ss')$ is not a "loop", and any "loop" $l$ contained in this union must be contained either in $\nu_i(\ss)$ or in $\nu_i(\ss')$.
		\end{claimproof}

		 Let $\rr\in \mrun_{T}$ be a "run" in $\T$, let $l\in \mloop(\T)$ be the "loop" of $\T$ such that $\minf(\rr)=l$ and let $i$ be the "index" of the edges in this "loop". Let $\tt$ be a maximal node in $t_i$ (for $\prefix$) such that $l\subseteq \nu_i(\tt)$.  If $l\in \F$, this node $\tt$ is a round node, and from the previous claim it follows that there is some edge $e\in l$ such that $e$ does not belong to any child of $\tt$, so $e\in E_{i,\tt}$ and $e\notin F_{i,\tt}$, so the "run" $\rr$ is accepted by the Rabin condition too. If $l\notin \F$, then for every round node $\tt$, if $l\subseteq \nu_i(\tt)$ then $l\subseteq \nu_i(\ss)$ for some child $\ss$ of $\tt$. Therefore, for every Rabin pair $(E_{i,\tt},F_{i,\tt})$ and every $e\in l$, it is verified that $e\in E_{i,\tt} \, \Rightarrow \, e\in F_{i,\tt}$.
	\end{description}
\end{proof}
\begin{remark*}
	The Rabin condition presented in this proof does not necessarily use the optimal number of Rabin pairs required to define a Rabin condition "equivalent to" $\F$ over $\T$.
\end{remark*}

\begin{proposition}\label{Prop_RelabellinStreett}
	Let $\T=(V,E,\msource,\mtarget,I_0,\F)$ be a Muller transition system. The following conditions are equivalent:
	\begin{enumerate}
		\item We can define a "Streett condition" 
		over $\T$ that is "equivalent to" $\F$ over $\T$.
		\item For every pair of "loops" $l_1,l_2\in \mloop(\T)$ such that $l_1\cup l_2$ is a "loop", if $l_1\in \F $ and $l_2\in \F$, then $l_1\cup l_2 \in \F$. 
		\item $\ACD(\T)$ is a "Streett ACD".
	\end{enumerate}
\end{proposition}

We omit the proof of Proposition~\ref{Prop_RelabellinStreett}, being the dual case of Proposition~\ref{Prop_RelabellinRabin}.

\begin{proposition}\label{Prop_RelabellinParity}
	Let $\T=(V,E,\msource,\mtarget,I_0,\F)$ be a Muller transition system. The following conditions are equivalent:
	
	\begin{enumerate}
		\item We can define a "parity condition" 
		over $\T$ that is "equivalent to" $\F$ over $\T$.
		\item For every pair of "loops" $l_1,l_2\in \mloop(\T)$ such that $l_1\cup l_2$ is a "loop", if $l_1\in \F \, \Leftrightarrow \, l_2\in \F$, then $l_1\cup l_2 \in \F \, \Leftrightarrow \,l_1\in \F$. That is, union of "loops" having the same ``accepting status'' preserves their ``accepting status''. 
		\item $\ACD(\T)$ is a "parity ACD".
	\end{enumerate}
	
	Moreover, the parity condition we can define over $\T$ is a "$[1,\eta]$-parity" (resp. $[0,\eta-1]$-parity~/~$"\mathit{Weak}_k"$) condition if and only if $\ACD(\T)$ is a "$[1,\eta]$-parity ACD" (resp. $[0,\eta-1]$-parity ACD~/~"$\mathit{Weak}_k$ ACD").
\end{proposition}
\begin{proof}
	\begin{description}
		\item[($1 \Rightarrow 2$)] Suppose that $\T$ uses a parity acceptance condition with the priorities given by $p:E \rightarrow \NN$. Then, since $l_1$ and $l_2$ are both accepting or both rejecting,  $p_1=\min p(l_1)$ and $p_2=\min p(l_2)$ have the same parity, that is also the same parity than $\min p(l_1 \cup l_2)=\min\{p_1,p_2\}$.
		
		\item[($2 \Rightarrow 3$)] Let $q\in V$ of "index" $i$, and $"t_q"$ be the subtree of $\ACD(\T)$ associated to $q$. Suppose that there is a node $\tt \in t_q$ with two different children $\ss_1$ and $\ss_2$. The "loops" $\nu_i(\ss_1)$ and $\nu_i(\ss_2)$ are different maximal "loops" with the property $\nu_i(\ss)\subseteq \nu_i(\tt)$ and $\nu_i(\ss)\in \F \, \Leftrightarrow \, \nu_i(\tt) \notin \F$. Since they share the state $q$, their union is also a "loop" contained in $\nu_i(\tt)$ and then
		\[ \nu_i(\ss_1) \cup \nu_i(\ss_2)\in \F \; \Leftrightarrow \; \nu_i(\tt) \in \F \; \Leftrightarrow \; \nu_i(\ss_1)\notin \F\] 
		contradicting the hypothesis.
		
		\item[($3 \Rightarrow 1$)] From the construction of the "ACD-transformation", it follows that $"\P_{\mathcal {ACD}(\T )}"$ is just a relabelling of $\T$ with an equivalent parity condition.
		
	\end{description}	
	For the implication from right to left of the last statement, we remark that if the trees of $\ACD(\T)$ have priorities assigned in $[\mu,\eta]$, then the parity transition system $"\P_{\mathcal {ACD}(\T )}"$ will use priorities in $[\mu,\eta]$. If $\ACD(\T)$ is a "$\mathit{Weak}_k$ ACD", then in each "strongly connected component" of $"\P_{\mathcal {ACD}(\T )}"$ the number of priorities used will be the same as the "height" of the corresponding tree of $\ACD(\T)$ (at most $k$).
	
	For the other implication it suffices to remark that the priorities assigned by $\ACD(\T)$ are optimal (Proposition~\ref{Prop_OptimalityACD_Priorities}).	
\end{proof}

\begin{corollary}\label{Coro_Rabin+Streett=Parity_TransSyst}
	Given a "transition system graph" $\T_G=(V,E,\msource,\mtarget,I_0)$ and a "Muller condition" $\F \subseteq \P(E)$, we can define a "parity condition" $p:E\rightarrow \NN$ "equivalent to" $\F$ over $\T_G$ if and only if we can define a "Rabin condition" $R$ and a "Streett condition" $S$ over $\T_G$ such that 
	$ (\T_G,\F) \,"\simeq"\, (\T_G,R)\, "\simeq"\, (\T_G,S) $.
	
	Moreover, if the Rabin condition $R$ uses $r$ Rabin pairs and the Streett condition $S$ uses $s$ Streett pairs, we can take the parity condition $p$ using priorities in
	\begin{itemize}
		\item $[1,2r+1]$ if $r\leq s$.
		\item $[0,2s]$ if $s\leq r$.
	\end{itemize} 
\end{corollary}

\begin{proof}
	
	The first statement is a consequence of the characterisations (2) or (3) from Propositions~\ref{Prop_RelabellinRabin},~\ref{Prop_RelabellinStreett} and~\ref{Prop_RelabellinParity}.
	
	For the second statement we remark that the trees of $\ACD(\T)$ have "height" at most $\min\{2r+1, 2s+1\}$. If $r\geq s$, then the height $2r+1$ can only be reached by \kl(ACD){odd} trees, and if $s\geq r$, the height $2s+1$ only by \kl(ACD){even} trees.
\end{proof}

From the last statement of Proposition~\ref{Prop_RelabellinParity} and thanks to the second item of Proposition~\ref{Prop_ACD_equivalentTypes}, we obtain: 

\begin{corollary}\label{Coro_Buchi+coB=Weak_TransSyst}
	Given a "transition system graph" $\T_G$ and a Muller condition $\F$ over $\T_G$, there is an equivalent $"\mathit{Weak}_k"$ condition over $\T_G$ if and only if there are both $[0,k]$ and "$[1,k+1]$-parity" conditions "equivalent to" $\F$ over $\T_G$. 
	
	In particular, there is an equivalent "Weak condition" if and only if there are "Büchi" and "co-Büchi" conditions equivalent to $\F$ over $\T_G$. 
\end{corollary}

%

\begin{remark*}
	It is important to notice that the previous results are stated for non-labelled transition systems. We must be careful when translating these results to automata and formal languages. For instance, in~\cite[Section 4]{kupferman2010ParityizingRA} there is an example of a non-deterministic automaton $\A$, such that we can put on top of it Rabin and Streett conditions $R$ and $S$ such that $\L(\A,R)=\L(\A,S)$, but we cannot put a parity condition on top of it recognising the same language. However, proposition~\ref{Prop_EquivaleceDetAutomata_ImpliesEqCondition} allows us to obtain analogous results for "deterministic automata".
\end{remark*}

\begin{theorem}[{{\cite[Theorem 7]{kupferman2010ParityizingRA}}}]\label{Th_Rabin+Streett->Parity}
	Let $\A$ be the "transition system graph" of a "deterministic automaton" with set of states $Q$. Let $R$ be a Rabin condition over $\A$ with $r$ pairs and $S$ a Streett condition over $\A$ with $s$ pairs such that $\L(\A,R)=\L(\A,S)$. Then, there exists a parity condition $p: Q \times \Sigma \rightarrow \NN$ over $\A$ such that $\L(\A,p)=\L(\A,R)=\L(\A,S)$.
	Moreover, 
	\begin{itemize}
		\item if $r\leq s$, we can take $p$ to be a "$[1,2r+1]$-parity condition".
		\item if $s\leq r$, we can take $p$ to be a "$[0,2s]$-parity condition".
	\end{itemize}
\end{theorem}
\begin{proof}
	 
	 Proposition~\ref{Prop_EquivaleceDetAutomata_ImpliesEqCondition} implies that $(\A,R)"\simeq"(\A,S)$, and after corollary~\ref{Coro_Rabin+Streett=Parity_TransSyst}, there is a parity condition $p$ using the proclaimed priorities such that $(\A,p)"\simeq"(\A,R)$. Therefore $\L(\A,p)"\simeq"\L(\A,R)$ (since for both deterministic and non-deterministic $(\A,p)"\simeq"(\A,R)$ implies $\L(\A,p)"\simeq"\L(\A,R)$). 
\end{proof}

\begin{theorem}\label{Th_Buchi+co-B->Weak}
	Let $\A$ be the "transition system graph" of a deterministic automaton and $p$ and $p'$ be $[0,k]$ and "$[1,k+1]$-parity conditions" respectively over $\A$ such that $\L(\A,p)= \L(\A,p')$. Then, there exists a $"\mathit{Weak}_k"$ condition $W$ over $\A$ such that $\L(\A,W)=\L(\A,p)$.
	
	In particular, there is a $"\mathit{Weak}"$ condition $W$ over $\A$ such that $\L(\A,W)=L$ if and only if there are both "Büchi" and "co-Büchi" conditions $B,B'$ over $\A$ such that $\L(\A,B)= \L(\A,B')=L$.
	
	
\end{theorem}

\begin{proof}
	If follows from proposition~\ref{Prop_EquivaleceDetAutomata_ImpliesEqCondition} and corollary~\ref{Coro_Buchi+coB=Weak_TransSyst}.
\end{proof}
%
%

\newpage

%% file: conclusions.tex
We have presented a transformation that, given a Muller "transition system", provides an equivalent "parity" transition system that has minimal size and uses an optimal number of priorities among those which accept a "locally bijective morphism" to the original Muller transition system. In order to describe this transformation we have introduced the "alternating cycle decomposition", a data structure that arranges all the information about the acceptance condition of the transition system and the interplay between this condition and the structure of the system.

We have shown in Section~\ref{sec:applications} how the alternating cycle decomposition can be useful to reason about acceptance conditions, and we hope that this representation of the information will be helpful in future works.

We have not discussed the complexity of effectively computing the "alternating cycle de\-com\-po\-si\-tion" of a Muller transition system. It is known that solving Muller games is $\mathrm{PSPACE}$-complete when the acceptance condition is given as a list of accepting sets of colours
\cite{Dawar2005ComplexityBounds}. However, given a Muller game $\G$ and the "Zielonka tree" of its Muller condition, we have a transformation into a parity game of polynomial size on the size of $\G$, so solving Muller games with this extra information is in $\mathrm{NP}\cap \mathrm{co}$-$\mathrm{NP}$. Also, in order to build $\ACD(\T)$ we suppose that the Muller condition is expressed using as colours the set of edges of the game (that is, as an explicit Muller condition), and solving explicit Muller games is in $\mathrm{PTIME}$~\cite{Horn2008Explicit}. Consequently, unless $\mathrm{PSPACE}$ is contained in $\mathrm{NP}\cap \mathrm{co}$-$\mathrm{NP}$, we cannot compute the "Zielonka tree" of a Muller condition, nor the "alternating cycle decomposition" of a Muller transition system in polynomial time.